%% file: alpha-deletion-streams-Full-Version.tex
\documentclass[11pt]{article}

\usepackage[margin=1in]{geometry}
\usepackage{float}
\usepackage{amssymb,amsfonts,amsmath}
\usepackage{enumerate}
\usepackage{enumitem}
\usepackage{thmtools,thm-restate}
\usepackage{hyperref}
\usepackage{caption}
\usepackage{amsthm}
\newcommand{\R}{\mathbb{R}}                     
     
\newcommand{\poly}{\text{poly}}

\newcommand{\pr}[1]{\text{\normalfont Pr}\normalfont\lbrack #1 \rbrack} 
\newcommand{\ex}[1]{\mathbb{E}\normalfont\lbrack #1 \rbrack}
\newcommand{\bpr}[1]{\text{\normalfont Pr}\normalfont \Big[#1 \Big]} 
\newcommand{\bex}[1]{\mathbb{E}\normalfont \Big[#1 \Big]}

\newcommand{\eps}{\epsilon}

\newcommand{\ttx}[1]{\texttt{#1}}
\newcommand{\err}[2]{\text{Err}_{#2}^k(#1)}

\newcommand{\ab}{\allowbreak}

\newtheorem{theorem}{Theorem}
\newtheorem{lemma}{Lemma}
\theoremstyle{corollary}
\newtheorem{corollary}{Corollary}
\theoremstyle{fact}
\newtheorem{fact}{Fact}
\theoremstyle{definition}
\newtheorem{definition}{Definition}
\newtheorem{proposition}{Proposition}

\newtheorem{remark}{Remark}

\title{Data Streams with Bounded Deletions}
\author{
	Rajesh Jayaram\\
	Carnegie Mellon University\\
	\texttt{rkjayara@cs.cmu.edu}
	\and
	David P. Woodruff\\
	Carnegie Mellon University \\
	\texttt{dwoodruf@cs.cmu.edu}
}
\date{}

\begin{document}
	\maketitle

\input{Intro.tex}


\input{Fake_Count_Sketch.tex}


\input{HeavyHittersViaCountSketch.aux.tex}


\input{L1_sampling_NEW.tex}


\input{L1Estimation.tex}


\input{DistinctElementEstimation.tex}

\input{L0Sampling.tex}

\input{Lower_Bounds.tex}

\input{Conclusion.tex}

\bibliography{cluster}

\input{Appendix.tex}

\end{document}

%% file: Intro.tex
\begin{abstract}
Two prevalent models in the data stream literature are the insertion-only and turnstile models. Unfortunately, many important streaming problems require a $\Theta(\log(n))$ multiplicative factor more space for turnstile streams than for insertion-only streams. This complexity gap often arises because the underlying frequency vector $f$ is very close to $0$, after accounting for all insertions and deletions to items. Signal detection in such streams is difficult, given the large number of deletions. 

In this work, we propose an intermediate model which, given a parameter $\alpha \geq 1$, lower bounds the norm $\|f\|_p$ by a $1/\alpha$-fraction of the $L_p$ mass of the stream had all updates been positive. Here, for a vector $f$, $\|f\|_p = \left (\sum_{i=1}^n |f_i|^p \right )^{1/p}$, and the value of $p$ we choose depends on the application. This gives a fluid medium between insertion only streams (with $\alpha = 1$), and turnstile streams (with $\alpha = \poly(n)$), and allows for analysis in terms of $\alpha$. 

We show that for streams with this $\alpha$-property, for many fundamental streaming problems we can replace a $O(\log(n))$ factor in the space usage for algorithms in the turnstile model with a $O(\log(\alpha))$ factor. This is true for identifying heavy hitters, inner product estimation, $L_0$ estimation, $L_1$ estimation, $L_1$ sampling, and support sampling. For each problem, we give matching or nearly matching lower bounds for $\alpha$-property streams. We note that in practice, many important turnstile data streams are in fact $\alpha$-property streams for small values of $\alpha$. For such applications, our results represent significant improvements in efficiency for all the aforementioned problems. 	
\end{abstract}\newpage

\tableofcontents
\newpage
\section{Introduction}

Data streams have become increasingly important in modern applications, where the sheer size of a dataset imposes stringent restrictions on the resources available to an algorithm. Examples of such datasets include internet traffic logs, sensor networks, financial transaction data, database logs, and scientific data streams (such as huge experiments in particle physics, genomics, and astronomy). 
Given their prevalence, there is a substantial body of literature devoted to designing extremely efficient one-pass algorithms for important data stream problems. We refer the reader to \cite{babcock2002models, muthukrishnan2005data} for surveys of these algorithms and their applications.  

Formally, a data stream is given by an underlying vector $f \in \R^n$, called the \emph{frequency vector}, which is initialized to $0^n$. The frequency vector then receives a stream of $m$ updates of the form $(i_t,\Delta_t) \in [n] \times \{-M,\dots,M\}$ for some $M > 0$ and $t \in [m]$. The update $(i,\Delta)$ causes the change $f_{i_t} \leftarrow f_{i_t} + \Delta_t$. For simplicity, we make the common assumption that $\log(mM) = O(\log(n))$, though our results generalize naturally to arbitrary $n,m$ \cite{braverman2016beating}. 

 Two well-studied models in the data stream literature are the \emph{insertion-only} model and the \textit{turnstile} model. In the former model, it is required that $\Delta_t > 0$ for all $t \in [m]$, whereas in the latter $\Delta_t$ can be any integer in $\{-M,\dots,M\}$.   
 It is known that there are significant differences between these models. For instance, identifying an index $i \in [n]$ for which $|x_i| > \frac{1}{10} \sum_{j=1}^n |x_j|$ can be accomplished with only $O(\log(n))$ bits of space in the insertion-only model \cite{bhattacharyya2016optimal}, but requires $\Omega(\log^2(n))$ bits in the turnstile model \cite{Jowhari:2011}. This $\log(n)$ gap between the complexity in the two models occurs in many other important streaming problems.
 
    Due to this ``complexity gap'', it is perhaps surprising that no intermediate streaming model had been systematically studied in the literature before. For motivation on the usefulness of such a model, we note that nearly all of the lower bounds for turnstile streams involve inserting a large number of items before deleting nearly all of them \cite{kapralov2017optimal, Jowhari:2011, kane2010exact, jayram2013optimal}. This behavior seems unlikely in practice, as the resulting norm $\|f\|_p$ becomes arbitrarily small regardless of the size of the stream.  In this paper, we introduce a new model which avoids the lower bounds for turnstile streams by lower bounding the norm $\|f\|_p$. We remark that while similar notions of bounded deletion streams have been considered for their practical applicability \cite{garofalakis2016data}  (see also \cite{cormode2016sparse}, where a bound on the maximum number of edges that could be deleted in a graph stream was useful), to the best of our knowledge there is no comprehensive theoretical study of data stream algorithms in this setting.

Formally, we define the \emph{insertion vector} $I \in \R^n$ to be the frequency vector of the substream of positive updates ($\Delta_t \geq 0$), and the \emph{deletion vector} $D \in \R^n$ to be the entry-wise absolute value of the frequency vector of the substream of negative updates. Then $f = I - D$ by definition. Our model is as follows.

\begin{definition}\label{def:alphaprop}
	For $\alpha \geq 1$ and $p \geq 0$, a data stream $f$ satisfies the \textit{$L_p$ $\alpha$-property} if $\|I + D\|_p \leq \alpha \|f\|_p$.
\end{definition}

 For $p=1$, the definition simply asserts that the final $L_1$ norm of $f$ must be no less than a $1/\alpha$ fraction of the total weight of updates in the stream $\sum_{t=1}^m |\Delta_t|$. For strict turnstile streams, this is equivalent to the number of deletions being less than a $(1 - 1/\alpha)$ fraction of the number of insertions, hence a bounded deletion stream.

For $p=0$, the $\alpha$-property simply states that $\|f\|_0$, the number of non-zero coordinates at the end of the stream, is no smaller than a $1/\alpha$ fraction of the number of distinct elements seen in the stream (known as the $F_0$ of the stream). Importantly, note that for both cases this constraint need only hold at the time of query, and not necessarily at every point in the stream.

Observe for $\alpha =1$, we recover the insertion-only model, whereas for $\alpha = mM$ in the $L_1$ case or $\alpha = n$ in the $L_0$ case we recover the turnstile model (with the minor exception of streams with $\|f\|_p = 0$). 
 So $\alpha$-property streams are a natural parameterized intermediate model between the insertion-only and turnstile models. For clarity, we use the term \emph{unbounded deletion stream} to refer to a (general or strict) turnstile stream which does not satisfy the $\alpha$ property for $\alpha = o(n)$.
  
For many applications of turnstile data streaming algorithms, the streams in question are in fact $\alpha$-property streams for small values of $\alpha$.  For instance, in network traffic monitoring it is useful to estimate differences between network traffic patterns across distinct time intervals or routers \cite{muthukrishnan2005data}. If $f_i^1,f_i^2$ represent the number of packets sent between the $i$-th [source, destination] IP address pair in the first and second intervals (or routers), then the stream in question is $f^1 - f^2$. In realistic systems, the traffic behavior will not be identical across days or routers, and even differences as small as $0.1\%$ in overall traffic behavior (i.e. $\|f^1 - f^2\|_1 > .001\|f^1 + f^2\|_1$) will result in $\alpha < 1000$ (which is significantly smaller than the theoretical universe size of $n \approx 2^{256}$ potential IP addresses pairs in IPv6).

 A similar case for small $\alpha$ can be made for differences between streams whenever these differences are not arbitrarily small. This includes applications in streams of financial transactions, sensor network data, and telecommunication call records \cite{indyk2004algorithms, cormode2004holistic}, as well as for identifying newly trending search terms, detecting DDoS attacks, and estimating the spread of malicious worms \cite{pike2005interpreting, estan2003bitmap, moore:codered, lakhina2005mining, wagner2005entropy}.

\begin{figure}[]
	\begin{center}
		\begin{tabular}{| c | c| c| c | c |} 
			\hline
			Problem & Turnstile L.B.  &  $\alpha$-Property U.B.  & Citation& Notes\\
			\hline
			$\eps$-Heavy Hitters   &   $\Omega(\eps^{-1} \log^2(n))$ &	$O(\eps^{-1} \log(n)\log(\alpha ))$ &\cite{Jowhari:2011} &
			\begin{tabular}{c}
				Strict-turnstile \\
				succeeds w.h.p.
			\end{tabular} \\
			\hline
			$\eps$-Heavy Hitters   &   $\Omega(\eps^{-1} \log^2(n))$ & 
			$O(\eps^{-1} \log(n)\log(\alpha ))$
			&\cite{Jowhari:2011} & 
			\begin{tabular}{c}
				General-turnstile \\
			$\delta = O(1)$ 
			\end{tabular} \\	\hline
			Inner Product  & $\Omega(\eps^{-1}\log(n))$ & $O(\eps^{-1}\log(\alpha))$ & Theorem \ref{thm:InnerProdHard} & General-turnstile \\ \hline
			$L_1$ Estimation & $\Omega(\log(n))$   & $O(\log(
			\alpha))$ &	Theorem \ref{thm:L1esthardstrict}	 & Strict-turnstile \\ 
			\hline
			$L_1$ Estimation & $\Omega(\eps^{-2}\log(n))$   & 
			\begin{tabular}{c}$O(\eps^{-2}\log(
				\alpha))$\\ $+ \log(n))$
			\end{tabular} &	\cite{kane2010exact}	 & General-turnstile \\ 
			\hline
			
			$L_0$ Estimation & $\Omega(\eps^{-2}\log(n))$ & \begin{tabular}{c}$O(\eps^{-2}\log(
				\alpha))$\\ $+ \log(n))$
			\end{tabular} & \cite{kane2010exact} &General-turnstile \\ \hline

			$L_1$ Sampling & $\Omega(\log^2(n))$& $O(\log(n)\log(\alpha))$& \cite{Jowhari:2011} &
		
				General-turnstile $(*)$ 
			
			\\ \hline
			
			Support Sampling & $\Omega(\log^2(n))$& $O(\log(n)\log(\alpha))$& \cite{kapralov2017optimal} &
			Strict-turnstile
			\\ \hline
		\end{tabular} 
	\end{center} \caption{The best known lower bounds (L.B.) for classic data stream problems in the turnstile model, along with the upper bounds (U.B.) for $\alpha$-property streams from this paper. The notes specify whether an U.B./L.B. pair applies to the strict or general turnstile model. For simplicity, we have suppressed $\log\log(n)$ and $\log(1/\eps)$ terms, and all results are for $\delta = O(1)$ failure probability, unless otherwise stated. $(*)$ $L_1$ sampling note: strong $\alpha$-property, with $\eps = \Theta(1)$ for both U.B. \& L.B.} \label{fig:results}
\end{figure} 

A setting in which $\alpha$ is likely even smaller is database analytics. For instance, an important tool for database synchronization is Remote Differential Compression \ab(RDC)\ab \cite{teodosiu2006optimizing, ajtai2002system}, which allows similar files to be compared between a client and server by transmitting only the differences between them. For files given by large data streams, one can feed these differences back into sketches of the file to complete the synchronization. Even if as much as a half of the file must be resynchronized between client and sever (an inordinately large fraction for typical RDC applications), streaming algorithms with $\alpha = 2$ would suffice to recover the data. 

For $L_0$, there are important applications of streams with bounded ratio $\alpha =F_0/L_0$. For example, $L_0$ estimation is applicable to networks of cheap moving sensors, such as those monitoring wildlife movement or water-flow patterns \cite{indyk2004algorithms}. In such networks, some degree of clustering is expected (in regions of abundant food, points of water-flow accumulation), and these clusters will be consistently occupied by sensors, resulting in a bounded ratio of inactive to active regions. Furthermore, in networking one often needs to estimate the number of distinct IP addresses with active network connections at a given time \cite{GarofalakisData2016, muthukrishnan2005data}. Here we also observe some degree of clustering on high-activity IP's, with persistently active IP's likely resulting in an $L_0$ to $F_0$ ratio much larger than $1/n$ (where $n$ is the universe size of IP addresses).
 
 In many of the above applications, $\alpha$ can even be regarded as a constant when compared with $n$. For such applications, the space improvements detailed in Figure \ref{fig:results} are considerable, and reduce the space complexity of the problems nearly or exactly to known upper bounds for insertion-only streams \cite{bhattacharyya2016optimal, morris1978counting, kane2010optimal, Jowhari:2011}. 

Finally, we remark that in some cases it is not unreasonable to assume that the magnitude of \textit{every} coordinate would be bounded by some fraction of the updates to it. For instance, in the case of RDC it is seems likely that none of the files would be totally removed. We summarize this stronger guarantee as the \textit{strong $\alpha$-property}.  
 
 \begin{definition}\label{def:strongalphaprop}
 	For $\alpha \geq 1$, a data stream $f$ satisfies the \textit{strong $\alpha$-property} if $I_i + D_i \leq \alpha |f_i|$ for all $i \in [n]$.
 \end{definition}
Note that this property is strictly stronger that the $L_p$ $\alpha$-property for any $p\geq0$. In particular, it forces $f_i \neq 0$ if $i$ is updated in the stream. In this paper, however, we focus primarily on the more general $\alpha$-property of Definition \ref{def:alphaprop}, and use $\alpha$-property to refer to Definition \ref{def:alphaprop} unless otherwise explicitly stated. Nevertheless, we show that our lower bounds for $L_p$ heavy hitters, $L_1$ estimation, $L_1$ sampling, and inner product estimation, all hold even for the more restricted strong $\alpha$-property streams.

\subsection{Our Contributions}
We show that for many well-studied streaming problems, we can replace a $\log(n)$ factor in algorithms for general turnstile streams with a $\log(\alpha)$ factor for $\alpha$-property streams. This is a significant improvement for small values of $\alpha$. Our upper bound results, along with the lower bounds for the unbounded deletion case, are given in Figure \ref{fig:results}. Several of our results come from the introduction of a new data structure, \texttt{CSSampSim} (Section \ref{sec:fakecs}), which produces point queries for the frequencies $f_i$ with small additive error. Our improvements from \texttt{CSSampSim} and other $L_1$ problems are the result of new sampling techniques for $\alpha$-property streams 

While sampling of streams has been studied in many papers, most have been in the context of insertion only streams (see, e.g., \cite{cohen2015stream,cohen2009stream, DBLP:journals/jcss/CohenDKLT14,DBLP:journals/pvldb/CohenCD11, DBLP:journals/tocs/EstanV03,DBLP:conf/sigmod/GibbonsM98,DBLP:books/lib/Knuth98,DBLP:journals/pvldb/MankuM12,DBLP:journals/toms/Vitter85}). Notable examples of the use of sampling in the presence of deletions in a stream are \cite{cohen2012don, DBLP:journals/vldb/GemullaLH08, DBLP:books/sp/16/Haas16, DBLP:conf/vldb/GemullaLH06, DBLP:conf/pods/GemullaLH07}. We note that these works are concerned with unbiased estimators and do not provide the $(1\pm \epsilon)$-approximate relative error guarantees with small space that we obtain. They are also concerned with unbounded deletion streams, whereas our algorithms exploit the $\alpha$-property of the underlying stream to obtain considerable savings.   

 In addition to upper bounds, we give matching or nearly matching lower bounds in the $\alpha$-property setting for all the problems we consider. In particular, for the $L_1$-related problems (heavy hitters, $L_1$ estimation, $L_1$ sampling, and inner product estimation), we show that these lower bounds hold even for the stricter case of strong $\alpha$-property streams.
 
We also demonstrate that for general turnstile streams, obtaining a constant approximation of the $L_1$ still requires $\Omega(\log(n))$-bits for $\alpha$-property streams. For streams with unbounded deletions, there is an $\Omega(\epsilon^{-2} \log(n))$ lower bound for $(1\pm\epsilon)$-approximation \cite{kane2010exact}. Although we cannot remove this $\log n$ factor for $\alpha$-property streams, we are able to show an upper bound of $\tilde{O}(\epsilon^{-2} \log \alpha + \log n)$ bits of space for strong $\alpha$-property streams, where the $\tilde{O}$ notation hides $\log(1/\epsilon) + \log \log n$ factors. We thus separate the dependence of $\epsilon^{-2}$ and $\log n$ in the space complexity, illustrating an additional benefit of $\alpha$-property streams, and show a matching lower bound for strong $\alpha$-property streams. 
  
\subsection{Our Techniques}
Our results for $L_1$ streaming problems in Sections \ref{sec:fakecs} to \ref{sec:L1Est} are built off of the observation that for $\alpha$ property streams the number of insertions and deletions made to a given $i \in [n]$ are both upper bounded by $\alpha \|f\|_1$. We can then think of sampling updates to $i$ as a biased coin flip with bias at least $1/2 + (\alpha\|f\|_1)^{-1}$. By sampling $\poly(\alpha/\eps)$ stream updates and scaling up, we can recover the difference between the number of insertions and deletions up to an additive $\eps \|f\|_1$ error.  

To exploit this fact, in Section \ref{sec:fakecs} we introduce a data structure \texttt{CSSampSim} inspired by the classic Countsketch of \cite{charikar2002finding}, which simulates running each row of Countsketch on a small uniform sample of stream updates. Our data structure does not correspond to a valid instantiation
of Countsketch on any stream since we sample different stream updates for different rows of Countsketch.
Nevertheless, we show via a Bernstein inequality that our data structure obtains the Countsketch guarantee plus an $\eps\|f\|_1$ additive error, with only a logarithmic dependence on $\eps$ in the space. This results in more efficient algorithms for the $L_1$ heavy hitters problem (Section \ref{sec:heavyhitters}), and is also used in our $L_1$ sampling algorithm (Section \ref{sec:L1Samp}). We are able to argue
that the counters used in our algorithms can be represented with much fewer than $\log n$ bits because
we sample a very small number of stream updates. 

 Additionally, we demonstrate that sampling $\poly(\alpha/\eps)$ updates preserves the inner product between $\alpha$-property streams $f,g$ up to an additive $\eps \|f\|_1\|g\|_1$ error. Then by hashing the sampled universe down modulo a sufficiently large prime, we show that the inner product remains preserved, allowing us to estimate it in $O(\eps^{-1}\log(\alpha))$ space (Section \ref{sec:innerprod}).

Our algorithm for $L_1$ estimation (Section \ref{sec:L1Est}) utilizes our biased coin observation to show that sampling will recover the $L_1$ of a strict turnstile $\alpha$-property stream. To carry out the sampling in $o(\log(n))$ space, give a alternate analysis of the well known Morris counting algorithm, giving better space but worse error bounds. This allows us to obtain a rough estimate of the position in the stream so that elements can be sampled with the correct probability. 
For $L_1$ estimation in general turnstile streams, we 
analyze a virtual stream which corresponds to scaling our
input stream by Cauchy random variables, argue it still has the $\alpha$-property, and apply our sampling
analysis for $L_1$ estimation on it. 

Our results for the $L_0$ streaming problems in Sections \ref{sec:L0Est} and \ref{sec:suppsamp} mainly exploit the $\alpha$-property in sub-sampling algorithms. Namely, many data structure for $L_0$ streaming problems subsample the universe $[n]$ at $\log(n)$ levels, corresponding to $\log(n)$ possible thresholds which could be $O(1)$-approximations of the $L_0$. If, however, an $O(1)$ approximation were known in advance, we could immediately subsample to this level and remove the $\log(n)$ factor from the space bound.  
For $\alpha$ property streams, we note that the non-decreasing values $F_0^t$ of $F_0$ after $t$ updates must be bounded in the interval $[L_0^t, O(\alpha)L_0]$. Thus, by employing an $O(1)$ estimator $R^t$ of $F_0^t$, we show that it suffices to subsample to only the $O(\log(\alpha/\eps))$ levels which are closest to $\log(R^t)$ at time $t$, from which our space improvements follows. 

\subsection{Preliminaries}
If $g$ is any function of the updates of the stream, for $t \in [m]$ we write $g^t$ to denote the value of $g$ after the updates $(i_1,\Delta_1),\dots,(i_t,\Delta_t)$. 
  For Sections \ref{sec:fakecs} to \ref{sec:L1Est}, it will suffice to assume $\Delta_t \in \{-1,1\}$ for all $t \in [m]$. For general updates, we can implicitly consider them to be several consecutive updates in $\{-1,1\}$, and our analysis will hold in this expanded stream. This implicit expanding of updates is only necessary for our algorithms which sample updates with some probability $p$. If an update $|\Delta_t| > 1$ arrives, we update our data structures with the value $\text{sign}(\Delta_t)\cdot\text{Bin}(|\Delta_t|,p)$, where Bin$(n,p)$ is the binomial random variable on $n$ trials with success probability $p$, which has the same effect. In this unit-update setting, the $L_1$ $\alpha$ property reduces to $m \leq \alpha \|f^m\|_1$, so this is the definition which we will use for the $L_1$ $\alpha$ property. We use the term \textit{unbounded deletion stream} to refer to streams without the $\alpha$-property (or equivalently streams that have the $\alpha = \poly(n)$-property and can also have all $0$ frequencies).  For Sections \ref{sec:fakecs} to \ref{sec:L1Est}, we will consider only the $L_1$ $\alpha$-property, and thus drop the $L_1$ in these sections for simplicity, and for Sections \ref{sec:L0Est} and \ref{sec:suppsamp} we will consider only the $L_0$ $\alpha$-property, and similarly drop the $L_0$ there.
 
We call a vector $y \in \R^n$ $k$\textit{-sparse} if it has at most $k$ non-zero entries. For a given vector $f \in \R^n$, let $\err{f}{p} = \min_{y\; k-\text{sparse}} \|f - y\|_p$. In other words, $\err{f}{p}$ is the $p$-norm of $f$ with the $k$ heaviest entries removed. We call the argument minimizer of $ \|f - y\|_p$ the best $k$-sparse approximation to $f$. 
Finally, we use the term \textit{with high probability} (w.h.p.) to describe events that occur with probability $1-n^{-c}$, where $c$ is a constant. Events occur with low probability (w.l.p.) if they are a complement to a w.h.p. event. We will often ignore the precise constant $c$ when it only factors into our memory usage as a constant.

%% file: Fake_Count_Sketch.tex
\section{Frequency Estimation via Sampling}
\label{sec:fakecs}
In this section, we will develop many of the tools needed for answering approximate queries about $\alpha$ property streams. Primarily, we develop a data structure \texttt{CSSampSim}, inspired by the classic Countsketch of \cite{charikar2002finding}, that computes frequency estimates of items in the stream by sampling. This data structure will immediately result in an improved heavy hitters algorithm in Section \ref{sec:heavyhitters}, and is at the heart of our $L_1$ sampler in Section \ref{sec:L1Samp}. In this section, we will write $\alpha$-property to refer to the $L_1$ $\alpha$-property throughout.

Firstly, for $\alpha$-property streams, the following observation is crucial to many of our results. Given a fixed item $i \in [n]$, by sampling at least $\poly(\alpha/\eps)$ stream updates we can preserve $f_i$ (after scaling) up to an additive $\eps \|f\|_1$ error.

\begin{lemma}[Sampling Lemma]
	\label{lem:l1-preservation}
	Let $f$ be the frequency vector of a general turnstile stream with the $\alpha$-property, and let $f^*$ be the frequency vector of a uniformly sampled substream scaled up by $\frac{1}{p}$, where each update is sampled uniformly with probability $p > \alpha^2 \eps^{-3}\log(\delta^{-1})/m$. Then with probability at least $1-\delta$ for $ i \in [n]$,  we have \[| f_i^*   -  f_i| < \eps \|f\|_1\]
Moreover, we have $\sum_{i=1}^n f_i^* =\sum_{i=1}^n f_i  \pm \eps \|f\|_1$.
\end{lemma}
\begin{proof}
Assume we sample each update to the stream independently with probability $p$. Let $f_i^+,f_i^-$ be the number of insertions and deletions of element $i$ respectively, so $f_i = f_i^+ - f_i^-$.  Let $X_j^+$ indicate that the $j$-th insertion to item $i$ is sampled. First, if $\eps \|f\|_1  <  f^+_i$ then by Chernoff bounds:
 \[\pr{ | \frac{1}{p}\sum_{j=1}^{f_i^+}\ab X_j^+ - f_i^+ | \geq  \eps \|f\|_1} \leq 2\exp\big(\frac{
 -p f_i^+(\eps \|f\|_1 )^2}{3(f_i^+)^2}\big)\]
  \[\ab\leq \exp\big(-p \eps^3 m/\alpha^2		\big)\] 
  where the last inequality holds because $f^+_i \leq m \leq \alpha \|f\|_1$. Taking $p \geq \alpha^2 \log(1/\delta)/(\eps^3 m)$ gives the desired probability $\delta$. 
 Now if $\eps \|f\|_1  \geq f^+_i$, then $\pr{ \frac{1}{p}\sum_{j=1}^{f_i^+}\ab X_j^+ \geq f_i^+ + \eps \|f\|_1} \leq \ab\exp \ab\big(-p f_i^+ \eps \|f\|_1\ab/\ab f_i^+ \big)$ $\leq \exp \big(-p \eps m/\alpha	\big) \leq \delta$ for the same value of $p$. Applying the same argument to $f_i^-$, we obtain $f_i^* = f_i^+ - f_i^- \pm 2\eps\|f\|_1 = f_i \pm 2 \eps \|f\|_1$ as needed after rescaling $\eps$.  
For the final statement, we can consider all updates to the stream as being made to a single element $i$, and then simply apply the same argument given above. 
\end{proof}
\subsection{Count-Sketch Sampling Simulator}

\begin{figure*}
	\fbox{\parbox{\textwidth}{ \texttt{CSSampSim (CSSS)}
			\\
			\textbf{Input:} sensitivity parameters $k \geq  1, \; \eps \in (0,1)$.
			\begin{enumerate}[topsep=0pt,itemsep=-1ex,partopsep=1ex,parsep=1ex] 
				\item Set $S =  \Theta( \frac{\alpha^2}{\eps^2 }T^2 \log(n))$, where $T \geq 4/\eps^2 + \log(n)$.	
				\item Instantiate $d \times 6k$ count-sketch table $A$, for $d = O(\log(n))$. For each table entry $a_{ij} \in A$, store two values $a_{ij}^+$ and $a_{ij}^-$, both initialized to $0$. 
				\item Select $4$-wise independent hash functions $h_i:[n] \to [6k],\: g_i:[n] \to \{1,-1\}$, for $i \in [d]$. 
				\item Set $p \leftarrow 0$, and start $\log(n)$ bit counter to store the position in the stream.
				\item \textbf{On Update $(i_t,\Delta_t)$:}
				\begin{enumerate}[topsep=0pt,itemsep=-1ex,partopsep=1ex,parsep=1ex]
					\item  \textbf{if} $t= 2^r\log(S)+1$ for any $r \geq 1$, then for every entry $a_{ij} \in A$ set $a_{ij}^+ \leftarrow$Bin$(a_{ij}^+,1/2)$,  $a_{ij}^- \leftarrow$Bin$(a_{ij}^-,1/2)$, and $p \leftarrow p + 1$

					\item \textbf{On Update $(i_t,\Delta_t)$:} Sample $(i_t,\Delta_t)$ with probability $2^{-p}$. If sampled, then for $i \in [d]$
					\begin{enumerate}[topsep=0pt,itemsep=-1ex,partopsep=1ex,parsep=1ex]
						\item \textbf{if} $\Delta_tg_i(i_t)>0$, set $a_{i,h_i(i_t)}^+ \leftarrow a_{i h_i(i_t)}^+ + \Delta_tg_i(i_t)$.
						\item \textbf{else} set $a_{i,h_i(i_t)}^- \leftarrow a_{i h_i(i_t)}^- + | \Delta_tg_i(i_t)|$
					\end{enumerate}
					
				\end{enumerate}
				\item \textbf{On Query for $f_j$:} return $y^*_j = \text{median}\{ 2^p g_i(j)\cdot (a_{i,h_i(j)}^+ - a_{i,h_i(j)}^-) \; | \; i \in [d] \} $.
			\end{enumerate}
	}}\caption{Our data structure to simulate running Countsketch on a uniform sample of the stream.} \label{fig:csss}
\end{figure*}
We now describe the Countsketch algorithm of \cite{charikar2002finding}, which is a simple yet classic algorithm in the data stream literature. For a parameter $k$ and $d = O(\log(n))$, it creates a $d \times 6k$ matrix $A$ initialized to $0$, and for every row $i \in [d]$ it selects hash functions $h_i : [n] \to [6k]$, and $g_i: [n] \to \{1,-1\}$ from $4$-wise independent uniform hash families. On every update $(i_t, \Delta_t)$, it hashes this update into every row $i \in [d]$ by updating $a_{i,h_i(i_t)} \leftarrow a_{i,h_i(i_t)} + g_i(i_t)\Delta_t$. It then outputs $y^*$ where $y^*_j = \text{median}\{ g_i(j)a_{i,h_i(j)} \; | \; i \in [d]	\}$. The guarantee of one row of the Countsketch is as follows.

\begin{lemma}
	\label{lem:cs1row}
	Let $f \in \R^n$ be the frequency vector of any general turnstile stream hashed into a $d \times 6k$ Countsketch table. Then with probability at least $2/3$, for a given $j \in [n]$ and row $i \in [d]$ we have $\big|g_i(j)a_{i,h_i(j)} - f_j\big| < k^{-1/2}\err{f}{2}$. It follows if $d = O(\log(n))$, then with high probability, for all $j \in [n]$ we have $\big|y_j^* - f_j\big| < k^{-1/2}\err{f}{2}$, where $y^*$ is the estimate of Countsketch. The space required is $O(k \log^2(n))$ bits.
\end{lemma}

We now introduce a data structure which simulates running Countsketch on a uniform sample of the stream of size $\poly(\alpha \log(n)/\eps)$. The full data structure is given in Figure \ref{fig:csss}. 
Note that for a fixed row, each update is chosen with probability at least $2^{-p} \geq S/(2m) = \Omega (\alpha^2 T^2 \log(n) / (\eps^2\ab m))$. We will use this fact to apply Lemma \ref{lem:l1-preservation} with $\eps' = (\eps/T)$ and $\delta = 1/\poly(n)$. The parameter $T$ will be $\poly(\log(n)/\eps)$, and we introduce it as a new symbol purely for clarity. 

 Now the updates to each row in \texttt{CSSS} are sampled independently from the other rows, thus \texttt{CSSS} may not represent running Countsketch on a single valid sample. However, each row \textit{independently} contains the result of running a row of Countsketch on a valid sample of the stream. Since the Countsketch guarantee holds with probability $2/3$ for each row, and we simply take the median of $O(\log(n))$ rows to obtain high probability, it follows that the output of \texttt{CSSS} will still satisfy an additive error bound w.h.p. if each row also satisfies that error bound with probability $2/3$.

By Lemma \ref{lem:l1-preservation} with sensitivity parameter $(\eps/T)$, we know that we preserve the weight of all items $f_i$ with $|f_i| \geq 1/T \|f\|_1$ up to a $(1 \pm \eps)$ factor w.h.p. after sampling and rescaling. For all smaller elements, however, we obtain error additive in $\eps \|f\|_1/T$.
 This gives rise to the natural division of the coordinates of $f$. Let $big \subset [n]$ be the set of $i$ with  $|f_i| \geq 1/T \|f\|_1$, and let $small \subset [n]$ be the complement. Let $\mathcal{E} \subset [n]$ be the top $k$ heaviest coordinates in $f$, and let $s \in \R^n$ be a fixed sample vector of $f$ after rescaling by $p^{-1}$. Since $\err{s}{2} = \min_{\hat{y} \text{ k-sparse}} \|s - \hat{y}\|_2$, it follows that $\err{s}{2} \leq \|s_{big \setminus \mathcal{E}} \|_2 + \|s_{small}\|_2 $. Furthermore, by Lemma \ref{lem:l1-preservation} $\|s_{big \setminus \mathcal{E}}\|_2 \leq(1+ \eps) \|f_{big \setminus \mathcal{E}}\|_2 \leq (1+  \eps) \err{f}{2}$. So it remains to upper bound $\|s_{small}\|_2^2$, which we do in the following technical lemma. 

The intuition for the Lemma is that $\|f_{small}\|_2$ is maximized when all the $L_1$ weight is concentrated in $T$ elements, thus $\|f_{small}\|_2 \leq (T(\|f\|_1/T)^2)^{1/2} =\|f\|_1/T^{1/2}$. By the $\alpha$ property, we know that the number of insertions made to the elements of $small$ is $\ab$bounded by $\alpha \|f\|_1$. Thus, computing the variance of $\ab\|s_{small}\|_2^2$ and applying Bernstein's inequality, we obtain a similar upper bound for $\|s_{small}\|_2$. 

\begin{lemma}
	\label{lem:smallbound}
			If $s$ is the rescaled frequency vector resulting from uniformly sampling with probability $p \geq S/(2m)$, where $S = \Omega(\alpha^2T^2\; \ab\log(n))$ for $T= \Omega(\log(n))$, of a general turnstile stream $f$ with the $\alpha$ property, then we have $\|s_{small}\|_2 \leq 2 T^{-1/2}\|f\|_1$ with high probability.
\end{lemma}
\begin{proof}
	
	Fix  $f \in \R^n$, and assume that $\|f\|_1 > S$ (otherwise we could just sample all the updates and our counters in Countsketch would never exceed $\log(\alpha S)$). For any $i \in [n]$, let $f_i = f_i^+ - f_i^-$, so that $f = f^+ - f^-$, and let $s_i$ be our rescaled sample of $f_i$. By definition, for each $i \in small$ we have $f_i < \frac{1}{T}\|f\|_1$. Thus the quantity $\|f_{small}\|_2^2$ is maximized by having $T$ coordinates equal to $\frac{1}{T}\|f\|_1$. Thus $\|f_{small}\|_2^2 \leq T (\frac{\|f\|_1}{T})^2 \leq \frac{\|f\|_1^2}{T}$.  Now note that if we condition on the fact that $s_i \leq 2/T \|f\|$ for all $i \in small$, which occurs with probability greater than $1 - n^{-5}$ by Lemma \ref{lem:l1-preservation}, then since  $\ex{\sum_i |s_i|^4} = O(n^4)$, all of the following expectations change by a factor of at most $(1 \pm 1/n)$ by conditioning on this fact. Thus we can safely ignore this conditioning and proceed by analyzing $\ex{\|s_{small}\|_2^2} $ and $\ex{\|s_{small}\|_4^4}$ without this condition, but use the conditioning when applying Bernstein's inequality later.
	
	We have that $|s_i| = |\frac{1}{p}\sum_{j=1}^{f_i^+ + f_i^-} X_{ij}|$, where  $X_{ij}$ is an indicator random variable that is $\pm1$ if the $j$-th update to $f_i$ is sampled, and $0$ otherwise, where the sign depends on whether or not the update was an insertion or deletion and $p \geq S/(2m)$ is the probability of sampling an update. Then $\ex{|s_i|} = \ex{ | \frac{1}{p}\sum_{j=1}^{f_i^+ + f_i^-} X_{ij}|} = |f_i|$. Furthermore, we have $\ex{s_i^2} = \frac{1}{p^2}\ex{(\sum_{j=1}^{f_i^+ + f_i^-} X_{ij})^2} $ \[
	= \frac{1}{p^2}(\sum_{j=1}^{f_i^+ + f_i^-}\ex{ X_{ij}^2} + \sum_{j_1 \neq j_2}\ex{X_{ij_1}X_{ij_2}}) \]
	\[ = \frac{1}{p}(f_i^+ + f_i^-) + ((f_i^+)(f_i^+ - 1) + (f_i^-)(f_i^- - 1)  - 2f_i^+ f_i^- )		\]
	Substituting $f_i^- = f_i^+ - f_i$ in part of the above equation  gives  $\ex{s_i^2} = \frac{1}{p}(f_i^+ + f_i^-) + f_i^2 + f_i - 2f_i^+ \leq\frac{1}{p}(f_i^+ + f_i^-) + 2f_i^2$.
	So $\ex{\sum_{i \in small} s_i^2} \leq  \frac{1}{p}(\|f^+_{small}\|_1 + \|f^-_{small}\|_1) + 2\|f_{small}\|_2^2$, which is at most $  \frac{\alpha}{p}\|f\|_1 + 2\frac{\|f\|_1^2}{T}$ by the $\alpha$-property of $f$ and the upper bound on $\|f_{small}\|_2^2$. Now $\ex{s_i^4} = \frac{1}{p^4}\ex{(\sum_{j=1}^{f_i^+ + f_i^-} \; \ab X_{ij})^4 }$, which we can write as
	\begin{equation} \label{eqn:1}
	\bex{  \sum_j X_{ij}^4 + 4\sum_{j_1 \neq j_2} X_{ij_1}^3 X_{ij_2}	 +12 \sum_{\substack{j_1,j_2,j_3 \\ \text{distinct}}} X_{ij_1}^2 X_{ij_2}X_{ij_3} + \sum_{\substack{j_1,j_2,j_3,j_4 \\ \text{distinct}}} X_{ij_1} X_{ij_2}X_{ij_3}X_{ij_4}	}\frac{1}{p^4}
	\end{equation}
	We analyze Equation \ref{eqn:1} term by term. First note that $\ex{ \sum_j \; \ab X_{ij}^4} = 	p (f_i^+ + f_i^-)$, and $\ex{ \sum_{j_1 \neq j_2} X_{ij_1}^3 X_{ij_2}	} =  p^2 \big((f_i^+)(f_i^+ - 1) + (f_i^-)(f_i^- - 1)- 2(f_i^+ f_i^-)\big)$. Substituting $f_i^- = f_i^+ - f_i$, we obtain $\ex{ \sum_{j_1 \neq j_2} X_{ij_1}^3 X_{ij_2}	} \leq 2p^2f_i^2 $. Now for the third term we have $\ex{ \sum_{j_1 \neq j_2 \neq j_3} X_{ij_1}^2\ab X_{ij_2}X_{ij_3}  } 
	= p^3 \big((f_i^+)(f_i^+ -1)(f_i^+ -2) + (f_i^-)(f_i^- - 1)(f_i^- - 2) + f_i^+ (f_i^-)(f_i^- - 1) + f_i^-(f_i^+)(f_i^+ - 1)  -2(f_i^+f_i^-)(f_i^+ - 1) - 2(f_i^+f_i^-)(f_i^- - 1) \big)$, which after the same substitution is upper bounded by $10p^3 \max\{f_i^+,|f_i|\} f_i^2 \leq 10 p^3 \alpha \ab \|f\|_1 f_i^2$, where the last inequality follows from the $\alpha$-property of $f$. Finally, the last term is $\ex{\sum_{j_1 \neq j_2 \neq j_3 \neq j_4}\allowbreak X_{ij_1} X_{ij_2}X_{ij_3}X_{ij_4} } =  p^4\big(f_i^+(f_i^+ - 1)(f_i^+ - 2)(f_i^+ - 3) + f_i^-(f_i^- - 1)(f_i^- - 2)(f_i^- - 3)  + 6(f_i^+(f_i^+ - 1))(f_i^-(f_i^- - 1))- 4 \big(f_i^+ (f_i^+ - 1) (f_i^+ - 2)f_i^- +f_i^-(f_i^- - 1)(f_i^- - 2)f_i^+	\big)	\big) $. Making the same substitution allows us to bound this above by $p^4 (24 f_i^4 - 12 f_i f_i^+ + 12(f_i^+)^2) \leq p^4(36 (f_i)^4 + 24(f_i^+)^2)$.
	
	Now $f$ satisfies the $\alpha$ property. Thus $\|f_{small}^+\|_1+|f_{small}^-\|_1\allowbreak \leq \alpha \ab\|f\|_1$, so summing the bounds from the last paragraph over all $i \in small$ we obtain 
	\[\ex{ \frac{1}{p^4} \|s_{small}\|_4^4} \leq  \frac{1}{p^3} \alpha \|f\|_1 + \frac{8}{p^2} \|f_{small}\|_2^2  \] \[+\frac{120\alpha}{p}\|f_{small}\|_2^2 \
	|f\|_1 + 36 \|f_{small}\|_4^4 + 24 \alpha^2 \|f\|_1^2 \]
	Now  $1/p \leq \frac{2m}{S} \leq \frac{2\alpha}{S}\|f\|_1$ by the $\alpha$-property. Applying this with the fact that $\|f_{small}\|_2^2$ is maximized at $ \frac{\|f\|_1^2}{T}$, and similarly $\|f_{small}\|_2^2$ is maximized at $T(\|f\|_1/T)^4 = \|f\|_1^4/T^3$ we have that
	$	\ex{ \frac{1}{p^4} \|s_{small}\|_4^4} $ is at most\[   \frac{8\alpha^4}{S^3} \|f\|_1^4 + \frac{36\alpha^2}{S^2T} \|f\|_1^4 + \frac{240\alpha^2}{ST} \|f\|_1^4 +  \frac{36}{T^3} \|f\|_1^4 + 24 \alpha^2 \|f\|_1^2 \]
	Since we have $\|f\|_1 > S > \alpha^2 T^2$, the above expectation is upper bounded by $ \frac{300}{T^3}\|f\|_1^4$.
	We now apply Bernstein's inequality. Using the fact that we conditioned on earlier, we have the upper bound $\frac{1}{p}s_i \leq 2 \|f\|_1/T$, so plugging this into Bernstein's inequality yields:
	$\pr{ \big| \|s_{small}\|_2^2 - \ex{\|s_{small}\|_2^2 } \big| >  \frac{\|f\|_1^2}{T} } \leq \exp\big(-	\frac{\|f\|_1^4/(2T^2)  }{300\|f\|_1^4/T^3  +  2\|f\|_1^3/(3T^2)}	\big) \leq \exp\big(-T/604		\big)	$
	Finally,  $T= \Omega(\log(n))$, so the above probability is $\text{poly}(\frac{1}{n})$ for $T = c\log(n)$ and a sufficiently large constant $c$.		
	Since the expectation $\ex{\|s_{small}\|_2^2 }$ is at most $\frac{\alpha}{p}\|f\|_1 + 2\|f\|_1^2/T \leq 3\|f\|_1^2/T$, it follows that $\|s_{small}\|_2 \leq \frac{2 \|f\|_1}{\sqrt{T}}$ with high probability, which is the desired result.
\end{proof}

Applying the result of Lemma \ref{lem:smallbound}, along with the bound on $\err{s}{2}$ from the previous paragraphs, we obtain the following corollary.

\begin{corollary}
	\label{cor:Errbound}
	With high probability, if $s$ is as in Lemma \ref{lem:smallbound}, then $\err{s}{2} \leq (1+\eps )\err{f}{2}  + 2 T^{-1/2} \|f\|_1$
\end{corollary}

Now we analyze the error from \texttt{CSSS}. Observe that each row in \texttt{CSSS} contains the result of hashing a uniform sample into $6k$ buckets.  Let $s^i \in \R^n$ be the frequency vector, \emph{after scaling} by $1/p$, of the sample hashed into the $i$-th row of \texttt{CSSS}, and let $y^i \in \R^n$ be the estimate of $s^i$ taken from the $i$-th row of \texttt{CSSS}.
Let $\sigma(i):n \to [O(\log(n))] $ be the row from which Countsketch returns its estimate for $f_i$, meaning $y^*_i = y^{\sigma(i)}_i$.

\begin{theorem}
	\label{thm:cssserror}
For $\eps>0, k>1$, with high probability, when run on a general turnstile stream $f \in \R^n$ with the $\alpha$ property, \texttt{CSSS} with $6k$ columns, $O(\log(n))$ rows, returns $y^*$ such that, for every $i \in [n]$ we have \[|y^*_i - f_i| \leq   2( \frac{1}{k^{1/2}}\err{f}{2}  + \eps \|f\|_1 ) 	\]
	It follows that if $\hat{y}$ is the best $k$-sparse approximation to $y^*$, then  $\err{f}{2}\leq \|f - \hat{y} \|_2 \leq 5(k^{1/2}\eps\|f\|_1+ \err{f}{2})$ with high probability. The space required is $O(k\log(n)\ab\log(\allowbreak\alpha \log(n)/\eps))$.
\end{theorem}
\begin{proof}
	Set $S =  \Theta( \frac{\alpha^2}{\eps^2 }T^2 \log(n))$ as in Figure \ref{fig:csss}, and set $T = 4/\eps^2 + O(\log(n))$. Fix any $i \in [n]$. Now \texttt{CSSS} samples updates uniformly with probability $p > S/(2m)$, so applying Lemma \ref{lem:l1-preservation} to our sample $s^j_i$ of $f_i$ for each row $j$ and union bounding, with high probability we have $s^j_i = f_i \pm \frac{\eps}{T}\|f\|_1 =  s^q_i$ for all rows $j,q$. Then by the Countsketch guarantee of Lemma \ref{lem:cs1row}, for each row $q$ that $y^q_i =  f_i \pm( \frac{\eps}{T}\|f\|_1 + k^{-1/2}\max_j \err{s^j}{2})$ with probability $2/3$. Thus $y^*_i = y_i^{\sigma(i)} =  f_i \pm ( \frac{\eps}{T}\|f\|_1 + k^{-1/2}\max_j \err{s^j}{2})$ with high probability. Now noting that $\sqrt{T} \geq 2/\eps$, we apply Corollary \ref{cor:Errbound} to $\err{s^j}{2}$ and union bound over all $j \in [O(\log(n))]$ to obtain $\max_j \err{s^j}{2} \leq (1 + \eps) \err{f}{2} + \eps\|f\|_1$ w.h.p., and union bounding over all $i \in [n]$ gives
	$|y_i^* - f_i| \leq 2(  \frac{1}{k^{1/2}}\err{f}{2}+ \eps\|f\|_1)$ for all $i \in [n]$ with high probability.
	
	For the second claim, note that $\err{f}{2} \leq \|f - f'\|_2$ for any $k$-sparse vector $f'$, from which the first inequality follows. Now if the top $k$ coordinates are the same in $y^*$ as in $f$, then $ \|f - \hat{y}\|_2$ is at most $\err{f}{2}$ plus the $L_2$ error from estimating the top $k$ elements, which is at most $(4k(\eps\|f\|_1 + k^{-1/2}\err{f}{2})^2)^{1/2} \leq 2(k^{1/2}\eps \|f\|_1+ \err{f}{2})$. In the worst case the top $k$ coordinates of $y^*$ are disjoint from the top $k$ in $f$. Applying the triangle inequality, the error is at most the error on the top $k$ coordinates of $y^*$ plus the error on the top $k$ in $f$. Thus $ \|f - \hat{y}\|_2 \leq  5(k^{1/2}\eps \|f\|_1+ \err{f}{2})$
	as required. 
	
	For the space bound, note that the Countsketch table has $O(k\log(n))$ entries, each of which stores two counters which holds $O(S)$ samples in expectation. So the counters never exceed $\poly(S) = \poly(\frac{\alpha}{\eps }\log(n))$ w.h.p. by Chernoff bounds, and so can be stored using $O(\log(\alpha \log(n)/\eps))$ bits each (we can simply terminate if a counter gets too large).	
\end{proof}

We now address how the error term of Theorem \ref{thm:cssserror} can be estimated so as to bound the potential error. This will be necessary for our $L_1$ sampling algorithm. We first state the following well known fact about norm estimation \cite{thorup2004tabulation}, and give a proof for completeness.

\begin{lemma}
	\label{Lem:csL2}
	Let $R\in \R^k$ be a row of the Countsketch matrix with $k$ columns run on a stream with frequency vector $f$. Then with probability $99/100$, we have $\ab\sum_{i=1}^{k} R_i^2\ab = (1 \pm O(k^{-1/2}))\|f\|_2^2$
\end{lemma}
\begin{proof} Let $\mathbf{1}(E)$ be the indicator function that is $1$ if the event $E$ is true, and $0$ otherwise. Let $h:[n] \to [k]$ and $g:[n] \to \{0,1\}$ be $4$-wise independent hash functions which specify the row of Countsketch. Then $\ex{ \sum_{i=1}^{k} R_i^2 } = \ex{ \sum_{j=1}^k (\sum_{i=1}^n\mathbf{1}(h(i) = j) g(i)f_i)^2} = \ex{ \sum_{i=1}^nf_i^2 } + \ab \ex{ \sum_{j=1}^k\ab (\sum_{i1 \neq i2}\mathbf{1}(h(i1) = j) \mathbf{1}(h(i2) = j) g(i1) g(i2) f_{i1} f_{i2})^2}$. By the $2$-wise independence of $g$, the second quantity is $0$, and we are left with $\|f\|_2^2$. By a similar computation using the full $4$-wise independence, we can show that $Var(\sum_{i=1}^{k} R_i^2 ) = 2(\|f\|_2^4 - \|f\|_4^4)/k \leq 2\|f\|_2^4/k$.  Then by Chebyshev's inequality, we obtain $\pr{ |\sum_{i=1}^{k} R_i^2  - \|f\|_2^2 | > 10\sqrt{2} \|f\|_2^2/\sqrt{k}} < 1/100$ as needed.
\end{proof}

\begin{lemma}
	\label{lem:ErrEst}
	For $k>1, \eps \in (0,1)$, given a $\alpha$-property stream $f$, there is an algorithm that can produce an estimate $v$ such that $\err{f}{2} \leq v \leq 45k^{1/2}\eps \|f\|_1+ 20\err{f}{2}$ with high probability. The space required is the space needed to run two instances of \texttt{CSSS} with paramters $k,\eps$.
\end{lemma}
\begin{proof}
	By Lemma \ref{Lem:csL2}, the $L_2$ of row $i$ of $\ttx{CSSS}$ with a constant number of columns will be a $(1 \pm 1/2)$ approximation of $\|s^i\|_2$ with probability $99/100$, where $s^i$ is the scaled up sampled vector corresponding to row $i$. 
	Our algorithm is to run two copies of \texttt{CSSS} on the entire stream, say $\ttx{CSSS}_1$ and $\ttx{CSSS}_2$, with $k$ columns and sensitivity $\eps$.  At the end of the stream we compute $y^*$ and $\hat{y}$ from $\ttx{CSSS}_1$ where $\hat{y}$ is the best $k$-sparse approximation to $y^*$. We then feed $-\hat{y}$ into $\ttx{CSSS}_2$. The resulting $L_2$ of the $i$-th row of $\ttx{CSSS}_2$ is $(1 \pm 1/2)\|s^i_2 - \hat{y} \|_2$ (after rescaling of $s^i$) with probability $99/100$, where $s^i_2$ is the sample corresponding to the $i$-th row of $\ttx{CSSS}_2$.
	
	Now let $T = 4/\eps^2$ as in Theorem \ref{thm:cssserror}, and let $s^i$ be the sample corresponding to the $i$-th row of \texttt{CSSS}$_1$. Then for any $i$, $\|s^i_{small}\|_2 \leq 2T^{-1/2} \|f\|_1$ w.h.p. by Lemma \ref{lem:smallbound}, and the fact that $\|f_{small}\|_2 \leq T^{-1/2}\|f\|_1$ follows from the definition of $small$.  Furthermore, by Lemma \ref{lem:l1-preservation}, we know w.h.p. that $|s^i_j - f_j| < \eps |f_j|$ for all $j \in big$, and thus $\|s^i_{big} - f_{big} \|_2 \leq \eps \|f\|_1$. Then by the reverse triangle inequality we have $\big|\|s^i - \hat{y}\|_2 - \|f - \hat{y}\|_2\big| \leq \|s^i - f \|_2 \leq\|s^i_{small} - f_{small} \|_2 + \|s^i_{big} - f_{big} \|_2 \leq  5\eps \|f\|_1$. Thus $\|s^i - \hat{y}\|_2 = \|f - \hat{y}\|_2 \pm 5\eps \|f\|_1$ for all rows $i$ w.h.p, so the $L_2$ of the $i$-th row of $\ttx{CSSS}_2$ at the end of the algorithm is the value $v_i$ such that $ \frac{1}{2}(\|f - \hat{y}\|_2 - 5\eps\|f\|_1)\leq  v_i \leq 2 ( \|f - \hat{y}\|_2 + 5\eps\|f\|_1)$ with probability $9/10$ (note that the bounds do not depend on $i$). Taking $v = 2\cdot \ttx{median}(v_1,\dots,v_{O(\log(n))}) +  5\eps\|f\|_1$, it follows that  $\|f - \hat{y}\|_2  \leq  v \leq 4\|f - \hat{y}\|_2 + 25\eps \|f\|_1$ with high probability. Applying the upper and lower bounds on $\|f - \hat{y}\|_2$ given in Theorem \ref{thm:cssserror} yields the desired result. The space required is the space needed to run two instances of \texttt{CSSS}, as stated.
\end{proof}

\subsection{Inner Products}\label{sec:innerprod}
Given two streams $f,g \in \R^n$, the problem of estimating the inner product $\langle f,g \rangle = \sum_{i=1}^n f_ig_i$ has attracted considerable attention in the streaming community for its usage in estimating the size of join and self-join relations for databases \cite{alon1999tracking, rusu2008sketches,rusu2007pseudo}. 
 For unbounded deletion streams, to obtain an $\eps \|f\|_1\|g\|_1$-additive error approximation the best known result requires $O(\eps^{-1}\log(n))$ bits with a constant probability of success \cite{cormode2005improved}. 
  We show in Theorem \ref{thm:InnerProdHard} that $\Omega(\eps^{-1}\log(\alpha))$ bits are required even when $f,g$ are \textit{strong} $\alpha$-property streams. This also gives a matching $\Omega(\eps^{-1}\ab\log(n))$ lower bound for the unbounded deletion case 
  
   In Theorem \ref{thm:innerprod}, we give a matching upper bound for $\alpha$-property streams up to $\log\log(n)$ and $\log(1/\eps)$ terms.  We first prove the following technical Lemma, which shows that inner products are preserved under a sufficiently large sample.
   
   	\begin{lemma}\label{lem:sampinnerprod}
   	Let $f,g \in \R^n$ be two $\alpha$-property streams with lengths $m_f,m_g$ respectively. Let $f',g' \in \R^n$ be unscaled uniform samples of $f$ and $g$, sampled with probability $p_f \geq s/m_f$ and $p_g \geq s/m_g$ respectively, where $s = \Omega(\alpha^2 /\eps^2)$. Then with probability $99/100$, we have $\langle p_f^{-1}f', p_g^{-1}g' \rangle = \langle f,g, \rangle \pm \eps \|f\|_1 \|g\|_1$. 
\end{lemma}
\begin{proof}
	We have $\ex{\langle p_f^{-1}f' , p_g^{-1}g' \rangle} \ab= \sum_i \ex{p_f^{-1} f_i'} \ex{p_g^{-1} g_i'} = \langle f,g \rangle$. Now the $(f_i'g_i')$'s are independent, and thus we need only compute $\text{Var}(p_f^{-1} p_g^{-1}f_i'g_i')$. We have $\text{Var} (p_f^{-1} p_g^{-1}f_i' g_i') = \ab(p_f^{-1}\ab p_g^{-1})^2\ab\ex{(f_i')^2} \ex{(g_i')^2} - (f_ig_i)^2$. Let $X_{i,j}$ be $\pm1$ if the $j$-th update to $f_i$ is sampled, where the sign indicates whether the update was an insertion or deletion. Let $f^+,f^-$ be the insertion and deletion vectors of $f$ respectively (see Definition \ref{def:alphaprop}), and let $F = f^+ + f_i$. Then $f = f^+ - f^-$, and $f_i' = \sum_{j=1}^{f_i^+ + f_i^-} X_{ij}$. We have
	\[	\ex{(f_i')^2/p_f^2} = \ex{p_f^{-2} (\sum_{j=1}^{f_i^+ + f_i^-} X_{ij}^2 + \sum_{j_1 \neq j_2} X_{ij_1}X_{ij_2} )  }	\]
	\[	\leq p_f^{-1} F_i + (f_i^+)^2 + (f_i^-)^2 - 2f_i^+ f_i^-	\]
	\[		= p_f^{-1}F_i +   f_i^2 \]
	Now define $g^+,g^-$ and $G = g^+ + g^-$ similarly.	 Note that the $\alpha$-property states that $m_f = \|F\|_1 \leq \alpha \|f\|_1$ and $m_g= \|G\|_1 \leq \alpha \|g\|_1$. Moreover, it gives $p^{-1}_f \leq m_f / s \leq \eps^2\|f\|_1/\alpha $, and $p^{-1}_g \leq m_g / s \leq \eps^2 \|g\|_1/\alpha$. 
	Then $\text{Var}(\langle p_f^{-1} f' , p_g^{-1} g' \rangle = \sum_{i=1}^n\text{Var}(p_f^{-1} p_g^{-1}f_i'g_i') \leq \sum_{i=1}^n \ab (f_i^2 p_g^{-1} G_i + g_i^2 p_f^{-1} F_i + p_f^{-1} p_g^{-1}\ab F_iG_i)$. First note that $\sum_{i =1}^n p_g^{-1}\ab f_i^2  G_i \leq p_g^{-1} \|G\|_1 \|f\|_1^2 \leq \eps^2 \|g\|_1^2 \|f\|_1^2$ and similarly $\sum_{i=1}^n p_f^{-1} g_i^2\ab F_i \leq \eps^2 \|g\|_1^2 \|f\|_1^2$. Finally, we have $\sum_{i=1}^n p_f^{-1}p_g^{-1}F_iG_i \leq p_f^{-1}p_g^{-1} \|F\|_1\|G\|_1 \leq \eps^4 \|f\|_1^2 \|g\|_1^2$, so the variance is at most  $3\eps^2 \|f\|_1^2 \|g\|_1^2$. Then by Chebyshev's inequality:
	\[	\pr{\big|	\langle p_f^{-1}f', p_g^{-1}g' \rangle - \langle f,g\rangle	\big| > 30 \eps \|f\|_1 \|g\|_1}\leq 1/100	\]
	and rescaling $\eps$ by a constant gives the desired result.		  
\end{proof}

  Our algorithm obtains such a sample as needed for Lemma \ref{lem:sampinnerprod} by sampling in exponentially increasing intervals. Next, we hash the universe down by a sufficiently large prime to avoid collisions in the samples, and then run a inner product estimator in the smaller universe. Note that this hashing is not pairwise independent, as that would require
the space to be at least $\log n$ bits; rather the hashing just has the property that it preserves distinctness
with good probability. 
We now prove a fact that we will need for our desired complexity. 

\begin{lemma}\label{lem:easymod}
	Given a $\log(n)$ bit integer $x$, the value $x \pmod p$ can be computed using only $\log\log(n) + \log(p)$ bits of space.
\end{lemma}
\begin{proof}
	We initialize a counter $c \leftarrow 0$. Let $x_1,x_2,\dots,x_{\log(n)}$ be the bits of $x$, where $x_1$ the least significant bit. Then we set $y_1 = 1$, and at every step $t \in [\log(n)]$ of our algorithm, we store $y_{t}, y_{t-1}$, where we compute $y_t$ as $y_t = 2 y_{t-1} \pmod p$. This can be done in $O(\log(p))$ space. Our algorithm then takes $\log(n)$ steps, where on the $t$-th step it checks if $x_t = 1$, and if so it updates $c \leftarrow c + y_t \pmod p$, and otherwise sets $c \leftarrow c$. At the end we have $c = \sum_{i=0}^{\log(n)} 2^i x_i \pmod p = x \pmod p $ as desired, and $c$ never exceeds $2p$, and can be stored using $\log(p)$ bits of space. The only other value stored is the index $t \in [\log(n)]$, which can be stored using $\log\log(n)$ bits as stated.   
\end{proof}

We now recall the Countsketch data structure (see Section \ref{sec:fakecs}). Countsketch picks a $2$-wise independent hash function $h:[n] \to [k]$ and a $4$-wise independent hash function $\sigma: [n] \to \{1,-1\}$, and creates a vector $A \in \R^k$. Then on every update $(i_t,\Delta_t)$ to a stream $f \in \R^n$, it updates $A_{h(i_t)} \leftarrow A_{h(i_t)}  + \sigma_{i_t} \Delta_t$ where $\sigma_i = \sigma(i)$. Countsketch can be used to estimate inner products. In the following Lemma, we show that feeding uniform samples of $\alpha$-property streams $f,g$ into two instances of Countmin, denoted $A$ and $B$, then $\langle A,B\rangle$ is a good approximation of $\langle f,g \rangle$ after rescaling $A$ and $B$ by the inverse of the sampling probability.  

\begin{lemma}\label{lem:countmin}
	Let $f',g'$ be uniformly sampled rescaled vectors of general-turnstile $\alpha$-property streams $f,g$ with lengths $m_f,m_g$ and sampling probabilities $p_f,p_g$ respectively. Suppose that $p_f \geq s/m_f$ and $p_g \geq s / m_g$, where $s = \Omega(\alpha^2\log^7(n)\ab T^2\ab\eps^{-10})$. Then let
	$A \in \R^k$ be the Countsketch vector run on $f'$, and let $B \in \R^k$ be the Countsketch vector run on $g'$, where $A,B$ share the same hash function $h$ and $k = \Theta(1/\eps)$ Then
	\[ \sum_{i =1}^k A_{i} B_{i} = \langle f,g \rangle \pm \eps \|f\|_1 \|g\|_1\]
	with probability $11/13$. 
\end{lemma}
\begin{proof}
	Set $k = 100^2/\eps$. Let $Y_{ij}$ indicate $h(i) = h(j)$, and let $X_{ij} = \sigma_i \sigma_j f_i'g_j'\ab Y_{ij}$. We have $\sum_{i =1}^k A_{i} B_{i} = \langle f',g' \rangle + \sum_{i \neq j} X_{ij}$. 
	Let $\eps_0 = \Theta(\log^2(n)/\eps^3)$ and let $T = \Theta(\log(n)/\eps^2)$, so we can write $s = \Omega(\alpha^2 \log(n)\ab T^2/\eps_0^2)$ (this will align our notation with that of Lemma \ref{lem:smallbound}). Define $\mathcal{F} = \{ i \in [n]\; | \; |f_i| \geq \|f\|_1/T\}$ and $\mathcal{G} = \{ i \in [n]\; | \; |g_i| \geq \|g\|_1/T\}$. We now bound the term $|\sum_{i \neq j, (i,j) \in \mathcal{F} \times \mathcal{G}} \ab X_{ij}| \leq \sum_{(i,j) \in \mathcal{F} \times \mathcal{G}} \ab Y_{ij} |f_i' g_j'|$. 
	
	Now by Lemma \ref{lem:l1-preservation} and union bounding over all $i \in [n]$, we have $\|f' - f\|_\infty \leq \eps_0 \|f\|_1/T$ and $\|g' - g\|_\infty \leq \eps_0 \|g\|_1/T$ with high probability, and we condition on this now. Thus for every $(i,j) \in \mathcal{F} \times \mathcal{G}$, we have $f_i' = (1 \pm \eps_0)f_i$ and $g_j' = (1 \pm \eps_0)g_j$.  It follows that $\sum_{(i,j) \in \mathcal{F} \times \mathcal{G}} \ab Y_{ij} |f_i' g_j'| \leq 2\sum_{(i,j) \in \mathcal{F} \times \mathcal{G}} Y_{ij} \ab |f_i| | g_j|$ (using only $\eps_0 < 1/4$). Since $\ex{Y_{ij}} = 1/k$, we have 
	\[	\ex{\sum_{(i,j) \in \mathcal{F} \times \mathcal{G}} Y_{ij} \ab |f_i| | g_j|} \leq \frac{1}{k} \|f\|_1\|g\|_1\]
	By Markov's inequality with $k = 100^2/\eps$, it follows that $|\sum_{i \neq j, (i,j) \in \mathcal{F} \times \mathcal{G}} \ab X_{ij}| \leq 2\sum_{(i,j) \in \mathcal{F} \times \mathcal{G}} \ab Y_{ij} \ab |f_i| | g_j| \leq \eps\|f\|_1 \|g\|_1$ with probability greater than $1 - (1/1000 + O(n^{-c})) > 99/100$, where the $O(n^{-c})$ comes from conditioning on the high probability events from Lemma \ref{lem:l1-preservation}. Call this event $\mathcal{E}_1$.
	
	Now let $\mathcal{F}^C = [n] \setminus \mathcal{F}$, and  $\mathcal{G}^C = [n] \setminus \mathcal{G}$. Let $\mathcal{A} \subset [n]^2$ be the set of all $(i,j)$ with $i \neq j$ and such that either $i \in \mathcal{F}^C$ or $j \in \mathcal{G}^C$. Now consider the variables $\{|X_{ij}|\}_{i < j}$, and let $X_{ij},X_{pq}$ be two such distinct variables. Then $\ex{X_{ij}X_{pq}} = \ex{f_i' g_i' f_p' g_q' Y_{ij} Y_{pq} \sigma_i \sigma_j \sigma_p \sigma_q}$. Now the variables $\sigma$ are independent from $Y_{ij} Y_{pq}$, which are determined by $h$. Since $i < j$ and $p < q$ and $(i,j) \neq (p,q)$, it follows that one of $i,j,p,q$ is unique among them. WLOG it is $i$, so by $4$-wise independence of $\sigma$ we have $\ex{f_i' g_i' f_p' g_q' Y_{ij} Y_{pq} \sigma_i \sigma_j \sigma_p \sigma_q} = \ex{f_i' g_i' f_p' g_q' Y_{ij} Y_{pq}  \sigma_j \sigma_p \sigma_q}\ab  \ex{\sigma_i} = 0 = \ex{X_{ij}}\ex{X_{pq}}$. Thus the variables $\{|X_{ij}|\}_{i < j}$ (and symmetrically $\{|X_{ij}|\}_{i > j}$) are uncorrelated so 
	\[\text{Var}(\sum_{ i < j,(i,j) \in \mathcal{A}} X_{ij}) =\sum_{ i < j,(i,j) \in \mathcal{A}}\text{Var}(X_{ij}) \leq \sum_{(i,j) \in \mathcal{A}} (f_i' g_j')^2/k \]
	Since $\ex{\sum_{ i < j,(i,j) \in \mathcal{A} }X_{ij} }= 0$, by Chebyshev's inequality with $k=100^2/\eps$, we have $|\sum_{ i < j,(i,j) \in \mathcal{A} } X_{ij}| \leq ( \eps\sum_{(i,j) \in \mathcal{A}} \ab(f_i' g_j')^2)^{1/2}$ with probability $99/100$. So by the union bound and a symmetric arguement for $j > i$, we have $|\sum_{(i,j) \in \mathcal{A} } X_{ij}|\ab \leq |\sum_{ i < j,(i,j) \in \mathcal{A} } X_{ij}| + |\sum_{ i > j,(i,j) \in \mathcal{A} } X_{ij}| \leq 2 ( \eps\sum_{(i,j) \in \mathcal{A}} \ab(f_i'\ab g_j')^2)^{1/2}$ with probability $1-(2/100) = 49/50$. Call this event $\mathcal{E}_2$.
	We have
	\[	\sum_{(i,j) \in \mathcal{A}} (f_i' g_j')^2  = \sum_{i \neq j, (i,j) \in \mathcal{F} \times \mathcal{G}^C} (f_i' g_j')^2 	 + \sum_{i \neq j, (i,j) \in \mathcal{F}^C \times \mathcal{G}} (f_i' g_j')^2\] \[ + \sum_{i \neq j, (i,j) \in \mathcal{F}^C \times \mathcal{G}^C} (f_i' g_j')^2	\]
	Now the last term is at most $\|f_{\mathcal{F}^C}'\|_2^2 \|g_{\mathcal{G}^C}'\|_2^2$, which is at most $16\|f\|_1^2\|g\|_1^2/T^2 \leq \eps^4 \|f\|_1^2\|g\|_1^2$ w.h.p. by Lemma \ref{lem:smallbound} applied to both $f$ and $g$ and union bounding (note that $\mathcal{F}^C,\mathcal{G}^C$ are exactly the set $small$ in Lemma \ref{lem:smallbound} for their respective vectors). We hereafter condition on this w.h.p. event that $\|f_{\mathcal{F}^C}'\|_2^2 \leq 4\|f\|_1^2/T$ and $\|g_{\mathcal{G}^C}'\|_2^2 \leq 4\|g\|_1^2/T$ (note $T > 1/\eps^2$).

	Now as noted earlier, w.h.p. we have $f_i' = (1\pm \eps_0) f_i$ and $f_j' = (1\pm \eps_0) f_j$ for $i \in \mathcal{F}$ and $j \in \mathcal{G}$. Thus  $\|f_{\mathcal{F}}'\|_2^2 = (1 \pm O(\eps_0)) \|f_{\mathcal{F}}\|_2^2 < 2\|f\|_1^2$ and $\|g_{\mathcal{G}}'\|_2^2 \leq 2\|g\|_1^2$.  Now the term $\sum_{i \neq j, (i,j) \in \mathcal{F} \times \mathcal{G}^C} (f_i' g_j')^2$ is at most $\sum_{i \in \mathcal{F}} (f_i')^2 \|g_{\mathcal{G}^c}'\|_2^2 \leq  O(\eps^2) \|f_{\mathcal{F}}'\|_2^2 \|g\|_1^2 \leq O(\eps^2)\|f\|_1^2 \|g\|_1^2$. Applying a symmetric argument, we obtain $ \sum_{i \neq j, (i,j) \in \mathcal{F}^C \times \mathcal{G}} \ab(f_i' g_j')^2	  \leq O(\eps^2)\|f\|_1^2 \|g\|_1^2$ with high probability. Thus each of the three terms is $O(\eps^2)\|f\|_1^2 \ab \|g\|_1^2$, so 
	\[	\sum_{(i,j) \in \mathcal{A}} (f_i' g_j')^2  = O(\eps^2)\|f\|_1^2 \ab \|g\|_1^2 \]
	with high probability. Call this event $\mathcal{E}_3$. Now by the union bound, $\pr{\mathcal{E}_1 \cup \mathcal{E}_2 \cup \mathcal{E}_3} > 1 - (1/50 + 1/100 + O(n^{-c})) > 24/25$. Conditioned on this, the error term $| \sum_{i \neq j} X_{ij}| \leq  |\sum_{i \neq j, (i,j) \in \mathcal{F} \times \mathcal{G}} \ab X_{ij}| + |\sum_{(i,j) \in \mathcal{A}} \ab X_{ij}|$ is at most $\eps\|f||_1\|g\|_1 + 2 ( \eps\sum_{(i,j) \in \mathcal{A}} \ab(f_i'\ab g_j')^2)^{1/2} = \eps\|f||_1\|g\|_1 +2 (O(\eps^3) \|f\|_1^2\|g\|_1^2)^{1/2}  = O(\eps) \|f\|_1\|g|_1$ as desired. Thus with probability $24/25$ we have  $\sum_{i =1}^k A_{i} B_{i} = \langle f',g' \rangle \pm O(\eps) \|f\|_1 \|g\|_1$. By Lemma \ref{lem:sampinnerprod}, noting that in the notation of this Lemma the vectors $f',g'$ have already been scaled by $p_f^{-1}$ and $p_g^{-1}$, we have  $\langle f',g' \rangle  = \langle f,g\rangle \pm \eps \|f\|_1\|g\|_1$ with probability $99/100$. Altogether, with probability $1-(1/25 + 1/100) > 11/13$ we have $\sum_{i =1}^k A_{i} B_{i} = \langle f,g \rangle \pm O(\eps) \|f\|_1 \|g\|_1$, which is the desired result after a constant rescaling of $\eps$. 
\end{proof}

We will apply Lemma \ref{lem:countmin} to complete our proof of our main Theorem.

\begin{theorem}\label{thm:innerprod}
	Given two general-turnstile stream vectors $f,g$ with the $\alpha$-property, there is a one-pass algorithm which with probability $11/13$ produces an estimate $\texttt{IP}(f,g)$ such that $\texttt{IP}(f,g) =\langle f,g \rangle \pm O(\eps)\|f\|_1\|g\|_1$ using $O(\eps^{-1}\ab \log(\frac{\alpha\log(n)}{\eps}))$ bits of space.
\end{theorem}
\begin{proof}
	
	Let $s = \Theta(\alpha^2 \log(n)^7/\eps^{10})$, and let $I_r = [s^r,s^{r+2}]$. Then for every $r = 1,2,\dots, \log_s(n)$, we choose a random prime $P = [D,D^3]$ where $D = 100s^4$. We then do the following for the stream $f$ (and apply the same following procedure for $g$). For every interval $I_r$ and time step $t \in I_r$, we sample the $t$-th update $(i_t,\Delta_t)$ to $f$ with probability $s^{-r}$ (we assume $\Delta_t \in \{1,-1\}$, for bigger updates we expand them into multiple such unit updates.). If sampled, we let $i_t'$ be the $\log(|P|)$ bit identity obtained by taking $i_t$ modulo $P$. We then feed the unscaled update $(i_t', \Delta_t )$ into an instance of count-sketch with $k = O(1/\eps)$ buckets. Call this instance \texttt{CS}$_r^f$. At any time step $t$, we only store the instance \texttt{CS}$_r^f$ and \texttt{CS}$_{r+1}^f$ such that $t \in I_r \cap I_{r+1}$. At the end of the stream it is time step $m^f$, and fix $r$ such that $m^f \in I_{r}\cap I_{r+1}$.
	
	Now let $f' \in \R^{n}$ be the (scaled up by $s^r$) sample of $f$ taken in $I_r$. Let $f'' \in \R^p$ be the (unscaled) stream on $|P|$ items that is actually fed into \texttt{CS}$_r^f$. Then \texttt{CS}$_r^f$ is run on the stream $f''$ which has a universe of $|P|$ items. Let $F \in \R^k$ be the Countsketch vector from the instance \texttt{CS}$_r^f$.
	
	Let $\hat{f}$ be the frequency vector $f$ restricted to the suffix of the stream $s^r,s^{r}+1,\dots, m^f $ (these were the updates that were being sampled from while running \texttt{CS}$_r$). Since $m^f  \in I_{r+1}$, we have $m^f \geq s^{r+1}$, so $\|f^{(s^r)}\|_1 \leq m/s < \eps \|f\|_1$ (by the $\alpha$ property), meaning the $L_1$ mass of the prefix of the stream $1,2,\dots, s^r$ is an $\eps$ fraction of the whole stream $L_1$, so removing it changes the $L_1$ by at most $\eps \|f\|_1$. It follows that $\|\hat{f} - f\|_1 \leq O(\eps)\|f\|_1$ and thus $\|\hat{f}\|_1 = (1 \pm O(\eps))\|f\|_1$.  If we let $\hat{g}$ be defined analogously by replacing $f$ with $g$ in the previous paragraphs, we obtain $\|\hat{g} - g\|_1 \leq O(\eps)\|g\|_1$ and  $\|\hat{g}\|_1 = (1 \pm O(\eps))\|g\|_1$ as well.

	Now with high probability we sampled fewer than $2s^{j+2}/s^j = 2s^2$ distinct identities when creating $f'$, so $f'$ is $2s^2$-sparse. Let $J \subset [n] \times [n]$ be the set of pairs pf indices $i,j$ with $f_i', f_{j}'$ non-zero. Then $|J| < 2s^4$. Let $Q_{i,j}$ be the event that $i - j = 0  \pmod P$. For this to happen, $P$ must divide the difference. By standard results on the density of primes, there are $s^8$ primes in $[D,D^3]$, and since $|i-j|\leq n$ it follows that $|i-j|$ has at most $\log(n)$ prime factors. So $\pr{Q_{i,j}} < \log(n)/s^8$,  Let $Q =\cup_{(i,j) \in J} Q_{i,j} $, then by the union bound $\pr{ Q} < s^{-3}$. It follows that no two sampled identities collide when being hashed to the universe of $p$ elements with probability $1 - 1/s^3$. 
	
	Let $\text{supp}(f') \subset [n]$ be the support of $f'$ (non-zero indices of $f'$).  Conditioned on $Q$, we have $p_f^{-1} f_{i \pmod p}'' = f_i'$ for all $i \in\text{supp}(f')$, and $f_i'' = 0$ if $i \neq j \pmod p$ for any $j \in \text{supp}(f')$. Thus there is a bijection between 
	$\text{supp}(f')$ and  $\text{supp}(f'')$, and the values of the respective coordinates of $f',f''$ are equal under this bijection. Let $g,g'',\hat{g},m_g,p_g$ be defined analogously to $f,f''$ by replacing $f$ with $g$ in the past paragraphs, and let $G \in \R^k$ be the Countsketch vector obtained by running Countsketch on $f$ just as we did to obtain $F \in \R^k$. 
	
	Conditioned on $Q$ occurring (no collisions in samples) for both $f''$ and $g''$ (call these $Q_f$ and $Q_g$ respectively), which together hold with probability $1-O(s^{-3})$ by a union bound, the non-zero entries of $p_f^{-1}f''$ and $f'$ and of $p_g^{-1}g''$ and $g'$ are identical. Thus
	the scaled up Countsketch vectors $p_f^{-1}F$ and $p_g^{-1}G$ of $F,G$ obtained from running Countsketch on $f'',g''$ are identical in distribution to running Countsketch on $f',g'$. This holds because Countsketch hashes the non-zero coordinates $4$-wise independently into $k$ buckets $A,B \in \R^k$. Thus conditioned on $Q_f,Q_g$, we can assume that $p_f^{-1}F$ and $p_g^{-1}G$ are the result of running Countsketch on $f'$ and $g'$. We claim that $p_f^{-1} p_g^{-1} \langle F,G\rangle$ is the desired estimator.
	
	Now recall that $f'$ is a uniform sample of $\hat{f}$, as is $g'$ of $\hat{g}$. Then applying the Countsketch error of Lemma \ref{lem:countmin}, we have $p_f^{-1} p_g^{-1}\langle F,G \rangle = \langle \hat{f}, \hat{g} \rangle + \eps \|\hat{f}\|_1  \|\hat{g}\|_1 = \langle \hat{f}, \hat{g} \rangle + O(\eps) \|f\|_1  \|g\|_1$ with probability $11/13$. Now since $\|\hat{f} - f\|_1 \leq O(\eps)\|f\|_1$  and $\|\hat{g} - g\|_1 \leq O(\eps)\|g\|_1$, we have $\langle \hat{f}, \hat{g} \rangle = \langle f,g\rangle \pm  \sum_{i=1}^n( \eps g_i \|f\|_1 + \eps f_i \|g\|_1 + \eps^2 \|f\|_1 \|g\|_1 = \langle f,g \rangle \pm O(\eps)\|f\|_1\|g\|_1$. Thus 
	\[p_f^{-1} p_g^{-1} \langle F,G \rangle = \langle f ,g\rangle \pm O(\eps) \|f\|_1 \|g\|_1 \] 
	as required. This last fact holds deterministically using only the $\alpha$-property, so the probability of success is $11/13$ as stated.
	
	For the space, at most $2s^2$ samples were sent to each of $f'',g''$ with high probability. Thus the length of each stream was at most $\poly(\alpha\log(n)\ab/\eps)$, and each stream had $P = \poly(\alpha\log(n)\ab/\eps)$ items. Thus each counter in $A,B \in \R^k$ can be stored with $O(\log(\alpha \log(n)/\eps))$ bits. So storing $A,B$ requires $O(\eps^{-1} \ab \log(\ab\alpha \log(n)/\eps))$ bits. Note we can safely terminate if too many samples are sent and the counters become too large, as this happens with probability $O(\poly(1/n))$ . The $4$-wise independent hash function $h:[P] \to [k]$ used to create $A,B$ requires $O(\log(\alpha \log(n)/\eps))$ bits.  
	
	Next, by Lemma \ref{lem:easymod}, the space required to hash the $\log(n)$-bit identities down to $[P]$ is $\log(P) + \log\log(n)$, which is dominated by the space for Countsketch. Finally, we can assume that $s$ is a power of two so $p_f^{-1},p_g^{-1} = \poly(s) = 2^q$ can be stored by just storing the exponent $q$, which takes $\log\log(n)$ bits. To sample we then flip $\log(n)$ coins sequentially, and keep a counter to check if the number of heads reaches $q$ before a tail is seen. 
\end{proof}

%% file: HeavyHittersViaCountSketch.aux.tex
\section{$L_1$ Heavy Hitters}
\label{sec:heavyhitters}
As an application of the Countsketch sampling algorithm presented in the last section, we give an improved upper bound for the classic $L_1$ $\eps$-heavy hitters problem in the $\alpha$-property setting. Formally, given $\eps \in (0,1)$, the $L_1$ $\eps$-heavy hitters problem asks to return a subset of $[n]$ that contains all items $i$ such that $|f_i| \geq \eps \|f\|_1$, and no items $j$ such that $|f_j| < (\eps /2)\|f\|_1$. 

The heavy hitters problem is one of the most well-studied problems in the data stream literature. For general turnstile unbounded deletion streams, there is a known lower bound of $\Omega(\eps^{-1}\log(n)\log(\eps n))$ (see \cite{ba2010lower}, in the language of compressed sensing, and \cite{Jowhari:2011}), and the Countsketch of \cite{charikar2002finding} gives a matching upper bound (assuming $\eps^{-1} = o(n)$). In the insertion only case, however, the problem can be solved using $O(\eps^{-1}\log(n))$ bits \cite{bhattacharyya2016optimal}, and for the strictly harder $L_2$ heavy hitters problem (where $\|f\|_1$ is replaced with $\|f\|_2$ in the problem definition), there is an $O(\eps^{-2}\log(1/\eps)\log(n))$-bit algorithm \cite{braverman2016bptree}.
In this section, we beat the lower bounds for unbounded deletion streams in the $\alpha$-property case. We first run a subroutine to obtain a value  $R = (1 \pm 1/8)\|f\|_1$ with probability $1-\delta$. To do this, we use the following algorithm from \cite{kane2010exact}.
\begin{fact}[\cite{kane2010exact}]
	\label{fact:L1est}
	There is an algorithm which gives a $(1 \pm \eps)$ multiplicative approximation with probability $1-\delta$ of the value $\|f\|_1$ using space $O(\eps^{-2}\log(n)\log(1/\delta))$.
\end{fact}

Next, we run an instance of \texttt{CSSS} with parameters $k = 32/\eps$ and $\epsilon/32$ to obtain our estimate $y^*$ of $f$. This requires space $O(\eps^{-1} \log(n) \log(\frac{\alpha \log(n)}{\eps}))$, and by Theorem \ref{thm:cssserror} gives an estimate $y^* \in \R^n$ such that $|y_i^* - f_i| < 2(\sqrt{\eps/32} \err{f}{2} + \eps \|f\|_1/32)$ for all $i \in [n]$ with high probability. We then return all items $i$ with $|y_i^*| \geq 3\eps R/4$. Since the top $1/\eps$ elements do not contribute to $\err{f}{2}$, the quantity is maximized by having $k$ elements with weight $\|f\|_1/k$, so $\err{f}{2} \leq k^{-1/2}\|f\|_1$. Thus $\|y_i^* - f\|_\infty < (\eps/8)\|f\|_1$.

Given this, it follows that for any $i \in [n]$ if  $|f_i| \geq \eps \|f\|_1$, then $|y_i^*|>  (7\eps/8) \|f\|_1 > (3\eps/4)R$.  Similarly if $|f_i| < (\eps /2)\|f\|_1$, then $|y^*_i| < (5\eps/8)\|f\|_1 < (3\eps/4)R$. So our algorithm correctly distinguishes $\eps$ heavy hitters from items with weight less than $\eps/2$. The probability of failure is $O(n^{-c})$ from \texttt{CSSS} and $\delta$ for estimating $R$, and the space required is  $O(\eps^{-1}\log(n)$ $\log(\frac{\alpha\log(n)}{\eps}))$ for running \texttt{CSSS} and $O(\log(n)\log(1/\delta))$ to obtain the estimate $R$. This gives the following theorem.

\begin{theorem}\label{thm:HHgen}
Given $\eps \in (0,1)$, there is an algorithm that solves the $\eps$-heavy hitters problem for general turnstile $\alpha$-property streams with probability $1-\delta$ using space $O(\eps^{-1}\log(n)$ $\log(\frac{\alpha\log(n)}{\eps}) + \log(n)\allowbreak\log(1/\delta))$. 
\end{theorem}

Now note for strict turnstile streams, we can compute $R = \|f\|_1$ exactly with probability $1$ using an $O(\log(n))$-bit counter. Since the error bounds from \texttt{CSSS} holds with high probability, we obtain the following result.
\begin{theorem}\label{thm:HHstrict}
Given $\eps \in (0,1)$, there is an algorithm that solves the $\eps$-heavy hitters problem for strict turnstile $\alpha$-property streams with high probability using space $O(\eps^{-1}\log(n)$ $\log(\alpha\; \ab\log(n)/\ab\eps))$. 
\end{theorem}

%% file: L1_sampling_NEW.tex
	\section{$L_1$ Sampling}
	\label{sec:L1Samp}
	Another problem of interest is the problem of designing $L_p$ samplers. First introduced by Monemizadeh and Woodruff in \cite{monemizadeh20101}, it has since been observed that $L_p$ samplers lead to alternative algorithms for many important streaming problems, such as heavy hitters, $L_p$ estimation, and finding duplicates \cite{andoni2010streaming, monemizadeh20101, Jowhari:2011}.
	
	Formally, given a data stream frequency vector $f$, the problem of returning an $\epsilon$-approximate relative error uniform $L_p$ sampler is to design an algorithm that returns an index $i\in[n]$ such that
	\[	\pr{i = j} = (1 \pm \epsilon)\frac{|f_j|^p}{\|f\|_p^p}	\]
for every $j \in [n]$. An approximate $L_p$ sampler is allowed to fail with some probability $\delta$, however in this case it must not output any index.  For the case of $p=1$, the best known upper bound is  $O(\eps^{-1}\log(\eps^{-1}) \log^2(n)\log(\delta^{-1}))$ bits of space, and there is also an $\Omega(\log^2(n))$ lower bound for $L_p$ samplers with $\eps = O(1)$ for any $p$ \cite{Jowhari:2011}. 
	In this section, using the data structure \texttt{CSSS} of Section \ref{sec:fakecs}, we will design an $L_1$ sampler for strict-turnstile strong $L_1$ $\alpha$-property streams using $O(\eps^{-1}\log(\eps^{-1})\allowbreak \log(n)\log(\frac{\alpha \log(n)}{\eps}) \log(\delta^{-1}))$ bits of space. Throughout the section we use \textit{$\alpha$-property} to refer to the $L_1$ $\alpha$-property.
	
\subsection{The $L_1$ Sampler}

Our algorithm employs the technique of \textit{precision sampling} in a similar fashion as in the $L_1$ sampler of \cite{Jowhari:2011}. 		
The idea is to scale every item $f_i$ by $1/t_i$ where $t_i \in [0,1]$ is a uniform random variable, and return any index $i$ such that $z_i = |f_i|/t_i >\frac{1}{\eps}\|f\|_1$, since this occurs with probability exactly $\eps \frac{|f_i|}{\|f\|_1}$. One can then run a traditional Countsketch on the scaled stream $z$ to determine when an element passes this threshold.  

In this section, we will adapt this idea to \textit{strong} $\alpha$-property streams (Definition \ref{def:strongalphaprop}). The necessity of the strong $\alpha$-property arises from the fact that if $f$ has the strong $\alpha$-property, then any coordinate-wise scaling $z$ of $f$ still has the $\alpha$-property with probability $1$. Thus the stream $z$ given by $z_i = f_i/t_i$ has the $\alpha$-property (in fact, it again has the strong $\alpha$-property, but we will only need the fact that $z$ has the $\alpha$-property). Our full $L_1$ sampler is given in Figure \ref{fig:L1Samp}. 

By running \texttt{CSSS} to find the heavy hitters of $z$, we introduce error additive in $O(\eps' \|z\|_1) = O(\eps^3 /\log^2(n)\|z\|_1)$, but as we will see the heaviest item in $z$ is an $\Omega(\eps^2/\log^2(n))$ heavy hitter with probability $1-O(\eps)$ conditioned on an arbitrary value of $t_i$, so this error will only be an $O(\eps)$ fraction of the weight of the maximum weight element. Note that we use the term $c$-heavy hitter for $c \in (0,1)$ to denote an item with weight at least $c\|z\|_1$. 
Our algorithm then attempts to return an item $z_i$ which crosses the threshold $\|f\|_1/\eps$, and we will be correct in doing so if the tail error $\err{z}{2}$ from $\texttt{CSSS}$ is not too great.

To determine if this is the case, since we are in the strict turnstile case we can compute $r = \|f\|_1$ and $q = \|z\|_1$ exactly by keeping a $\log(n)$-bit counter (note however that we will only need constant factor approximations for these). Next, using the result of Lemma \ref{lem:ErrEst} we can accurately estimate $\err{z}{2}$, and abort if it is too large in Recovery Step $4$ of Figure \ref{fig:L1Samp}. If the conditions of this step hold, we will be guaranteed that if $i$ is the maximal element, then $y_i^* = (1 \pm O(\eps))z_i$. This allows us to sample $\eps$-approximately, as well as guarantee that our estimate of $z_i$ has relative error $\eps.$
We now begin our analysis our $L_1$ sampler. First, the proof of the following fact can be found in \cite{Jowhari:2011}.
\begin{lemma}
	\label{lem:errbound2}
	Given that the values of $t_i$ are $k = \log(1/\eps)$-wise independent, then conditioned on an arbitrary fixed value $t = t_l \in [0,1]$ for a single $l \in [n]$, we have $\pr{20\err{z}{2} > k^{1/2}\|f\|_1} = O(\eps + n^{-c})$.
\end{lemma}

\begin{figure*}
	\fbox{\parbox{\textwidth}{ \texttt{$\alpha$L1Sampler}: $L_1$ sampling algorithm for strict-turnstile strong $\alpha$-property streams:
			
			\textbf{Initialization}
			\begin{enumerate}[topsep=0pt,itemsep=-1ex,partopsep=1ex,parsep=1ex]
				\item Instantiate \texttt{CSSS} with $k=O(\log(\eps^{-1}))$ columns and parameter $\eps' = \eps^3/\log^2(n)$.
				\item	Select $k=O(\log(\frac{1}{\eps}))$-wise independent uniform scaling factors $t_i \in  [0,1]$ for $i \in [n]$.
				\item  Run  \ttx{CSSS} on scaled input $z$ where $z_i = f_i/t_i$.
				\item Keep $\log(n)$-bit counters $r,q$ to store $r = \|f\|_1$ and $q = \|z\|_1$.

			\end{enumerate}
			\textbf{Recovery}
			\begin{enumerate}[topsep=0pt,itemsep=-1ex,partopsep=1ex,parsep=1ex]
				\item Compute estimate $y^*$ via \ttx{CSSS}.				
				\item Via algorithm of Lemma \ref{lem:ErrEst}, compute $v$ such that $\err{z}{2} \leq v \leq 45\frac{k^{1/2}\eps^3}{\log^2(n)} \|z\|_1+ 20\err{z}{2}$.
				\item Find $i$ with $|y^*_i|$ maximal. 
				\item If $v> k^{1/2}r+45\frac{k^{1/2}\eps^3}{\log^2(n)}q  $, or  $|y^*_i| <\max\{ \frac{1}{\eps}r, \frac{(c/2)\eps^2}{\log^2(n)} q \}$ where the constant $c>0$ is as in Proposition \ref{prop:bighitter}, output FAIL, otherwise output $i$ and $t_i y^*_i$ as the estimate for $f_i$.
			\end{enumerate}				
	}}	\caption{Our $L_1$ sampling algorithm with sucsess probability $\Theta(\eps)$}	\label{fig:L1Samp}
\end{figure*}

The following proposition shows that the $\eps/\ab\log^2(n)$ term in the additive error of our \texttt{CSSS} will be an $\eps$ fraction of the maximal element with high probability. 

\begin{proposition}
	
	\label{prop:bighitter}
 There exists some constant $c > 0$ such that conditioned on an arbitrary fixed value $t = t_l \in [0,1]$ for a single $l \in [n]$, if $j$ is such that $|z_j|$ is maximal, then with probability $1-O(\eps)$ we have $|z_j| \geq c\eps^2/\log^2(n)\|z\|_1$.

\end{proposition}
\begin{proof}
	Let $f_i' = f_i, z_i'=z_i$ for all $i \neq l$, and let $f'_l = 0 = z_l'$. Let $I_t = \{i \in [n]\setminus l \; | \; z_i \in [\frac{\|f'\|_1}{2^{t+1}}, \frac{\|f'\|_1}{2^{t}}) \}$. Then $\pr{i \in I_t, i \neq l} = 2^t|f_i|/\|f'\|_1$, so $\ex{|I_t|} = 2^t$ and by Markov's inequality $\pr{|I_t| > \log(n)\eps^{-1}2^t } < \eps/\log(n)$. By the union bound $\sum_{t \in [\log(n)]} \sum_{i \in I_t}|z_i| \leq \log^2(n)\|f'\|_1/\eps$ with probability $1 - \eps$. Call this event $\mathcal{E}$. 
	Now for $i \neq l$ let $X_i$ indicate $z_i \geq \eps \|f'\|/2$. Then Var$(X_i) = (2/\eps) |f_i|/\|f'\|_1- ((2/\eps) |f_i|/\|f'\|_1)^2 < \ex{X_i}$, and so pairwise independence of the $t_i$ is enough to conclude that Var$(\sum_{i\neq l} X_i) \ab < \ex{\sum_{i\neq l} X_i} = 2/\eps$. So by Chebyshev's inequality, \[\bpr{\big|\sum_{i\neq l} X_i - 2\eps^{-1}\big| > \eps^{-1}} < 2\eps\]
	so there is at least one and at most $3/\eps$ items in $z'$ with weight $\eps\|f'\|_1/2$ with probability $1-O(\eps)$. Call these ``heavy items''. By the union bound, both this and $\mathcal{E}$ occur with probability $1-O(\eps)$.
	Now the largest heavy item will be greater than the average of them, and thus it will be an $\eps/3$-heavy hitter among the heavy items. Moreover, we have shown that the non-heavy items in $z'$ have weight at most $\log^2(n)\|f'\|_1/\eps$ with probability $1-\eps$, so it follows that the maximum item in $z'$ will have weight at least $\eps^2/(2\log^2(n))\|z\|_1$ with probability $1-O(\eps)$.
	
	Now if $z_l$ is less than the heaviest item in $z'$, then that item will still be an $\eps^2/(4\log^2(n))$ heavy hitter. If $z_l$ is greater, then $z_l$ will be an $\eps^2/(4\log^2(n))$ heavy hitter, which completes the proof with $c = 1/4$.
\end{proof}

\begin{lemma}
		\label{lem:L1SampMain}
	The probability that \texttt{$\alpha$L1Sampler} outputs the index $i \in [n]$ is $(\eps \pm O(\eps^2))\frac{|f_i|}{\|f\|_1} + O(n^{-c})$. The relative error of the estimate of $f_i$ is $O(\eps)$ with high probability.
\end{lemma}
\begin{proof}
	Ideally, we would like to output $i \in [n]$ with $|z_i| \geq \eps^{-1}r$, as this happens if $t_i \leq \eps |f_i|/r$, which occurs with probability precisely $ \eps |f_i|/r$. We now examine what could go wrong and cause us to output $i$ when this condition is not met or vice-versa. We condition first on the fact that $v$ satisfies $\err{z}{2} \leq v \leq (45k^{1/2}\eps^3/\log^2(n)) \|z\|_1+ 20\err{z}{2}$ as in Lemma \ref{lem:ErrEst} with parameter $\eps'=\eps^3/\log^2(n)$, and on the fact that $|y^*_j - z_j| \leq   2( k^{-1/2}\err{z}{2}  + (\eps^3/\log^2(n))\|z\|_1 )$ for all $j \in [n]$ as detailed in Theorem \ref{thm:cssserror}, each of which occur with high probability. 
	
	The first type of error is if $z_i \geq \eps^{-1} \|f\|_1$ but the algorithm fails and we do not output $i$. First, condition on the fact that $20 \err{z}{2} < k^{1/2}\|f\|_1$ and that $|z_{j^*}| \geq c\eps^2/\log^2(n) \|z\|_1$ where $j^* \in [n]$ is such that $|z_{j^*}|$ is maximal, which by the union bound using Lemma \ref{lem:errbound2} and Proposition \ref{prop:bighitter} together hold with probability $1 - O(\eps)$ conditioned on any value of $t_i$. Call this conditional event $\mathcal{E}$. Since $v \leq 20\err{z}{2} + (45k^{1/2}\eps^3/\log^2(n)) \|z\|_1$ w.h.p., it follows that conditioned on $\mathcal{E}$ we have $v \leq k^{1/2} \|f\|_1+ (45k^{1/2}\ab\eps^3/\log^2(n)) \|z\|_1 $. So the algorithm does not fail due to the first condition in Recovery Step 4. of Figure \ref{fig:L1Samp}. Since $v \geq \err{z}{2}$ w.h.p, we now have  $\frac{1}{k^{1/2}}\err{z}{2} \leq  \|f\|_1+45\frac{\eps^3}{\log^2(n)}\|z\|_1$, and so $|y^*_j - z_j| \leq   2( \frac{1}{k^{1/2}}\err{z}{2}  + \frac{\eps^3}{\log^2(n)}\|z\|_1 ) \leq 2\|f\|_1+92\frac{\eps^3}{\log^2(n)}\|z\|_1 $ for all $j \in [n]$.

	The second way we output FAIL when we should not have is if $|y_i^*| <( (c/2)\eps^2 / \log^2(n)) \|z\|_1$ but $|z_i| \geq \eps^{-1} \|f\|_1$. Now $\mathcal{E}$ gives us that $|z_{j^*}| \geq c\eps^2/\log^2(n) \|z\|_1$ where $|z_{j^*}|$ is maximal in $z$, and since $|y_i^*|$ was maximal in $y^*$, it follows that $|y_i^* |\geq |y_{j^*}^*| >  |z_{j^*}| - (2\|f\|_1+92\frac{\eps^3}{\log^2(n)}\|z\|_1)$. But $|z_{j^*}|$ is maximal, so $|z_{j^*}| \geq |z_i| \geq \eps^{-1} \|f\|_1$. The two lower bounds on $|z_{j^*}|$ give us $(2\|f\|_1+92\frac{\eps^3}{\log^2(n)}\|z\|_1) = O(\eps)|z_{j^*}| < |z_{j^*}|/2$, so  $|y_i^* | \geq |z_{j^*}|/2 \geq ( (c/2)\eps^2 / \log^2(n)) \|z\|_1$. So conditioned on $\mathcal{E}$, this type of failure can never occur. Thus the probability that we output FAIL for either of the last two reasons when $|z_i| \geq \eps^{-1} \|f\|_1$ is  $O(\eps^2 \frac{|f_i|}{\|f\|_1})$ as needed. 
	So we can now assume that $y_i^* >( (c/2)\eps^2 / \log^2(n)) \|z\|_1$.

	Given this, if an index $i$ was returned we must have $y_i^* > \frac{1}{\eps}r = \frac{1}{\eps}\|f\|$ and $y_i^* > ((c/2)\eps^2/\log^2(n))\|z\|_1$. These two facts together imply that our additive error from \texttt{CSSS} is at most $O(\eps)|y_i^*|$, and thus at most $O(\eps)|z_i|$, so $|y_i^* - z_i| \leq O(\eps)|z_i|$.
	
	With this in mind, another source of error is if we output $i$ with $y_i^* \geq \eps^{-1}r$ but $z_i <  \eps^{-1}r$, so we should not have output $i$. This can only happen if $z_i$ is close to the threshold $r\eps^{-1}$. Since the additive error from our Countsketch is $O(\eps)|z_i|$, it must be that $t_i$ lies in the interval $\frac{|f_i|}{r}(1/\eps + O(1))^{-1} \leq t_i \leq \frac{|f_i|}{r}(1/\eps - O(1))^{-1}$, which occurs with probability $\frac{O(1)}{1/\eps^2 - O(1)} \frac{|f_i|}{r} = O(\eps^2 \frac{|f_i|}{\|f\|_1})$ as needed. 
	
	Finally, an error can occur if we should have output $i$ because $z_i \geq \eps^{-1}\|f\|_1$, but we output another index $i' \neq i$ because $y_{i'}^* > y_i^*$. This can only occur if $t_{i'} < (1/\eps -O(1))^{-1} \frac{|f_{i'}|}{r}$, which occurs with probability $O(\eps \frac{|f_{i'}|}{\|f\|_1})$. By the union bound, the probability that such an $i'$ exists is $O(\eps)$, and by pairwise independence this bound holds conditioned on the fact that $z_i > \eps^{-1}r$. So the probability of this type of error is $O(\eps^2 \frac{|f_i|}{\|f\|_1})$ as needed.

	Altogether, this gives the stated $\eps\frac{|f_i|}{\|f\|_1}(1 \pm O(\eps)) + O(n^{-c})$ probability of outputting $i$, where the $O(n^{-c})$ comes from conditioning on the high probability events. For the $O(\eps)$ relative error estimate, if we return an index $i$ we have shown that our additive error from \texttt{CSSS} was at most $O(\eps)|z_i|$, thus $t_iy_i^* = (1 \pm O(\eps))t_i z_i = (1 \pm O(\eps))f_i$ as needed.
\end{proof}

\begin{theorem}\label{thm:l1samp}
	For $\epsilon,\delta>0$, there is an $O(\eps)$-relative error one-pass $L_1$ sampler for $\alpha$-property streams which also returns an $O(\eps)$-relative error approximation of the returned item. The algorithm outputs FAIL with probability at most $\delta$, and the space is $O(\frac{1}{\eps}\log(\frac{1}{\eps})\log(n)\log(\alpha \; \ab \log\ab(n)/\eps) \log(\frac{1}{\delta}))$.
\end{theorem}
\begin{proof}
	By the last lemma, it follows that the prior algorithm fails with probability at most $1 - \epsilon + O(n^{-c})$. Conditioned on the fact that an index $i$ is output, the probability that $i=j$ is $(1\pm O(\eps))\frac{|f_i|}{\|f\|_1} + O(n^{-c})$. By running $O(1/\eps \log(1/\delta))$ copies of this algorithm in parallel and returning the first index returned by the copies, we obtain an $O(\eps)$ relative error sampler with failure probability at most $\delta$. The $O(\eps)$ relative error estimation of $f_i$ follows from Lemma \ref{lem:L1SampMain}.

For space, note that \texttt{CSSS} requires $O(k \log(n) \log\ab(\ab\alpha\ab \log(n)\ab/\eps)\allowbreak =O(\log(n)\log(1/\eps)\log(\frac{\alpha \log(n)}{\eps}))$ bits of space, which dominates the cost of storing $r,q$ and the cost of computing $v$ via Lemma \ref{lem:ErrEst}, as well as the cost of storing the randomness to compute $k$-wise independent scaling factors $t_i$.
Running $O(1/\eps \log(1/\delta))$ copies  in parallel gives the stated space bound.
\end{proof}

\begin{remark}
	Note that the only action taken by our algorithm which requires more space in the general turnstile case is the $L_1$ estimation step, obtaining $r,q$ in step $2$ of the Recovery in Figure \ref{fig:L1Samp}. Note that $r,q$, need only be constant factor approximations in our proof, and such constant factor approximations can be obtained with high probability using $O(\log^2(n))$ bits (see Fact \ref{fact:L1est}). This gives an $O(\frac{1}{\eps}\log(\frac{1}{\eps})\log(n)\ab \log(\frac{\alpha \log(n)}{\eps})\ab + \log^2(n))$-bit algorithm for the general turnstile case.
\end{remark}

%% file: L1Estimation.tex
		\section{$L_1$ estimation}
		\label{sec:L1Est}

		
		\begin{figure*}
			\fbox{\parbox{\textwidth}{ \texttt{$\alpha L_1$Estimator}: Input $(\epsilon,\delta)$ to estimate $L_1$ of an $\alpha$-property strict turnstile stream.
					\begin{enumerate}[topsep=0pt,itemsep=-1ex,partopsep=1ex,parsep=1ex] 
						\item \textbf{Initialization:} Set $s \leftarrow O(\alpha^2 \delta^{-1}\log^3(n)/\eps^2)$, and initialize \ttx{Morris-Counter} $v_t$ with parameter $\delta'$.  Define $I_j = [s^j, s^{j+2}]$.
						\item \textbf{Processing:} on update $u_t$, for each $j$ such that $v_t \in I_j$, sample $u_t$ with probability $s^{-j}$. 
						\item For each update $u_t$ sampled while $v_t \in I_j$, keep counters $c_j^+,c_j^-$ initialized to $0$. Store all positive updates sampled in $c_j^+$, and (the absolute value of) all negative updates sampled in $c_j^-$.
						\item if $v_t \notin I_j$ for any $j$, delete the counters $c_j^+,c_j^-$.
						
						\item \textbf{Return:} $ s^{-j^*}(c^+_{j^*} - c^-_{j^*})$ for which $j^*$ is such that $c_{j^*}^+,c_{j^*}^-$ have existed the longest (the stored counters which have been receiving updates for the most time steps).
					\end{enumerate}
			}}	\caption{$L_1$ Estimator for strict turnstile $\alpha$-property streams.} \label{figure:L1Est}
		\end{figure*}
		\noindent
		We now consider the well-studied $L_1$ estimation problem in the $\alpha$-property setting (in this section we write $\alpha$-property to refer to the $L_1$ $\alpha$-property). We remark that in the general turnstile unbounded deletion setting, an $O(1)$ estimation of $\|f\|_1$ can be accomplished in $O(\log(n))$ space \cite{kane2010exact}. We show in Section \ref{sec:lowerbounds}, however, that even for $\alpha$ as small as $3/2$, estimating $\|f\|_1$ in general turnstile $\alpha$-property streams still requires $\Omega(\log(n))$-bits of space. 
Nevertheless, in \ref{sec:genturnstilel1} we show that for $\alpha$-property general turnstile streams there is a $\tilde{O}(\eps^{-2}\log(\alpha) + \log(n))$ bits of space algorithm, where $\tilde{O}$ hides $\log(1/\eps)$ and $\log\log(n)$ terms, thereby separating the $\eps^{-2}$ and $\log n$ factors. Furthermore, we show a nearly matching lower bound of $\Omega(\frac{1}{\eps^2}\log(\eps^2\alpha))$ for the problem (Theorem \ref{thm:l1estHardTwo}).

\subsection{Strict Turnstile $L_1$ Estimation}
		 Now for strict-turnstile $\alpha$-property streams, we show that the problem can be solved with $\tilde{O}(\log(\alpha))$-bits. Ideally, to do so we would sample $\poly(\alpha\log(n)/\eps)$ updates uniformly from the stream and apply Lemma \ref{lem:l1-preservation}. To do this without knowing the length of the stream in advance, we sample in exponentially increasing intervals, throwing away a prefix of the stream. At any given time, we will sample at two different rates in two overlapping intervals, and we will return the estimate given by the sample corresponding to the interval from which we have sampled from the longest upon termination.  We first give a looser analysis of the well known Morris counting algorithm.
		 	
		 	\begin{lemma}\label{lem:Morris}
		 	There is an algorithm, \ttx{Morris-Counter}, that given $ \delta \in(0, 1)$ and a sequence of $m$ events, produces non-decreasing estimates $v_t$ of $t$ such that \[ \delta/(12\log(m)) t \leq v_t \leq 1/\delta t\] for a fixed $t \in [m]$ with probability $1-\delta$. The algorithm uses $O(\log\log(m) )$ bits of space.  	 
		\end{lemma}
	\begin{proof}
		The well known Morris Counter algorithm is as follows. We initialize a counter $v_0 = 0$, and on each update $t$ we set $v_t = v_{t-1} + 1$ with probability $1/2^{v_t}$, otherwise $v_t = v_{t-1}$. The estimate of $t$ at time $t$ is $2^{v_t}-1$. It is easily shown that $\ex{2^{v_t} } = t + 1$,
		and thus by Markov's bound, $\pr{2^{v_t} - 1 > t/\delta} < \delta$. 
		
		For the lower bound consider any interval $E_i = [2^i,\ab2^{i+1}]$.
		Now suppose our estimate of $t = 2^i$ is less than $ 6(2^{i} \delta/\ab\log(n))$. Then we expect to sample at least $3\ab\log(n)/\delta$ updates in $E_i$ at $1/2$ the current rate of sampling (note that the sampling rate would decrease below this if we did sample more than $2$ updates). Then by Chernoff bounds, with high probability w.r.t. $n$ and $\delta$ we sample at least two updates. Thus our relative error decreases by a factor of at least $2$ by the time we reach the interval $E_{i+1}$. Union bounding over all $\log(m) = O(\log(n))$ intervals, the estimate never drops below $\delta/12\log(n)t$ for all $t \in [m]$ w.h.p. in $n$ and $\delta$. The counter is $O(\log\log(n))$ bits with the same probability, which gives the stated space bound.
	\end{proof}
		
Our full $L_1$ estimation algorithm is given in Figure \ref{figure:L1Est}.
Note that the value $s^{-j^*}$ can be returned symbolically by storing $s$ and $j^*$, without explicitly computing the entire value. Also observe that we can assume that $s$ is a power of $2$ by rescaling, and sample with probability $s^{-i}$ by flipping $\log(s)i$ fair coins sequentially and sampling only if all are heads, which requires $O(\log\log(n))$ bits of space. 
	\begin{theorem}
	\label{thm:l1est}
		The algorithm \texttt{$\alpha L_1$Estimator} gives a $(1 \pm \eps)$ approximation of the value $\|f\|_1$ of a strict turnstile stream with the $\alpha$-property with probability $1-\delta$ using $O(\log(\alpha/\eps) +\log(1/\delta) +  \log(\log(n))))$ bits of space.
	\end{theorem}
\begin{proof}
	Let $\psi = 12\log^2(m)/\delta$.  By the union bound on Lemma \ref{lem:Morris}, with probability $1-\delta$ the Morris counter $v_t$
	will produce estimates $v_t$ such that $t/\psi \leq v_t \leq \psi t$ for all points $t = s^i/\psi$ and $t= \psi s^i$ for $i=1,\dots,\log(m)/\log(s)$. 
	Conditioned on this, $I_j$ will be initialized by time $\psi s^j$ and not deleted before $s^{j+2}/\psi$ for every 
	$j=1,2,\dots, \ab\log(n)/\log(s)$. $\ab$Thus, upon termination, the oldest set of counters $c_{j^*}^+,c_{j^*}^-$ must have been receiving samples from an interval of size at least $m - 2\psi m/s$, with sampling probability $s^{-j^*}\geq  s\ab/(2\psi m)$. Since $s \geq 2\psi \alpha^2/\eps^2$, it follows by Lemma \ref{lem:l1-preservation} that $ s^{-j^*}(c^+_{j^*} - c^-_{j^*}) = \sum_{i=1}^n \hat{f}_i \pm \eps \|\hat{f}\|_1$ w.h.p., where $\hat{f}$ is the frequency vector of all updates after time $t^*$ and $t^*$ is the time step where $c_{j^*}^+,c_{j^*}^-$ started receiving updates. By correctness of our Morris counter, we know that $t^* < 2\psi m /s < \eps \|f\|_1$, where the last inequality follows from the size of $s$ and the the $\alpha$-property, so the number of updates we missed before initializing $c_{j^*}^+,c_{j^*}^-$ is at most $\eps \|f\|_1$. Since $f$ is a strict turnstile stream, $\sum_{i=1}^n \hat{f}_i = \|f\|_1 \pm t^* = (1 \pm O(\eps)) \|f\|_1$ and $\|\hat{f}\|_1 =  (1\pm O(\eps)) \|f\|_1$. After rescaling of $\eps$ we obtain $ s^{-j^*}(c^+_{j^*} - c^-_{j^*}) = (1 \pm \eps)\|f\|_1$ as needed.
	
	For space, conditioned on the success of the Morris counter, which requires  $O(\log\log(n))$-bits of space, we never sample from an interval $I_j$ for more than $\psi s^{j+2}$ steps, and thus the maximum expected number of samples is $\psi s^2$, and is at most $s^3$ with high probability by Chernoff bounds. Union bounding over all intervals, we never have more than $s^2$ samples in any interval with high probability. At any given time we store counters for at most $2$ intervals, so the space required is $O(\log(s) + \log\log(m)) = O(\log(\alpha/\eps)+ \log(1/\delta) + \log\ab\log(n))$ as stated.  
\end{proof}

	\begin{remark}\label{rem:workingspace}
	Note that if an update $\Delta_t$ to some coordinate $i_t$ arrives with $|\Delta_t| > 1$, 
	our algorithm must implicitly expand $\Delta_t$ to updates in $\{-1,1\}$ by updating the counters by Sign($\Delta_t$)$\cdot$Bin$(|\Delta_t|,s^{-j})$ for some $j$. Note that computing this requires $O(\log(|\Delta_t|))$ bits of working memory, which is potentially larger than $O(\log(\alpha\log(n)/\eps))$. However, if the updates are streamed to the algorithm using $O(\log(|\Delta_t|))$ bits then it is reasonable to allow the algorithm to have at least this much working memory. Once computed, this working memory is no longer needed and does not factor into the space complexity of maintaining the sketch of \texttt{$\alpha L_1$Estimator}.
	\end{remark}

	\subsection{General Turnstile $L_1$ Estimator}\label{sec:genturnstilel1}

In \cite{kane2010exact}, an $O(\epsilon^{-2}\log(n))$-bit algorithm is given for general turnstile $L_1$ estimation. We show how modifications to this algorithm can result in improved algorithms for $\alpha$-property streams. We state their algorithm in Figure \ref{fig:l1genest}, along with the results given in \cite{kane2010exact}. Here $\mathcal{D}_1$ is the distribution of a $1$-stable random variable. In \cite{indyk2006stable, kane2010exact}, the variables $X = \tan(\theta)$ are used, where $\theta$ is drawn uniformly from $[-\frac{\pi}{2},\frac{\pi}{2}]$. We refer the reader to \cite{indyk2006stable} for a further discussion of $p$-stable distributions.

\begin{lemma}[ A.6 \cite{kane2010exact}] \label{lem:genmatrix}
	The entries of $A,A'$ can be generated to precision $\delta = \Theta(\epsilon/m)$ using $O(k\log(n/\eps))$ bits. 
\end{lemma}

\begin{theorem}[Theorem 2.2 \cite{kane2010exact}]\label{thm:kaneexact}
	The algorithm above can be implemented using precision $\delta$ in the variables $A_{i,j},A_{i,j}'$, and thus precision $\delta$ in the entries $y_{i},y_i'$,  such that the output $\tilde{L}$ satisfies $\tilde{L} = (1 \pm \epsilon)\|f\|_1$ with probability $3/4$, where $\delta = \Theta(\epsilon/m)$. In this setting, we have $y_{med}' = \Theta(1)\|f||_1$, and 
	\[\Big|\Big(\frac{1}{r} \sum_{i=1}^r \cos(\frac{y_i}{y_{med}'})\Big) - e^{-(\frac{\|f\|_1}{y_{med}'})}\Big| \leq O(\eps)		\].
\end{theorem}

We demonstrate that this algorithm can be implemented with reduced space complexity for $\alpha$-property streams by sampling to estimate the values $y_i,y_i'$. We first prove an alternative version of our earlier sampling Lemma.

\begin{lemma} \label{lem:samp2}
	Suppose a sequence of $I$ insertions and $D$ deletions are made to a single item, and let $m = I+D$ be the total number of updates. Then if $X$ is the result of sampling updates with probability $p= \Omega(\gamma^{-3} \log(n)/m)$, then with high probability \[X = (I-D) \pm \gamma m\]
\end{lemma}
\begin{proof}
	Let $X^+$ be the positive samples and $X^-$ be the negative samples. First suppose $I > \eps m$, then $\pr{|p^{-1}X^- -I | > \gamma m  } < 2\exp\big(-\frac{\gamma^2 pI}{3}\big) < \exp\big(-\frac{\gamma^3 pm}{3}\big) = 1/\poly(n)$. Next, if $I < \gamma m$, we have $\pr{p^{-1} X^+ > 2\gamma m} < \exp\big(-(\gamma m p/3\big) < 1/\poly(n)$ as needed. A similar bound shows that $X^- = D \pm O(\gamma) m$, thus $X = X^+ - X^- = I - D \pm O(\gamma m)$ as desired after rescaling $\gamma$. 
\end{proof}

\begin{figure*}
	\fbox{\parbox{\textwidth}{ 
			\begin{enumerate}[topsep=0pt,itemsep=-1ex,partopsep=1ex,parsep=1ex] 
				\item \textbf{Initialization:} Generate random matrices  $A \in \R^{r \times n}$ and $A \in \R^{r' \times n}$ of  variables drawn from $\mathcal{D}_1$, where $r = \Theta(1/\epsilon^2)$ and $r' = \Theta(1)$. The variables $A_{ij}$ are $k$-wise independent, for $k = \Theta(\log(1/\epsilon)/\log\log(1/\epsilon))$ , and the variables $A_{ij}'$ are $k'$-wise independent for $k' = \Theta(1)$. For $i \neq i'$, the seeds used to generate the variables $\{A_{i,j}\}_{j=1}^n$ and  $\{A_{i',j}\}_{j=1}^n$ are pairwise independent 
				\item \textbf{Processing:} Maintain vectors $y = Af$ and $y' = A'f$. 
				\item \textbf{Return:} Let $y_{med}' = \text{median} \{|y_i'|\}_{i=1}^{r'}$. Output $\tilde{L} = y_{med}'\big(-\ln\big(\frac{1}{r}\sum_{i=1}^r\cos(\frac{y_i}{y_{med}'})\big)	\big)$
			\end{enumerate}
	}}	\caption{$L_1$ estimator of \cite{kane2010exact} for general turnstile unbounded deletion streams.}\label{fig:l1genest}
\end{figure*}

	\begin{theorem}\label{thm:genl1est}
	There is an algorithm that, given a general turnstile $\alpha$-property stream $f$, produces an estimate $\tilde{L} = (1 \pm O(\eps)) \|f\|_1$ with probability $2/3$ using $O(\eps^{-2} \log(\alpha\ab \log(n)/\eps) +\frac{ \log(\frac{1}{\eps})\log(n)}{\log\log(\frac{1}{\eps})})$ bits of space.
	\end{theorem}
\begin{proof}
	Using Lemma \ref{lem:genmatrix}, we can generate the matrices $A,A'$ using $O(k\log(n)/\epsilon) = O(\log(\frac{1}{\eps})\log(n/\eps)\ab/\log\log(\frac{1}{\eps}))$ bits of space, with precision $\delta = \Theta(\epsilon/m)$.
	Every time an update $(i_t,\Delta_t)$ arrives, we compute the update $\eta_i = \Delta_t A_{i,i_t}$ to $y_i$ for $i\in [r]$, and the update $\eta_{i'}'  =\Delta_t A_{i',i_t}'$ to $y_{i'}'$ for $i' \in [r']$.
	
	Let $X \sim \mathcal{D}_1$ for a $1$-stable distribution $\mathcal{D}_1$. We think of $y_i$ and $y_{i'}'$ as streams on one variable, and we will sample from them and apply Lemma \ref{lem:samp2}. We condition on the success of the algorithm of Theorem \ref{thm:kaneexact}, which occurs with probability $3/4$. Conditioned on this, this estimator $\tilde{L}$ of Figure \ref{fig:l1genest} is a $(1 \pm \eps)$ approximation, so we need only show we can estimate $\tilde{L}$ in small space.

	Now the number of updates to $y_i$ is $\sum_{q=1}^n |A_{iq}|F_q$, where $F_q = I_q + D_q$ is the number of updates to $f_q$. 
	Conditioned on $\max_{j \in [n]} |A_{ij}| = O(n^2)$, which occurs with probability $1- O(1/n)$ by a union bound, $\ex{|A_{iq}|} = \Theta(\log(n))$ (see, e.g., \cite{indyk2006stable}). Then $\ex{\sum_{q=1}^n |A_{iq}|F_q} = O(\log(n))\ab \|F\|_1$ is the expected number of updates to $y_i$, so by Markov's inequality $\sum_{q=1}^n |A_{iq}|F_q = O(\log(n)/\epsilon^2\|F\|_1)$ with probability $1-1/(100(r+r'))$, and so with probability $99/100$, by a union bound this holds for all $i \in [r], i' \in [r']$. We condition on this now.
	
	Our algorithm then is as follows. We then scale $\Delta_i$ up by $\delta^{-1}$ to make each update $\eta_i,\eta_i'$ an integer, and sample updates to each $y_i$ with probability $p= \Omega(\eps_0^{-3} \log(n)/m)$ and store the result in a counter $c_i$. Note that scaling by $\delta^{-1}$ only blows up the size of the stream by a factor of $m/\eps$. Furthermore, if we can $(1 \pm \eps)$ approximation the $L_1$ of this stream, scaling our estimate down by a factor of $\delta$ gives a $(1 \pm \eps)$ approximation of the actual $L_1$, so we can assume from now on that all updates are integers.
	 
Let $\tilde{y}_i = p^{-1}c_i$.	We know by Lemma \ref{lem:samp2} that $\tilde{y}_i = y_i \pm \eps_0(\sum_{q=1}^n |A_{iq}|F_q) = y_i \pm O(\eps_0 \log(n)\|F\|_1 /\eps^2)$ with high probability. Setting $\eps_0 = \eps^3/(\alpha \log(n))$, we have $|\tilde{y}_i - y_i|< \eps/\alpha \|F\|_1 \leq \eps \|f\|_1$ by the $\alpha$-property for all $i \in [r]$.   Note that we can deal with the issue of not knowing the length of the stream by sampling in exponentially increasing intervals $[s^1,s^3], [s^2,s^4],\dots$ of size $s = \poly(\alpha/\eps_0)$ as in Figure \ref{figure:L1Est}, throwing out a $\eps_0$ fraction of the stream. Since our error is an additive $\eps_0$ fraction of the length of the stream already, our error does not change. We run the same routine to obtain estimates $\tilde{y'}_i$ of $y_i'$ with the same error guarantee, and output

	\[	L' = 	\tilde{y}_{med}' \Big(-\ln\big(\frac{1}{r}\sum_{i=1}^r \cos(\frac{\tilde{y_{i}}}{\tilde{y}_{med}'})	\big)	\Big) \]
	where $\tilde{y}_{med}' = \text{median}(\tilde{y_{i}}')$.  By Theorem \ref{thm:kaneexact}, we have that $y_{med}' = \Theta(\|f\|_1)$, thus $\tilde{y}_{med}' = y_{med}' \pm \eps \|f\|_1$ since $|\tilde{y}_i' - y_i'| < \eps \|f\|_1$ for all $i$. Using the fact that $y_{med}' = \Theta(\|f\|_1)$ by Theorem \ref{thm:kaneexact}, we have:
	\[	L' = (y_{med}' \pm \eps \|f\|_1) \Big(-\ln\big(\frac{1}{r}\sum_{i=1}^r \cos(\frac{y_{i}}{y_{med}'(1 \pm O(\eps))} \pm \frac{\eps \|f\|_1}{y_{med}'(1 \pm O(\eps))})	\big)	\Big)	\]
		\[	= (y_{med}' \pm \eps \|f\|_1) \Big(-\ln\big(\frac{1}{r}\sum_{i=1}^r \cos(\frac{y_{i}}{y_{med}'} )\pm O(\eps)	\big)	\Big)	\]
		Where the last equation follows from the angle formula $\cos(\nu+ \beta ) = \cos(\nu)\cos(\beta) - \sin(\nu)\sin(\beta)$ and the Taylor series expansion of $\sin$ and $\cos$. Next, since  $|\big(\frac{1}{r}\sum_{i=1}^r \cos(\frac{y_{i}}{y_{med}'} )\big) - e^{-(\frac{\|f\|_1}{y_{med}'})} | <  O(\eps)$ and $y_{med}' = \Theta(\|f\|_1)$ by Theorem \ref{thm:kaneexact}, it follows that $\big(\frac{1}{r}\sum_{i=1}^r \cos(\frac{y_{i}}{y_{med}'} )\big) = \Theta(1)$, so using the fact that  $\ln(1 \pm O(\eps)) = \Theta(\eps)$, this is 		
			$(y_{med}' \pm \eps \|f\|_1) \Big(-\ln\big(\frac{1}{r}\sum_{i=1}^r \cos(\frac{y_{i}}{y_{med}'} )	\big) \pm O(\eps)	\Big)$, which is $\tilde{L}  \pm O(\eps)\|f\|_1$, where $\tilde{L}$ is the output of Figure \ref{fig:l1genest}, which satisfies  $\tilde{L} = (1 \pm \eps)\|f\|_1$ by Theorem \ref{thm:kaneexact}. It follows that our estimate $L'$ satisfies $L = (1 \pm O(\eps))\|f\|_1$, which is the desired result. Note that we only conditioned on the success of Figure \ref{fig:l1genest}, which occurs with probability $3/4$, and the bound on the number of updates to every $y_i,y_i'$, which occurs with probability $99/100$, and high probability events, by the union bound our result holds with probability $2/3$ as needed.

	For space, generating the entries of $A,A'$ requires $O(\log(\frac{1}{\eps})\ab\log(n/\eps)\ab/\log\log(\frac{1}{\eps}))$ bits as noted, which dominates the cost of storing $\delta$. Moreover, every counter $\tilde{y}_i,\tilde{y}_i'$ is at most $\poly(p\sum_{q=1}^n |A_{iq}|F_q) = \poly(\alpha\log(n)/\eps)$ with high probability, and can thus be stored with $O(\log(\alpha \log(n)/\eps))$ bits each by storing a counter and separately storing $p$ (which is the same for every counter). As there are $O(1/\eps^2)$ counters, the total space is as stated.
\end{proof}

%% file: DistinctElementEstimation.tex
\section{$L_0$ Estimation}
\label{sec:L0Est}
The problem of estimating the support size of a stream is known as $L_0$ estimation. In other words, this is $L_0  = |\{	i \in [n] \; | \; f_i \neq 0	\}|$. $L_0$ estimation is a fundamental problem for network traffic monitoring, query optimization, and database analytics \cite{shukla1996storage, acharya1999aqua, finkelstein1988physical}. The problem also has applications in detecting DDoS attacks \cite{akella2003detecting} and port scans \cite{estan2003bitmap}. 

 For general turnstile streams, Kane, Nelson, and $\ab$Woodruff gave an $O(\epsilon^{-2}\ab\log(n)(\log(\epsilon^{-1}) + \log\log(n)))$-bit algorithm with constant probability of success \cite{kane2010optimal}, which nearly matches the known lower bound of $\Omega(\ab\eps^{-2}\log(\eps^2 n))$ \cite{kane2010exact}. For insertion only streams, they also demonstrated an $O(\eps^{-2} + \log(n))$ upper bound.  In this section we show that the ideas of \cite{kane2010optimal} can be adapted to yield more efficient algorithms for general turnstile $L_0$ $\alpha$-property streams. For the rest of the section, we will simply write \textit{$\alpha$-property} to refer to the $L_0$ $\alpha$-property.

The idea of the algorithm stems from the observation that if $A = \Theta(K)$, then the number of non-empty bins after hashing $A$ balls into $K$ bins is well concentrated around its expectation. Treating this expectation as a function of $A$ and inverting it, one can then recover $A$ with good probability. By treating the (non-zero) elements of the stream as balls, we can hash the universe down into $K = 1/\eps^2$ bins and recover $L_0$ if $L_0 = \Theta(K)$. The primary challenge will be to ensure this last condition. In order to do so, we subsample the elements of the stream at $\log(n)$ levels, and simultaneously run an $O(1)$ estimator $R$ of the $L_0$. To recover a $(1 \pm \eps)$ approximation, we use $R$ to index into the level of subsampling corresponding to a substream with $\Theta(K)$ non-zero elements. We then invert the number of non-empty bins and scale up by a factor to account for the degree of subsampling.
\subsection{Review of Unbounded Deletion Case}
%
\begin{figure*}
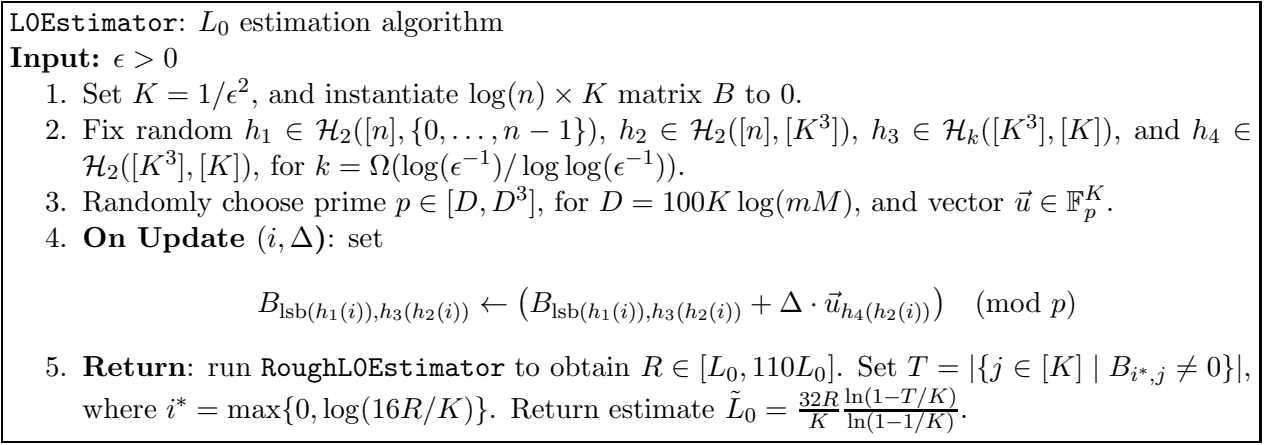
	
	\fbox{\parbox{\textwidth}{\texttt{L0Estimator}: $L_0$ estimation algorithm
			
			\textbf{Input:} $\eps>0$
			\begin{enumerate}[topsep=0pt,itemsep=-1ex,partopsep=1ex,parsep=1ex] 
				\item Set $K = 1/\epsilon^2,$ and instantiate $\log(n) \times K$ matrix $B$ to $0$. 
				\item Fix random $h_1 \in \mathcal{H}_2([n],\{0,\dots,n-1\})$, $h_2 \in \mathcal{H}_2([n],[K^3])$, $h_3 \in \mathcal{H}_k([K^3],[K])$, and $h_4 \in \mathcal{H}_2([K^3],[K])$, for $k = \Omega (\log(\epsilon^{-1})/\log\log(\epsilon^{-1}))$.
				\item Randomly choose prime $p \in [D,D^3]$, for $D = 100K\log(mM)$, and vector $\vec{u} \in \mathbb{F}_p^K$.	
				\item \textbf{On Update $(i,\Delta$)}: set 
				\[B_{\text{lsb}(h_1(i)),h_3(h_2(i))} \leftarrow \big(B_{\text{lsb}(h_1(i)),h_3(h_2(i))} + \Delta \cdot \vec{u}_{h_4(h_2(i))} \big) \pmod p\]
				\item \textbf{Return}: run \ttx{RoughL0Estimator} to obtain $R \in [L_0,110L_0]$. Set $T = |\{j\in [K] \; | \; B_{i^*,j} \neq 0 \}|$, where $i^* = \max\{0, \log(16R/K) \}$. Return estimate $\tilde{L}_0 = \frac{32R}{K} \frac{\ln(1 - T/K)}{\ln(1 - 1/K)}$. 			
	\end{enumerate}}} \caption{$L_0$ Estimation Algorithm of \cite{kane2010optimal}} \label{fig:L0theirs}
\end{figure*}


\label{sec:reviewunbounded}
For sets $U,V$ and integer $k$, let $\mathcal{H}_k(U,V)$ denote some $k$-wise independent hash family of functions mapping $U$ into $V$.  Assuming that $|U|,|V|$ are powers of $2$, such hash functions can be represented using $O(k\log(|U| + |V|))$ bits \cite{carter1979universal} (without loss of generality we assume $n,\eps^{-1}$ are powers of $2$ for the remainder of the section). For $x \in \mathbb{Z}_{\geq 0}$, we write lsb$(x)$ to denote the ($0$-based index of) the least significant bit of $x$ written in binary. For instance, lsb$(6) = 1$ and lsb$(5) = 0$. We set lsb$(0) = \log(n)$. 	
In order to fulfill the algorithmic template outlined above, we need to obtain a constant factor approximation $R$ to $L_0$. This is done using the following result which can be found in \cite{kane2010optimal}.

\begin{lemma}
	\label{lem:RoughL0EST}
	Given a fixed constant $\delta>0$, there is an algorithm, \ttx{RoughL0Estimator}, that with probability $1-\delta$ outputs a value $R = \tilde{L}_0$ satisfying $L_0 \leq R \leq 110 L_0$, using space $O(\log(n)\log\log(n))$. 
\end{lemma}

The main algorithm then subsamples the stream at $\log(n)$ levels. This is accomplished by choosing a hash function $h_1:[n] \to \{0,\dots,n-1\}$, and subsampling an item $i$ at level lsb$(h_1(i))$. Then at each level of subsampling, the updates to the subsampled items are hashed into $K = \frac{1}{\epsilon^2}$ bins $k = \Omega (\log(\epsilon^{-1}/\log(\log(\epsilon^{-1}))))$-wise independently. The entire data structure is then represented by a $\log(n) \times K$ matrix $B$. The matrix $B$ is stored modulo a sufficiently large prime $p$, and the updates to the rows are scaled via a random linear function to reduce the probability that deletions to one item cancel with insertions to another, resulting in false negatives in the number of buckets hit by items from the support. At the termination of the algorithm, we count the number $T$ of non-empty bins in the $i^*$-th row of $B$, where $i^* =\max\{0, \log(\frac{16R}{K})\}$. We then return the value $\tilde{L}_0 = \frac{32R}{K} \frac{\ln(1 - T/K)}{\ln(1 - 1/K)}$. The full algorithm is given in Figure \ref{fig:L0theirs}. First, the following Lemma can be found in \cite{kane2010optimal}.

\begin{lemma}
	\label{lem:cheby}
	There exists a constant $\eps_0$ such that the following holds. Let $A$ balls be mapped into $K = 1/\eps^2$ bins using a random $h \in \mathcal{H}_k([A],[K])$ for $k=c \log(1/\eps)/\allowbreak\log\log(1/\eps)$ for a sufficiently large constant $c$.  Let $X$ be a random variable which counts the number of non-empty bins after all balls are mapped, and assume $100 \leq A \leq K/20$ and $\eps \leq \eps_0$. Then $\ex{X} = K(1 - (1-K^{-1})^A)$
	and
	\[\pr{\big|X - \ex{X}		\big| \leq 8\eps\ex{X} } \geq 4/5		\]
\end{lemma}

Let $A$ be the $\log(n) \times K$ bit matrix such that $A_{i,j} = 1$ iff there is at least one $v \in [n]$ with $f_v \neq 0$ such that lsb$(h_1(v)) =i $ and $h_3(h_2(v)) = j$. In other words, $A_{i,j}$ is an indicator bit which is $1$ if an element from the support of $f$ is hashed to the entry $B_{i,j}$ in the above algorithm. Clearly if $B_{i,j} \neq 0$, then $A_{i,j} \neq 0$. However, the other direction may not always hold. The proofs of the following facts and lemmas can be found in \cite{kane2010optimal}. However we give them here for completeness.

\begin{fact}
	\label{fact2}
	Let $t,r > 0$ be integers. Pick $h \in \mathcal{H}_2([r],[t])$. For any $S \subset [r]$, $\ex{\sum_{i=1}^s \binom{|h^{-1}(i) \cap S|}{2} } \leq |S|^2/(2t)$.
\end{fact}
\begin{proof}
	Let $X_{i,j}$ be an indicator variable which is $1$ if $h(i) = j$. Utilizing linearity of expectation, the desired expectation is then $t \sum_{i < i'} \ex{X_{i,1}} \ex{X_{i',1}} = t \binom{|S|}{2}\frac{1}{t^2}\leq \frac{|S|^2}{2t}$.
\end{proof}

\begin{fact}
	\label{fact3}
	Let $\mathbb{F}_q$ be a finite field and $v \in \mathbb{F}_q^d$ be a non-zero vector. Then, if  $w \in \mathbb{F}_q^d$ is selected randomly, we have $\pr{v \cdot w = 0} = 1/q$ where $v \cdot w$ is the inner product over $\mathbb{F}_q$.
\end{fact}
\begin{proof}
	The set of vectors orthogonal to $v$ is a linear subspace  $V \subset \mathbb{F}_q^d$ of dimension $d-1$, and therefore contains $q^{d-1}$ points. Thus $\pr{w\in V} = 1/q$ as needed.
\end{proof}

\begin{lemma}[Lemma 6 of \cite{kane2010optimal}]
\label{lem:recoverA}
	Assuming that $L_0 \geq \ab K/32$, with probability $3/4$, for all $j \in [K]$ we have $A_{i^*,j} = 0$ if and only if $B_{i^*,j} = 0$. Moreover, the space required to store each $B_{i,j}$ is $O(\log\log(n) + \log(1/\eps))$.
	\end{lemma}
\begin{proof}
	The space follows by the choice of $p \in O(D^3)$, and thus it suffices to bound the probability that $B_{i^*,j} = 0$ when  $A_{i^*,j} \neq 0$. Define $I_{i^*} = \{j \in [n] \; | \; \text{lsb}(h_1(j))=i^*, f_j \neq 0	\}$. This is the set of non-zero coordinates of $f$ which are subsampled to row $i^*$ of $B$. Now conditioned on $R \in [L_0,110L_0]$, which occurs with arbitrarily large constant probability $\delta = \Theta(1)$, we have $\ex{|I_{i^*}|} \leq K/32$, and using the pairwise independence of $h_1$ we have that Var$(|I_{i^*}|)< \ex{|I_{i^*}|}$. So by Chebyshev's inequality $\pr{|I_{i^*}| \leq K/20} = 1-O(1/K)$, which we now condition on. Given this, since the range of $h_2$ has size $K^3$, the indices of $I_{i^*}$ are perfectly hashed by $h_2$ with probability $1- O(1/K) = 1-o(1)$, an event we call $Q$ and condition on occurring.
	
	Since we choose a prime $p \in [D,D^3]$, with $D = 100K\; \ab\log(mM)$, for $mM$ larger than some constant, by standard results on the density of primes there are at least $(K^2\;\ab \log^2(mM))$ primes in the interval $[D,D^3]$. Since each $f_j$ has magnitude at most $mM$ and thus has at most $\log(mM)$ prime factors, it follows that $f_j \neq 0 \pmod p$ with probability $1 - O(1/K^2) = 1-o(1)$. Union bounding over all $j \in I_{i^*}$, it follows that $p$ does not divide $|f_j|$ for any $j \in I_{i^*}$ with probability $1-o(1/K) = 1-o(1)$. Call this event $Q'$ and condition on it occurring. Also let $Q''$ be the event that $h_4(h_2(j)) \neq h_4(h_2(j'))$ for any distinct $j,j' \in I_{i^*}$ such that $h_3(h_2(j)) =  h_3(h_2(j'))$.
	
	To bound $\pr{\neg Q''}$, let $X_{j,j'}$ indicate $h_3(h_2(j)) = h_3(h_2(j'))$, and let $X = \sum_{j < j'} X_{j,j'}$. By Fact \ref{fact2} with $r = K^3,t=K$ and $|S| = |I_{i^*}|<K/20$, we have that $\ex{X} \leq K/800$. Let $Z = \{(j,j') \; | \; h_3(h_2(j)) =h_3(h_2(j'))\}$. For $(j,j') \in Z$, let $Y_{j,j'}$ indicate $h_4(h_2(j)) = h_4(h_2(j'))$, and let $Y = \sum_{(j,j') \in Z} Y_{j,j'}$. By the pairwise independence of $h_4$ and our conditioning on $Q$, we have $\ex{Y} = \sum_{(j,j') \in Z} \pr{h_4(h_2(j)) = h_4(h_2(j'))} = |Z|/K = X/K$. Now conditioned on $X < 20\ex{X} = K/40$ which occurs with probability $19/20$ by Markov's inequality, we have $\ex{Y} \leq 1/40$, so $\pr{Y \geq 1} \leq 1/40$. So we have shown that $Q''$ holds with probability $(19/20)(39/40) \geq 7/8$.

	Now for each $j \in [K]$ such that $A_{i^*,j} = 1$, we can view $B_{i^*,j}$ as the dot product of the vector $v$, which is $f$ restricted to the coordinates of $I_{i^*}$ that hashed to $j$, with a random vector $w \in \mathbb{F}_p$ which is $u$ restricted to coordinates of $I_{i^*}$ that hashed to $j$. Conditioned on $Q'$ it follows that $v$ is non-zero, and conditioned on $Q''$ it follows that $w$ is indeed random. So by Fact \ref{fact3} with $q=p$, union bounding over all $K$ counters $B_{i^*,j}$, we have that $B_{i^*,j} \neq 0$ whenever $A_{i^*,j} \neq 0$ with probability $1-K/p \geq 99/100$. Altogether, the success probability is then $(7/8)(99/100) - o(1) \geq 3/4$ as desired.
\end{proof}

\begin{theorem}
	\label{thm:theirL0est}
	Assuming that $L_0 > K/32$, the value returned by \ttx{L0Estimator} is a $(1\pm \epsilon)$ approximation of the $L_0$ using space $O(\epsilon^{-2}\log(n)(\log(\frac{1}{\epsilon}) + \log(\log(n))\log(\frac{1}{\delta}))$, with $3/4$ success probability.
\end{theorem}
\begin{proof}
	By Lemma \ref{lem:recoverA}, we have shown that $T_{A} = |\{j \in [K] \; | \; A_{i^*,j} \neq 0	\}| = T$ with probability $3/4$, where $T$ is as in Figure \ref{fig:L0theirs}. So it suffices to show that $\tilde{L}_0^A = \frac{32R}{K}\frac{\ln(1-T_A/K)}{\ln(1-1/K)}$ is a $(1 \pm \eps)$ approximation.
	
	Condition on the event $\mathcal{E}$ where $R \in [L_0,110L_0]$, which occurs with large constant probability $\delta = \Theta(1)$.  Let $I_{i^*} = \{j \in [n] \; | \; \text{lsb}(h_1(j))=i^*, f_j \neq 0	\}$. Then  $\ex{|I_{i^*}|} = L_0/2^{i^*+1} = L_0K/(32R)$ (assuming $L_0 > K/\ab32$) and Var$(|I_{i^*}|)< \ex{|I_{i^*}|}$ by the pairwise independence of $h_1$. Then
	$K/3520 \leq \ex{|I_{i^*}|} \leq  K/32$ by $\mathcal{E}$, and by Chebyshev's inequality $K/4224 \leq |I_{i^*}| \leq  K/20$ with probability $1 - O(1/K) = 1-o(1)$. Call this event $\mathcal{E}'$, and condition on it. We then condition on the event $\mathcal{E}''$ that the indices of $I_{i^*}$ are perfectly hashed, meaning they do not collide with each other in any bucket, by $h_2$. Given $\mathcal{E}'$, then by the pairwise independence of $h_2$ the event $\mathcal{E}''$ occurs with probability $1 - O(1/K)$ as well. 
	
	Conditioned on $\mathcal{E}' \wedge\mathcal{E}''$, it follows that $T_A$ is a random variable counting the number of bins hit by at least one ball under a $k$-wise independent hash function, where there are $C=|I_{i^*}|$ balls, $K$ bins, and $k = \Omega(\log(K/\eps)/\log\log(K/\eps))$. Then by Lemma \ref{lem:cheby},  we have $T_A = (1 \pm 8\eps)K(1-(1-1/K)^C)$ with probability $4/5$. So \[\ln(1-T_A/K) = \ln((1-1/K)^C \pm 8\eps (1 - (1-1/K)^C))\] Since we condition on the fact that $K/4244 \leq C \leq K/32$, it follows that $(1-1/K)^C = \Theta(1)$, so the above is  $\ln((1 \pm O(\eps)) (1- 1/K)^C) = C\ln(1-1/K) \pm O(\eps)$, and since $\ln(1+x) = O(|x|)$ for $|x|<1/2$, we have $\tilde{L}_0^A = \frac{32RC}{K} +O(\eps R)$. Now the latter term is $O(\eps L_0)$, since $R  =\Theta(L_0)$, so it suffices to show the concentration of $C$. 
	Now since Var$(C) \leq \ex{C}$ by pairwise independence of $h_1$, so Chebyshev's  inequality gives
	$	\pr{\big|C - L_0K/(32R)		\big| \geq c/\sqrt{K}} < \frac{\ex{C}}{(c^2/K) \ex{C}^2} \leq (\frac{16}{c})^2 $
	and this probability can be made arbitrarily small big increasing $c$, so set $c$ such that the probability is $1/100$. Note that $1/\sqrt{K} = \eps$, so it follows that $C = (1 \pm O(\eps)) L_0K/(32R)$. From this we conclude that $\tilde{L}_0^A = (1 \pm O(\eps))L_0$. By the union bound the events $\mathcal{E} \wedge \mathcal{E}' \wedge \mathcal{E}''$ occur with arbitrarily large constant probability, say $99/100$, and conditioned on this we showed that $\tilde{L}^A_0 = (1\pm\eps)L_0$ with probability $4/5 - 1/100= 79/100$. Finally, $\tilde{L}^A_0 = \tilde{L}_0$ with probability $3/4$, and so together we obtain the desired result with probability $1 - 21/100 - 1/100 - 1/4 = 53/100$. 	
	Running this algorithm $O(1)$ times in parallel and outputting the median gives the desired probability. 
\end{proof}

\subsection{Dealing With Small $L_0$}
\label{subsec:smallL0}
In the prior section it was assumed that $L_0 \geq K/32 = \eps^{-2}/32$. We handle the estimation when this is not the case the same way as \cite{kane2010optimal}. 
We consider two cases. First, if $L_0\leq 100$  we can perfectly hash the elements into $O(1)$ buckets and recovery the $L_0$ exactly with large constant probability by counting the number of nonzero buckets, as each non-zero item will be hashed to its own bucket with good probability (see Lemma \ref{lem8KaneOptimal}). 

Now for  $K/32 > L_0 > 100$, a similar algorithm as in the last section is used, except we use only one row of $B$ and no subsampling. In this case, we set $K' = 2K$, and create a vector $B'$ of length $K'$. 
 We then run the algorithm of the last section, but update $B_j$ instead of $B_{i,j}$ every time $B_{i,j}$ is updated. In other words, $B_{j}'$ is the $j$-th column of $B$ collapsed, so the updates to all items in $[n]$ are hashed into a bucket of $B'$. Let $I = \{i \in [n] \; | \; f_i \neq 0 \}$. Note that the only fact about $i^*$ that the proof of Lemma \ref{lem:recoverA} uses was that $\ex{|I_{i^*}|} < K/32$, and since $I = L_0 < K/32$, this is still the case. Thus by the the same argument given in Lemma \ref{lem:recoverA}, with probability $3/4$ we can recover a bitvector $A$ from $B$ satisfying $A_j = 1$ iff there is some $v \in [n]$ with $f_v \neq 0 $ and $h_3(h_2(v)) = j$. Then if $T_A$ is the number of non-zero bits of $A$, it follows by a similar argument as in Theorem \ref{thm:theirL0est} that $\tilde{L}_0' = \ln(1-T_A/K')/\ln(1-1/K') = (1 \pm \eps)L_0$ for $100 < L_0 < K/32$. So if $\tilde{L}_0' > K'/32 = K/16$, we return the output of the algorithm from the last section, otherwise we return $\tilde{L}_0'$.
The space required to store $B$ is $O(\eps^{-2}(\log\log(n) + \log(1/\eps)))$, giving the following Lemma.

\begin{lemma} \label{lem:smallL0}
	Let $\eps>0$ be given and let $\delta > 0$ be a fixed constant. Then there is a subroutine using $O(\eps^{-2}(\log\ab(\eps^{-1}) + \log\log(n)) + \log(n))$ bits of space which with probability $1-\delta$ either returns a $(1 \pm \eps)$ approximation to $L_0$, or returns LARGE, with the guarantee that $L_0 > \eps^{-2}/16$. 
\end{lemma}

\subsection{The Algorithm for $\alpha$-Property Streams}

\begin{figure*} 
	\fbox{\parbox{\textwidth}{\texttt{$\alpha$L0Estimator}: $L_0$ estimator for $\alpha$-property streams.
			\begin{enumerate}[topsep=0pt,itemsep=-1ex,partopsep=1ex,parsep=1ex] 
				\item Initialize instance of \texttt{L0Estimator}, constructing only the top $2\log(4\alpha/\eps)$ rows of $B$. Let all parameters and hash functions be as in \texttt{L0Estimator}.
				\item Initialize \texttt{$\alpha$StreamRoughL0Est} to obtain a value $\tilde{L_0^t} \in [L_0^t, 8\alpha L_0]$ for all $t \in [m]$, and set $\overline{L_0^t} = \max\{\tilde{L_0^t} , 8\log(n)/\log\log(n)\}$ 
				\item Update the matrix $B$ as in Figure \ref{fig:L0theirs}, but only store the rows with index $i$ such that $i =\log(16\overline{L_0^t}/K) \pm 2\log(4\alpha/\eps)$.
				\item \textbf{Return}: run \texttt{$\alpha$StreamConstL0Est} to obtain $R \in [L_0,100L_0]$, and set $T = |\{j\in [K] \; | \; B_{i^*,j} \neq 0 \}|$, where $i^* = \log(16R/K)$. Return estimate $\tilde{L}_0 = \frac{32R}{K} \frac{\ln(1 - T/K)}{\ln(1 - 1/K)}$.
	\end{enumerate}}}\caption{Our $L_0$ estimation algorithm for $\alpha$-property streams with $L_0 > K/32$} \label{fig:L0Est}
\end{figure*}

We will give a modified version of the algorithm in Figure \ref{fig:L0theirs} for $L_0$ $\alpha$ property streams. Our algorithm is given in Figure \ref{fig:L0Est}. We note first that the return value of the unbounded deletion algorithm only depends on the row $i^* = \log(16R/K)$, and so we need only ensure that this row is stored. Our $L_0$ $\alpha$-property implies that if $L_0^t$ is the $L_0$ value at time $t$, then we must have $L_0^m = L_0 \geq 1/\alpha L^t_0$. So if we can obtain an $O(\alpha)$ approximation $R^t$ to $L_0^t$, then at time $t$ we need only maintain and sample the rows of the matrix with index within $c\log(\alpha/\eps)$ distance of $i^t = \log(16R^t/K)$, for some small constant $c$. 

By doing this, the output of our algorithm will then be the same as the output of \texttt{L0Estimator} when run on the suffix of the stream beginning at the time when we first begin sampling to the row $i^*$. Since we begin sampling to this row when the current $L_0^t$ is less than an $\eps$ fraction of the final $L_0$, it will follow that the $L_0$ of this suffix will be an $\eps$-relative error approximation of the $L_0$ of the entire stream. Thus by the correctness of \ttx{L0Estimator}, the output of our algorithm will be a $(1 \pm \eps)^2$ approximation.	

To obtain an $O(\alpha)$ approximation to $L_0^t$ at all points $t$ in the stream, we employ another algorithm of \cite{kane2010optimal}, which gives an $O(1)$ estimation of the $F_0$ value at all points in the stream, where $F_0 = |\{i \in [n] \; | \; f_i^t \neq 0 \text{ for some } t \in [m]\}|$. So by definition for any time $t \in [m]$ we have $F_0^t \leq F_0 = \|I + D\|_0 \leq \alpha \|f\|_0 = \alpha L_0$ by the $\alpha$-property, and also by definition $F_0^t \geq L_0^t$ at all times $t$. These two facts together imply that $[F_0^t, 8F_0^t] \subseteq [L_0^t, 8\alpha L_0]$.

\begin{lemma}[\cite{kane2010optimal}]
	\label{lem:RoughF0EST}
	There is an algorithm, \texttt{RoughF0Est}, that with probability $1 - \delta$ outputs non decreasing estimates $\tilde{F_0}^t$ such that $\tilde{F_0}^t \in [F_0^t, 8F_0^t]$ for all $t \in [m]$ such that $F_0^t \geq \max\{8 ,\log(n)/ \log\log(n)\}$, where $m$ is the length of the stream. The spacedrequired is $O(\log(n) \log(\frac{1}{\delta}))$-bits.
\end{lemma}

\begin{corollary}
	\label{cor:alphaEst}
	There is an algorithm, \texttt{$\alpha$StreamRoughL0\ab Est}, that  with probability $1 - \delta$ on an $\alpha$-deletion stream outputs non-decreasing estimates $\tilde{L_0}^t$ such that $\tilde{L_0}^t \in [L_0^t, 8\alpha L_0]$ for all $t \in [m]$ such that $F_0^t \geq \max\{8,\log(n)/ \log\log(n)\}$, where $m$ is the length of the stream. The space required is $O(\log(n) \log(\frac{1}{\delta}))$ bits.
\end{corollary}

Note that the approximation is only promised for $t$ such that $F_0^t \geq \max\{8, \log(n)/ \log\log(n)\}$.
To handle this, we give a subroutine which produces the $L_0$ exactly for $F_0 <  8 \log(n)/ \log\log(n)$ using $O(\log(n))$ bits. Our main algorithm will assume that $F_0 > 8 \log(n)/\log\log(n)$, and initialize its estimate of $L_0^0$ to be $\overline{L_0^0} = 8\log(n)/\log\log\ab(n)$, where $\tilde{L_0^t} \in [L_0^t.8\alpha L_0]$ is the estimate produced by \texttt{$\ab\alpha$StreamRoughL0\ab Est}. 
\begin{lemma}
	\label{lem:smallF0}
	Given $c  \geq 1$, there is an algorithm that with probability $49/50$ returns the $L_0$ exactly if $F_0 \leq c$, and returns LARGE if $F_0 > c$. The space required is $O(c\log(c) +c \log\log(n) + \log(n))$ bits
\end{lemma}
\begin{proof}
	The algorithm chooses a random hash function $h \in \mathcal{H}_2([n],[C])$ for some $C = \Theta(c^2)$. Every time an update $(i_t,\Delta_t)$ arrives, the algorithm hashes the identity $i_t$ and keeps a counter, initialized to $0$, for each identity $h(i_t)$ seen. The counter for $h(i_t)$ is incremented by all updates $\Delta_\tau$ such that $h(i_\tau) = h(i_t)$. Furthermore, all counters are stored mod $p$, where $p$ is a random prime picked in the interval $[P,P^3]$ for $P = 100^2c\log(mM)$. Finally, if at any time the algorithm has more than $c$ counters stored, it returns LARGE. Otherwise, the algorithm reports the number of non-zero counters at the end of the stream.
	
	To prove correctness, first note that at most $F_0$ items will ever be seen in the stream by definition. Suppose $F_0 \leq c$. By the pairwise independence of $h$, and scaling $C$ by a sufficiently large constant factor, with probability $99/100$ none of the $F_0 \leq c$ identities will be hashed to the same bucket. Condition on this now. Let $I \subset [n]$ be the set of non-zero indices of $f$. Our algorithm will correctly report $|I|$ if $p$ does not divide $f_i$ for any $i \in I$. Now for $mM$ larger then some constant, by standard results on the density of primes there are at least $100c^2\log^2(mM)$ primes in the interval $[P,P^3]$. Since each $f_i$ has magnitude at most $mM$, and thus at most $\log(mM)$ prime factors, it follows that $p$ does not divide $f_i$ with probability $1-1/(100c^2)$. Since $|I| = L_0 < F_0 \leq c$, union bounding over all $i \in I$, it follows that $p \nmid f_i$ for all $i \in I$ with probability $99/100$. Thus our algorithm succeeds with probability $1 - (1/100 + 1/100) >49/50$. 
	
	If $F_0 > c$ then conditioned on no collisions for the first $c+1$ distinct items seen in the stream, which again occurs with probability $99/100$ for sufficiently large $C = \Theta(c^2)$, the algorithm will necessarily see $c+1$ distinct hashed identities once the $(c+1)$-st item arrives, and correctly return LARGE.
	
	For the space, each hashed identity requires $O(\log(c))$ bits to store, and each counter requires $O(\log(P)) = \log(c\log(n))$ bits to store.  There are at most $c$ pairs of identities and counters, and the hash function $h$ can be stored using $O(\log(n))$ bits, giving the stated bound.	
\end{proof}

Finally, to remove the $\log(n)\log\log(n)$ memory overhead of running the \ttx{RoughL0Estimator} procedure to determine the row $i^*$, we show that the exact same $O(1)$ approximation of the final $L_0$ can be obtained using $O(\log(\alpha\log(n))\log(\ab \log(n)) + \log(n))$ bits of space for $\alpha$-property streams. We defer the proof of the following Lemma to Section \ref{app:roughL0est}

\begin{lemma}
	\label{constL0Est}
	Given a fixed constant $\delta$, there is an algorithm, \texttt{$\alpha$StreamConstL0Est} that with probability $1-\delta$ when run on an $\alpha$ property stream outputs a value $R = \hat{L}_0$ satisfying $L_0 \leq R \leq 100 L_0$, using space $O(\log(\alpha)\log\allowbreak\log(n) + \log(n))$.
\end{lemma}

\begin{theorem}
	\label{thm:ours}
	There is an algorithm that gives a $(1\pm \epsilon)$ approximation of the $L_0$ value of a general turnstile stream with the $\alpha$-property, using space $O(\frac{1}{\epsilon^2}\log(\frac{\alpha}{\epsilon}) \allowbreak (\log\ab(\frac{1}{\epsilon}) +\log\log(n)) +\log(n))$, with $2/3$ success probability.
\end{theorem}

\begin{proof}
	The case of $L_0 < K/32$ can be handled by Lemma \ref{lem:smallL0} with probability $49/50$, thus we can assume $L_0 \geq K/32$. Now if $F_0 < 8\log(n)/\log\log(n)$, we can use Lemma \ref{lem:smallF0} with $c = 8\log(n)/\log\log(n)$ to compute the $L_0$ exactly. Conditioning on the success of Lemma \ref{lem:smallF0}, which occurs with probability $49/50$, we will know whether or not $F_0 < 8\log(n)/\log\log(n)$, and can correctly output the result corresponding to the correct algorithm. So we can assume that $F_0 > 8 \log(n)/\log\ab\log(n)$. Then by the $\alpha$-property, it follows that $L_0 > 8\log(n)/(\alpha\ab \log\log(n))$.  
	
	Let $t^*$ be the time step when our algorithm initialized and began sampling to row $i^*$. Conditioned on the success of \texttt{$\alpha$StreamRoughL0Est} and \texttt{$\ab\alpha$StreamConstL0Est}, which by the union bound together occur with probability $49/50$ for constant $\delta$, we argue that $t^*$ exists
	
	First, we know that $i^* = \log(16R/K) > \log(16L_0/\ab K) > \log(16 \cdot (8\log(n)/(\log\log(n)))/K) - \log(\alpha)$ by the success of \texttt{$\ab\alpha$StreamConstL0Est}. Moreover, at the start of the algorithm we initialize all rows with indices $i = \log(16\cdot (8\log\log(n)/\ab\log(n))/K)  \pm 2\log(4\alpha/\eps)$, so if $i^* < \log(16\cdot (8\log\log(n)/\log(n))\ab/K) + 2\log\ab(4\alpha\ab/\eps)$ then we initialize $i^*$ at the very beginning (time $t^* = 0$).  	
	Next, if $i^* > \log(16\cdot (8\log\log(n)/\log(n))\ab/L) + 2\log\ab(4\alpha\ab/\eps)$, then we initialize $i^*$ at the first time $t$ when $\overline{L_0^t} \geq R(\eps /(4\alpha))^2$. We know by termination that $\tilde{L_0}^m \in [L_0, 8\alpha L_0]$ since $F_0 > 8 \log(n)/\log\ab\log(n)$ and therefore by the end of the stream \texttt{$\alpha$Stream$\ab$RoughL0\ab Est} will give its promised approximation. So our final estimate satisfies $\overline{L_0}^m \geq \tilde{L_0}^m \geq L_0 \geq R/8 > R(\eps /(4\alpha))^2$. Thus $i^*$ will always be initialized at some time $t^*$.	
	
	Now because the estimates $\tilde{L_0^t}$ are non-decreasing and we have $\tilde{L_0}^m \in [L_0, 8\alpha L_0]$, it follows that $\tilde{L_0}^t < 8\alpha L_0$ for all $t$. Then, since $\overline{L_0^t} < \max\{8 \alpha L_0, 8\log(n)/\log\ab\log(n) \} < L_0 (4\alpha/\eps)^2$, it follows that at the termination of our algorithm the row $i^*$ was not deleted, and will therefore be stored at the end of the stream. 
	
	Now, at time $t^*-1$ right before row $i^*$ was initialized we have $L_0^{t^*-1} \leq \overline{L_0^{t^*-1} } <  R (\eps /(4\alpha))^2$, and since $R < 110 L_0$ we have $L_0^{t^*-1}/ L_0 \leq O(\eps^2)$. It follows that the $L_0$ value of the stream suffix starting at the time step $t^*$ is a value $\hat{L_0}$ such that  $\hat{L_0} = (1 \pm O(\eps^2))L_0^m$. Since our algorithm produces the same output as running  \ttx{L0Estimator}  on this suffix, we obtain a $(1 \pm \eps)$ approximation of $\hat{L_0}$ by the proof of Theorem \ref{thm:theirL0est} with probability $3/4$, which in turn is a $(1 \pm \eps)^2$ approximation of the actual $L_0$, so the desired result follows after rescaling $\eps$. Thus the probability of success is $1 - (3/50+  1/4) > 2/3$.
	
	For space, note that we only ever store $O(\log(\alpha/\eps))$ rows of the matrix $B$, each with entries of value at most the prime $p \in O((K\log(n))^3)$, and thus storing all rows of the matrix requires $O(1/\eps^2 \log(\alpha/\eps)(\log(\allowbreak1/\eps) + \log\log(n)))$ bits. The space required to run \ttx{$\alpha$\allowbreak StreamConstL0Est} is an additional additive  $O(\log(\alpha\allowbreak)\log\log(n) + \log(n))$ bits. The cost of storing the hash functions $h_1,h_2,h_3,h_4$ is $O(\log(n) + \log^2(1/\eps))$ which dominates the cost of running \texttt{$\alpha$\allowbreak StreamRoughL0Est}. Along with the cost of storing the matrix, this dominates the space required to run the small $L_0$ algorithm of Lemma \ref{lem:smallL0} and the small $F_0$ algorithm of Lemma \ref{lem:smallF0} with $c$ on the order of $O(\log(n)/\log\log(n))$. Putting these together yields the stated bound.
\end{proof}

\subsection{Our Constant Factor $L_0$ Estimator for $\alpha$-Property Streams}\label{app:roughL0est}
In this section we prove Lemma \ref{constL0Est}. Our algorithm \ttx{$\alpha$StreamConstL0Est} is a modification of the  \ttx{Rough\allowbreak L0Estimator} of \cite{kane2010optimal}, which gives the same approximation for turnstile streams. Their algorithm subsamples the stream at $\log(n)$ levels, and our improvement comes from the observation that for $\alpha$-property streams we need only consider $O(\log(\alpha))$ levels at a time. Both algorithms utilize the following lemma, which states that if the $L_0$ is at most some small constant $c$, then it can be computed exactly using $O(c^2\log\log(mM))$ space. The lemma follows from picking a random prime $p = \Theta(\log(mM)\log\log(mM))$ and pairwise independently hashing the universe into $[\Theta(c^2)]$ buckets. Each bucket is a counter which contains  the sum of frequencies modulo $p$ of updates to the universe which land in that bucket. The $L_0$ estimate of the algorithm is the total number of non-zero counters. The maximum estimate is returned after $O(\log(1/\eta))$ trials. 

\begin{lemma}[Lemma 8, \cite{kane2010optimal}]
	\label{lem8KaneOptimal}
	There is an algorithm which, given the promise that $L_0 \leq c$, outputs $L_0$ exactly with probability at least $1 - \eta$ using $O(c^2 \log\log(n))$ space, in addition to needing to store $O(\log(1/\eta))$ pairwise independent hash functions mapping $[n]$ onto $[c^2]$.
\end{lemma}

We now describe the whole algorithm \ttx{RoughL0Estimator} along with our modifications to it. The algorithm is very similar to the main algorithm of Figure \ref{fig:L0Est}. First, a random hash function $h: [n] \to [n]$ is chosen from a pairwise independent family. For each $0 \leq j \leq \log(n)$, a substream $\mathcal{S}_j$ is created which consists of the indices $i \in [n]$ with $\text{lsb}(h(i)) = j$. For any $t \in [m]$, let $\mathcal{S}^{t\to}_j$ denote the substream of $\mathcal{S}$ restricted to the updates $t,t+1,\dots,m$, and similarly let $L_0^{t \to}$ denote the $L_0$ of the stream suffix $t,t+1,\dots,m$. Let $L_0(\mathcal{S})$ denote the $L_0$ of the substream $\mathcal{S}$. 

We then initialize the algorithm \texttt{Rough$\alpha$StreamL0-Estimator}, which by Corollary \ref{cor:alphaEst} gives non-decreasing estimates $\tilde{L}_0^t \in [L_0^t, 8\alpha L_0]$ at all times $t$ such that $F_0 > 8\log(n)\ab/\log\log(n)$ with probability $99/100$. Let $\overline{L_0^t} = \max\{\tilde{L_0^t},\ab 8\log(n)/\log\log(n) \}$, and let $U_t \subset [\log(n)]$ denote the set of indices $i$ such that $i =  \log(\overline{L_0^t}) \pm 2\log(\alpha/\eps)$, for some constant $\eps$ later specified.

Then at time $t$, for each $\mathcal{S}_j$ with $j \in U_t$, we run an instantiation of Lemma \ref{lem8KaneOptimal} with $c= 132$ and $\eta = 1/16$ on $\mathcal{S}_j$, and all instantiations share the same $O(\log(1/\eta))$ hash functions $h_1,\dots,h_{\Theta(\log(1/\eta))}$. If $j \in U_{t}$ but $j \notin U_{t+1}$, then we throw away all data structures related to $\mathcal{S}_j$ at time $t+1$. Similarly, if $j$ enters $U_t$ at time $t$, we initialize a new instantiation of Lemma \ref{lem8KaneOptimal} for $\mathcal{S}_j$ at time $t$.

To obtain the final $L_0$ estimate for the entire stream, the largest value $j \in U_m$ with $j < 2\overline{L_0^m}$ such that $B_j$ declares $L_0(\mathcal{S}_j) > 8$ is found. Then the $L_0$ estimate is $\hat{L}_0 = 20000/99\cdot 2^j$, and if no such $j$ exists the estimate is $\hat{L}_0 = 50$. Note that the main difference between our algorithm and \ttx{RoughL0Estimator} is that \ttx{RoughL0Estimator} sets $U_t = [\log(n)]$ for all $t \in [m]$, so our proof of Lemma \ref{constL0Est} will follow along the lines of \cite{kane2010optimal}.

\begin{proof}[Proof of Lemma \ref{constL0Est} ]
	The space required to store the hash function $h$ is $O(\log(n))$ and each of the $O(\log(1/\eta)) = O(1)$ hash functions $h_i$ takes $\log(n)$ bits to store. The remaining space to store a single $B^j$ 
	is $O(\log\log(n))$ by Lemma \ref{lem8KaneOptimal}, and thus storing all $B^j$s for $j \in U_t$ at any time $t$ requires at most $O(|U_t| \log\log(n)) = O(\log(\alpha) \log\log(n))$ bits (since $\eps = O(1)$), giving the stated bound.
	
	We now argue correctness. First, for $F_0 \leq 8\log(n)\log\ab\log(n)$, we can run the algorithm of Lemma \ref{lem:smallF0} to produce the $L_0$ exactly using less space than stated above.  So we condition on the success of this algorithm, which occurs with probability $49/50$, and assume $F_0 > 8\log(n)\ab/\log\log(n)$. This gives $L_0 > 8\log(n)/(\alpha \log\log(n))$ by the $\alpha$-property, and it follows that $\overline{L_0^m} \leq 8\alpha L_0$.
	
	Now for any $t \in [m]$, $\ex{L_0(\mathcal{S}_j^{t \to})} = L_0^{t \to}/2^{j+1}$ if $j < \log(n)$, and $\ex{L_0(\mathcal{S}_j^{t \to})} = L_0^{t \to}/n$ if $j=\log(n)$. At the end of the algorithm we have all data structures stored for $B_j$'s with $j \in U_m$. Now let $j \in U_m$ be such that $j < \log(2\overline{L_0^m})$, and observe that $B_j$ will be initialized at time $t_j$ such that $\overline{L_0^{t_j}} > 2\overline{L_0^m}(\eps/\alpha)^2$, which clearly occurs before the algorithm terminates. If $j = \log(8\log(n)/\log\log(n)) \pm 2\log(\alpha/\eps)$, then $j \in U_0$ so $t_j = 0$, and otherwise $t_j$ is such that $L_0^{t_j} \leq \overline{L_0^{t_j}} \leq (\eps/\alpha)^2 2^j \leq (\eps/\alpha)^2 \overline(L_0^m) < 8\eps^2 /\alpha L_0$. So $L_0^{t_j}/L_0 = O(\eps^2)$. This means that when $B_j$ was initialized, the value $L_0^{t_j}$ was at most an $O(\eps^2)$ fraction of the final $L_0$, from which it follows that $L_0^{t_j \to} = (1 \pm \eps)L_0$ after rescaling $\eps$. Thus the expected output of $B_j$ (if it has not been deleted by termination) for $j < \log(2L_0)$ is $\ex{L_0(\mathcal{S}_j^{t_j \to} )} = (1 \pm \eps )L_0/2^{j+1}$.

	Now let $j^*$ be the largest $j$ satisfying $\ex{L_0(\mathcal{S}_j^{t_j\to})} \geq 1$. Then $j^* <  \log(2L_0) < \log(2\overline{L_0^m})$, and observe that  $1 \leq \ex{L_0(\mathcal{S}_{j^*}^{t_{j^*}\to})}  \leq 2(1+\eps)$ (since the expectations decrease geometrically with constant $(1 \pm \eps)/2$). Then for any $ \log(2\overline{L_0^m})> j > j^*$, by Markov's inequality we have $\pr{L_0(\mathcal{S}_j^{t_j \to}) > 8} < (1+\eps)1/(8\cdot 2^{j-j^*-1})$. By the union bound, the probability that any such $j \in (j^*,\log(2\overline{L_0^m}))$ has $L_0(\mathcal{S}_j^{t_j \to}) > 8$ is at most $\frac{(1+\eps)}{8}\sum_{ j = j^*+1}^{\log(2\overline{L_0^m})} 2^{-(j-j^*-1)}\ab \leq (1+\eps)/4$. 
	Now let $j^{**} < j^*$ be the largest $j$ such that $\ex{L_0(\mathcal{S}_j^{t_j\to})} \geq 50$. Since the expectations are geometrically decreasing by a factor of $2$ (up to a factor of $1 \pm \eps$), we have 
	$100(1+\eps) \geq \ex{L_0(\mathcal{S}_{j^{**}}^{t_{j^{**}}\to})} \geq 50$, and by the pairwise independence of $h$ we have $\text{Var}[L_0(\mathcal{S}_{j^{**}}^{t_{j^{**}}\to})] \leq \ex{L_0(\mathcal{S}_{j^{**}}^{t_{j^{**}}\to})}$, so by Chebyshev's inequality we have \[\pr{\big|L_0(\mathcal{S}_{j^{**}}^{t_{j^{**}}\to}) - \ex{L_0(\mathcal{S}_{j^{**}}^{t_{j^{**}}\to})} \big| < 3 \sqrt{\ex{L_0(\mathcal{S}_{j^{**}}^{t_{j^{**}}\to})}}} \] \[> 8/9\]
	Then assuming this holds and setting $\eps = 1/100$, we have \[L_0(\mathcal{S}_{j^{**}}^{t_{j^{**}}\to}) >  50  - 3\sqrt{50} > 28 \] \[ L_0(\mathcal{S}_{j^{**}}^{t_{j^{**}}\to})< 100(1+\eps) + 3\sqrt{100(1+\eps)} < 132\]
	What we have shown is that for every $\log(2\overline{L_0^m}) > j > j^*$, with probability at least $3/4(1-\eps/3)$ we will have $L_0(\mathcal{S}_{j}^{t_{j}\to}) \leq 8$. Since we only consider returning $2^j$ for $j \in U_m$ with $j<\log(2\overline{L_0^m})$, it follows that we will not return $\hat{L}_0 = 2^j$ for any $j > j^*$. In addition, we have shown that with probability $8/9$ we will have $ 28< L_0(\mathcal{S}_{j^{**}}^{t_{j^{**}}\to}) < 132$, and by our choice of $c = 132$ and $\eta = 1/16$, it follows that $B^{j^{**}}$ will output the exact value  $L_0(\mathcal{S}_{j^{**}}^{t_{j^{**}}\to}) > 8$ with probability at least $1 - (1/9 + 1/16) > 13/16$ by Lemma \ref{lem8KaneOptimal}. Hence, noting that $\eps = 1/100$, with probability $1- (3/16 + 1/4(1+\eps)) < 14/25$, we output $2^j$ for some $j^{**} \leq j \leq j^*$ for which $j \in U_m$. Observe that since $U_m$ contains all indices $i = \log(\overline{L_0^m}) \pm 2\log(\alpha/\eps)$, and along with the fact that $L_0< \overline{L_0^m} < 8 \alpha L_0$, it follows that all $j \in [j^{**}, j^*]$ will be in $U_m$ at termination for sufficiently large $\eps \in O(1)$. 
	
	Now since $(1 + 1/100)L_0/2 > 2^{j^*} > (1 - 1/100)L_0/4$, and $(1+1/100)L_0/100 > 2^{j^{**}} >(1-1/100) L_0/200$, it follows that $(99/100)L_0/200 < 2^j < 99L_0/200$, and thus $20000/99\cdot 2^j \in [L_0, 100L_0]$ as desired. If such a $j^{**}$ does not exist then $L_0 < 50$ and $50 \in [L_0,100L_0]$. 	
	Note that because of the $\alpha$ property, unless the stream is empty ($m=0$), then we must have $L_0 \geq 1$, and our our approximation is always within the correct range. Finally, if $F_0 \leq 8\log(n)\log\log(n)$ then with probability $49/50$ Lemma \ref{lem:smallF0} produces the $L_0$ exactly, and for larger $F_0$ we output the result of the algorithm just stated. This brings the overall sucsess probability to $14/25 - 49/50 > 13 / 25$. Running $O(\log(1/\delta)) = O(1)$ copies of the algorithm and returning the median, we can amplify the probability $13/25$ to $1-\delta$.		 
\end{proof}

%% file: L0Sampling.tex
\section{Support Sampling}
\label{sec:suppsamp}

\begin{figure*}
	\fbox{\parbox{\textwidth}{\texttt{$\alpha$-SupportSampler}: support sampling algorithm for $\alpha$ property streams.	\\
			\textbf{Initialization:}
			\begin{enumerate}[topsep=0pt,itemsep=-1ex,partopsep=1ex,parsep=1ex] 
				\item Set $s \leftarrow 205k$, and initialize linear sketch function $J:\mathbb{R}^n \to \mathbb{R}^q$ for $q = O(s)$  via Lemma \ref{lem:invertlin}.
				\item Select random $h \in \mathcal{H}_2([n],[n])$, set $I_j = \{i \in [n] \: | \: h(i) \leq 2^j\}$, and set $\eps = 1/48$.
			\end{enumerate}
			\textbf{Processing:}
			\begin{enumerate}[topsep=0pt,itemsep=-1ex,partopsep=1ex,parsep=1ex] 
				\item  Run \ttx{$\alpha$StreamRoughL0Est} of Corollary \ref{cor:alphaEst} with $\delta = 1/12$ to obtain non-decreasing $R^t \in [L_0^t,8 \alpha L_0]$.
				
				\item Let $B_t = \big\{j \in [\log(n)] \: | \: j = \log(ns/3R^t) \pm 2\log(\alpha/\eps) \text{ or } j \geq  \log\big(ns\log\log(n)/(24\log(n))\big) \big\}$. 
				\item Let $t_j\in[m]$ be the first time $t$ such that $j \in B_t$ (if one exists). Let $t_j'$ be the first time step $t' > t_j$ such that $j \notin B_{t'}$ (if one exists).
				\item For all $j \in B_t$, maintain linear sketch $x_j = J(f^{t_j:t_j'}|_{I_j})$. 
			\end{enumerate}
			\textbf{Recovery:}
			\begin{enumerate}[topsep=0pt,itemsep=-1ex,partopsep=1ex,parsep=1ex]
				\item For $j \in B_m$ at the end of the stream, attempt to invert $x_j$ into $f^{t_j:m}|_{I_j}$ via Lemma \ref{lem:invertlin}. Return all strictly positive coordinates of all successfully returned $f^{t_j:m}|_{I_j}$'s. 	
	\end{enumerate}}}\caption{Our support sampling algorithm for $\alpha$-property streams.}\label{fig:suppsamp}
\end{figure*}

The problem of support sampling asks, given a stream vector $f \in \R^n$ and a parameter $k \geq 1$, return a set $U \subset [n]$ of size at least $\min\{k,\|f\|_0 \}$ such that for every $i \in U$ we have $f_i \neq 0$. Support samplers are needed crucially as subroutines for many dynamic graph streaming algorithms, such as connectivity, bipartitness, minimum spanning trees, min-cut, cut sparsifiers, spanners, and spectral sparsifiers \cite{Ahn:2012}. They have also been applied to solve maximum matching \cite{konrad2015maximum}, as well as hyperedge connectivity \cite{guha2015vertex}. A more comprehensive study of their usefulness in dynamic graph applications can be found in \cite{kapralov2017optimal}. 

For strict turnstile streams, an $\Omega(k\log^2(n/k))$ lower bound is known \cite{kapralov2017optimal}, and for general turnstile streams there is an $O(k\log^2(n))$ algorithm \cite{Jowhari:2011}. In this section we demonstrate that for $L_0$ $\alpha$-property streams in the strict-turnstile case, more efficient support samplers exist. For the rest of the section, we write \textit{$\alpha$-property} to refer to the $L_0$ $\alpha$-property, and we use the notation defined at the beginning of Section \ref{sec:reviewunbounded}.
 
 First consider the following template for the unbounded deletion case (as in \cite{Jowhari:2011}). First, we subsample the set of items $[n]$ at $\log(n)$ levels, where at level $j$, the set $I_j \subseteq [n]$ is subsampled with expected size $|I_j| = 2^j$.  Let $f|_{I_j}$ be the vector $f$ restricted to the coordinates of $I_j$ (and $0$ elsewhere). Then for each $I_j$, the algorithm creates a small sketch $x_j$ of the vector $f|_{I_j}$. If $f|_{I_j}$ is sparse, we can use techniques from sparse recovery to recover $f|_{I_j}$ and report all the non-zero coordinates. We first state the following well known result which we utilize for this recovery.
  \begin{lemma}[\cite{Jowhari:2011}]
  	\label{lem:invertlin}
  Given $1 \leq s \leq n$, there is a linear sketch and a recovery algorithm which, given $f \in \R^n$, constructs a linear sketch $J(f):\R^n \to \R^{q}$ for $q = O(s)$ such that if $f$ is $s$-sparse then the recovery algorithm returns $f$ on input $J(f)$, otherwise it returns \ttx{DENSE} with high probability. The space required is $O(s\log(n))$ bits.
  \end{lemma}

 Next, observe that for $L_0$ $\alpha$-deletion streams, the value $F_0^t$ is at least $L_0^t$ and at most $\alpha L_0$ for every $t \in [m]$. Therefore, if we are given an estimate of $F_0^t$, we show it will suffice to only subsample at $O(\log(\alpha))$-levels at a time.
In order to estimate $F_0^t$ we utilize the estimator \texttt{$\alpha$StreamRoughL0Est} from Corollary \ref{cor:alphaEst} of Section \ref{sec:L0Est}.
For $t' \geq t$, let $f^{t:t'} \in \R^n$ be the frequency vector of the stream of updates $t$ to $t'$. We use the notation given in our full algorithm in Figure \ref{fig:suppsamp}.  Notice that since $R^t$ is non-decreasing, once $j$ is removed from $B_t$ at time $t_j'$ it will never enter again. So at the time of termination, we have $x_j = J(f^{t_j:m}|_{I_j})$ for all $j \in B_m$. 

\begin{theorem}\label{thm:suppsamp}
Given a strict turnstile stream $f$ with the $L_0$ $\alpha$-property and $k\geq 1$, the algorithm \texttt{$\alpha$-Support\ab Sampler} outputs a set $U \subset [n]$ such that $f_i \neq 0 $ for every $i \in U$, and such that with probability $1-\delta$ we have $|U| \geq \min\{k,\|f\|_0\}$ . The space required is $O(k\log(n) \log(\delta^{-1})\ab(\log(\alpha)+\log\log(n)))$ bits.
\end{theorem}
\begin{proof}
	First condition on the success of \ttx{$\alpha$StreamRough\ab L0Est}, which occurs with probability $11/12$.
	Set  $i^* = \min (\ab\lceil (\log(\frac{ns}{3L_0})) \rceil, \log(n))$. We first argue that $t_{i^*}$ exists. Now $x_{i^*}$ would be initialized as soon as $R^t \geq L_0 (\eps/\alpha)^2$, but $R^t \geq L_0^t$, so this holds before termination of the algorithm. Furthermore, for $x_{i^*}$ to have been deleted by the end of the algorithm we would need $R^t > (\alpha/\eps)^2L_0 $, but we know $R^t < 8\alpha L_0$, so this can never happen. Finally, if $L_0 \leq F_0 < 8 \log(n) /\log\log(n)$ and our \ttx{$\alpha$Stream\ab RoughL0\ab Est} fails, then note $i^* \geq \lceil \log(ns\ab\log\log(n)/(24\ab\log(n)))\rceil$, so we store $i^*$ for the entire algorithm.

	Now we have $L_0^{t_{i^*}} \leq  R^{t_{i^*}}  < (\eps/\alpha) L_0  $,
	and thus $L_0^{t_{i^*}}/L_0 < \eps$. Since $f$ is a turnstile stream, it follows that the number of strictly positive coordinates in $f^{t_{i^*}:m}$ is at least $L_0 - L_0^{t_{i^*}}$ and at most  $L_0$. Thus there are $(1 \pm \eps)L_0$ strictly positive coordinates in $f^{t_{i^*}:m}$. By same argument, we have $\|f^{t_{i^*}: m }\|_0 = (1 \pm \eps )L_0$.

	Let $X_i$ indicate the event that $f^{t_{i^*}:m}_i|_{I_{i^*}} \neq 0$, and  $X = \sum_{i} X_i$. Using the pairwise independence of $h$, the $X_i$'s with $f^{t_{i^*}:m}_i \neq 0$ are pairwise independent, so we obtain Var$(X) < \ex{X} = \|f^{t_{i^*}: m }\|_0 \ex{|I_{i^*}|}/n$. First assume $L_0 > s$. Then $ns/(3L_0) \leq \ex{|I_{i^*}|}  < 2ns/(3L_0),$ so  for $\eps < 1/48$, we have $\ex{X} \in [15s/48, 33s/48]$. Then $\sqrt{\ex{X}} < 1/8\ex{X}$, so by
	Chebyshev's inequality, $\pr{ | X - \ex{X} | > 1/4 \ex{X} } < 1/4$, and thus $\|f^{t_{i^*}:m}_i|_{I_{i^*}}\|_0 \leq 15/16s$ with probability $3/4$. In this case, $f^{t_{i^*}:m}_i|_{I_{i^*}}$ is $s$-sparse, so we recover the vector w.h.p.  by \ab Lemma \ref{lem:invertlin}. Now if $L_0 \leq  s$ then $F_0 < \alpha s$, so the index $i' = \log(n)$ will be stored for the entire algorithm. Thus $x_{i'}= J(f)$ and since $f$ is $s$ sparse we recover $f$ exactly w.h.p., and return all non-zero elements. So we can assume that $L_0 > s$.
	
	It suffices now to show that there are at least $k$ strictly positive coordinates in $f^{t_{i^*}:m}_i|_{I_{i^*}}$. Since this number is also $(1 \pm L_0)$, letting $X_i'$  indicate $f^{t_{i^*}:m}_i|_{I_{i^*}} > 0$ and using the same inequalities as in the last paragraph, it follows that there are at least $s/15 > k$ strictly positive coordinates with probability $3/4$. Since the stream is a strict-turnstile stream, every strictly positive coordinate of a suffix of the stream must be in the support of $f$, so we successfully return at least $k$ coordinates from the support with probability at least $1-(1/4+1/4+1/12+1/12) = 1/3$. Running $O(\log(\delta^{-1}))$ copies in parallel and setting $U$ to be the union of all coordinates returned, it follows that with probability $1-\delta$ at least $\min\{k,\|f\|_0 \}$ distinct coordinates will be returned.
	
	For memory, for each of the $O(\log(\delta^{-1}))$ copies we subsample at $O(\ab\log(\alpha) + \log\log(n))$ different levels, and each to a vector of size $O(k)$ (and each coordinate of each vector takes $\log(n)$ bits to store) which gives our desired bound. This dominates the additive $O(\log(n))$ bits needed to run \ttx{$\alpha$StreamRoughL0Est}.
\end{proof}

%% file: Lower_Bounds.tex
\section{Lower Bounds}
\label{sec:lowerbounds}
We now show matching or nearly matching lower bounds for all problems we have considered. 
Our lower bounds all follow via reductions from one-way randomized communication complexity. We consider both the public coin model, where Alice and Bob are given access to infinitely many shared random bits, as well as the private coin model, where they do not have shared randomness. 

We first state the communication complexity problems we will be reducing from.
The first such problem we use is the Augmented-Indexing problem (\textsc{Ind}), which is defined as follows. Alice is given a vector $y \in \{0,1\}^n$, and Bob is given a index $i^* \in [n]$, and the values $y_{i^*+1},\dots, y_n$. Alice then sends a message $M$ to Bob, from which Bob must output the bit $y_{i^*}$ correctly. A correct protocol for \textsc{Ind} is one for which Bob correctly outputs $y_{i^*}$ with probability at least $2/3$. The communication cost of a correct protocol is the maximum size of the message $M$ that the protocol specifies to deliver. This problem has a well known lower bound of $\Omega(n)$ (see \cite{miltersen1995data}, or \cite{kane2010exact}). 

\begin{lemma}[Miltersen et al. \cite{miltersen1995data}]
	\label{lem:augIndex}
	The one-way communication cost of any protocol for $\ab$Augmented Indexing (\textsc{Ind}) in the public coin model that succeeds with probability at least $2/3$ is $\Omega(n)$.
\end{lemma}

We present the second communication complexity result which we will use for our reductions. We define the problem \textsc{equality} as follows.  Alice is given $y \in \{0,1\}^n$ and Bob is given $x \in \{0,1\}^n$, and are required to decide whether $x = y$. This problem has a well known $\Omega(\log(n))$-bit lower bound when shared randomness is not allowed (see e.g., \cite{alon1996space} where it is used).

\begin{lemma}
	The one way communication complexity in the private coin model with $2/3$ success probability of \textsc{Equality} is $\Omega(\log(n))$.
\end{lemma}

We begin with the hardness of the heavy hitters problem in the strict-turnstile setting. Our hardness result holds not just for $\alpha$-property streams, but even for the special case of strong $\alpha$-property streams (Definition \ref{def:strongalphaprop}). The result matches our upper bound for normal $\alpha$-property streams from Theorem \ref{thm:HHstrict} up to $\log\log(n)$ and $\log(\eps^{-1})$ terms.

\begin{theorem}
		\label{HH:Hardness}
	For $p \geq 1$ and $\eps \in (0,1)$, any one-pass $L_p$ heavy hitters algorithm for strong $L_1$ $\alpha$-property streams in the strict turnstile model which returns a set containing all $i \in [n]$ such that $|f_i | \geq \eps \|f\|_p$  and no $i \in [n]$ such that $|f_i| < (\eps/2)\|f\|_p$ with probability at least $2/3$ requires $\Omega(\ab\eps^{-p} \log(n\eps^p) \log(\alpha))$ bits.
\end{theorem}
	\begin{proof}
	Suppose there is such an algorithm which succeeds with probability at least $2/3$, and consider an instance of augmented indexing. Alice receives $y \in \{0,1\}^d$, and Bob gets $i^* \in [d]$ and $y_{j}$ for $j > i^*$. Set $D=6$, and let $X$ be the set of all subsets of $[n]$ with $\lfloor 1/(2\eps)^p\rfloor$ elements, and set $d = \log_6(\alpha/4) \lfloor\log(|X|)\rfloor$. Alice divides $y$ into $r= \log_6(\alpha/4)$ contiguous chunks $y^1,y^2,\dots, \ab y^r$ each containing $\lfloor\log(|X|)\rfloor$ bits. She uses $y^j$ as an index into $X$ to determine a subset $x_j \subset[n]$ with $|x_j| = \lfloor 1/(2\eps)^p\rfloor$. Thinking of $x_j$ as a binary vector in $\R^n$, Alice defines the vector $v \in \R^n$ by.
	\[	v = (\alpha D+1)x_1 + (\alpha D^2+1) x_2 + \dots + (\alpha D^{r}+1)x_{r}	\]
	She then creates a stream and inserts the necessary items so that the current frequency vector of the stream is $v$. She then sends the state of her heavy hitters algorithm to Bob, who wants to know $y_{i^*} \in y^{j}$ for some $j = j(i^*)$. Knowing $y_{i^*+1},\dots,y_{d}$ already, he can compute $u = \alpha D^{j+1}x_{j+1} +\alpha D^{j+2} x_{j+2} + \dots + \alpha D^{r}x_{r}$. Bob then runs the stream which subtracts off $u$ from the current stream, resulting in a final frequency vector of $f = v -u$. He then runs his heavy hitters algorithm to obtain a set $S \subset [n]$ of heavy hitters.
	
	We now argue that if correct, his algorithm must produce $S = x_j$. Note that some items in $[n]$ may belong to multiple sets $x_i$, and would then be inserted as part of multiple $x_i$'s, so we must deal with this possibility. 	
	For $p \geq 1$, the weight $\|f\|_p^p$ is maximized by having $x_1 = x_2 = \dots = x_j$, and thus 
	\[\|f\|_p^p \leq  1/(2\eps)^p (  \ab\sum_{i=1}^j \alpha D^i + 1)^p \] \[\leq \eps^{-p} ( \alpha D^{j+1}/10+1)^p \] \[ \leq\eps^{-p} \alpha^p  D^{jp}\]
	and for every $k \in  x_j$ we have $|f_k|^p \geq (\alpha D^{j} + 1)^p \geq \eps^p \|f\|_p^p$ as desired. Furthermore, $\|f\|_p^p \geq |x_j|\alpha^p D^{jp} = \lfloor1/(2\eps)^p \rfloor \alpha^pD^{jp}$, and the weight of any element $k' \in [n] \setminus x_j$ is at most $(\sum_{i=1}^{j-1}\alpha  D^i + 1)^p \leq (\alpha D^{j}/5 + (j-1))^p< D^{jp}/4^p$ for $\alpha \geq 3$, so $f_{k'}$ will not be a $\eps/2$ heavy hitter. Thus, if correct, Bob's heavy hitter algorithm will return $S = x_j$.		
	So Bob obtains $S = x_j$, and thus recovers $y^j$ which indexed $x_j$, and can compute the relevant bit $y_{i^*} \in y^j$ and return this value successfully. Hence Bob solves \textsc{ind} with probability at least $2/3$. 
	
	Now observe that at the end of the stream each coordinate has frequency at least $1$ and received fewer than $3\alpha^2 $ updates (assuming updates have magnitude $1$). Thus this stream on $[n]$ items has the strong $(3\alpha^2)$-property. 
	Additionally, the frequencies of the stream are always non-negative, so the stream is a strict turnstile stream. It follows by Lemma \ref{lem:augIndex} that any heavy hitters algorithm for strict turnstile strong $\alpha$-property streams requires $\Omega(d) = \Omega( \log(\sqrt{\alpha/3}) \log(|X|)) \ab = \Omega(\eps^{-p} \log(\alpha) \log(n\eps^{p} )) $ bits as needed. 
\end{proof}

Next, we demonstrate the hardness of estimating the $L_1$ norm in the $\alpha$-property setting. First, we show that the problem of $L_1$ estimation in the general turnstile model requires $\Omega(\log(n))$-bits even for $\alpha$ property streams with $\alpha = O(1)$. We also give a lower bound of $\Omega(1/\eps^2\log(\alpha))$ bits for general turnstile $L_1$ estimation for strong $\alpha$-property streams.

\begin{theorem}
\label{thm:l1estHard}
	For any $\alpha \geq 3/2$, any algorithm that produces an estimate $\tilde{L_1} \in (1 \pm 1/16)\|f\|_1$ of a general turnstile stream $f$ with the $L_1$ $\alpha$ property with probability $2/3$ requires $\Omega(\log(n))$ bits of space.
\end{theorem}
\begin{proof}
	Let $\mathcal{G}$ be a family of $t = 2^{\Omega(n/2)} = 2^{n'}$ subsets of $[n/2]$, each of size $n/8$ such that no two sets have more than $n/16$ elements in common. As noted in \cite{alon1996space}, the existence of $\mathcal{G}$ follows from standard results in coding theory, and can be derived via a simple counting argument. We now reduce from \textsc{equality}, where Alice has $y\in \{0,1\}^{n'}$ and Bob has $x \in \{0,1\}^{n'}$. Alice can use $y$ to index into $\mathcal{G}$ to obtain a subset $s_y \subset [n/2]$, and similarly Bob obtains $s_x \subset [n/2]$ via $x$. Let $y' ,x'\in \{0,1\}^n$ be the characteristic vectors of $s_y,s_x$ respectively, padded with $n/2$ $0$'s at the end.
	
	Now Alice creates a stream $f$ on $n$ elements by first inserting $y'$, and then inserting the vector $v$ where $v_i= 1$ for $i > n/2$ and $0$ otherwise. She then sends the state of her streaming algorithm to Bob, who deletes $x'$ from the stream. Now if $x = y$, then $x' = y'$ and $\|f\|_1 = \|y' + v - x'\|_1 = n/2$. On the other hand, if $x \neq y$ then each of $s_x,s_y$ have at least $n/16$ elements in one and not in the other. Thus $\|y' - x'\|_1 \geq n/8$, so $\|f\|_1 \geq 5n/8$. Thus a streaming algorithm that produces $\tilde{L_1} \in (1 \pm 1/16)\|f\|_1$ with probability $2/3$ can distinguish between the two cases, and therefore solve \textsc{equality}, giving an $\Omega(\log(n'))$ lower bound. Since $n' = \Omega(n)$, it follows that such an $L_1$ estimator requires $\Omega(\log(n))$ bits as stated.
	
	Finally, note that at most $3n/4$ unit increments were made to $f$, and $\|f\|_1 \geq n/2$, so $f$ indeed has the $\alpha =3/2$ property.   
\end{proof}

\begin{theorem}
	\label{thm:l1estHardTwo}
Any algorithm that produces an estimate $\tilde{L} \in (1 \pm \eps)\|f\|_1$ with probability $11/12$ of a general turnstile stream $f \in \R^n$ with the strong $L_1$ $\alpha$ property requires $\Omega(\frac{1}{\eps^2}\log(\eps^2\alpha))$ bits of space.
\end{theorem}

To prove Theorem \ref{thm:l1estHardTwo}, we first define the following communication complexity problem.

\begin{definition}
	In the \textsc{Gap-Ham} communication complexity problem, Alice is given $x \in \{0,1\}^n$ and Bob is given $y \in \{0,1\}^n$. Bob is promised that either $\|x - y\|_1 < n/2 - \sqrt{n}$ (NO instance) or that $\|x - y\|_1 > n/2+ \sqrt{n}$ (YES instance), and must decide which instance holds.
\end{definition}

Our proof of Theorem \ref{thm:l1estHardTwo} will use the following reduction. Our proof is similar to the lower bound proof for unbounded deletion streams in \cite{kane2010exact}.

\begin{theorem}[\cite{jayram2008one}, \cite{woodruff2007efficient} Section 4.3]
	\label{thm:indextoHam}
	There is a reduction from \textsc{Ind} to \textsc{Gap-Ham}  such that
	deciding \textsc{Gap-Ham}  with probability at least $11/12$
	implies a solution to \textsc{Ind}  with probability at least
	$2/3$. Furthermore, in this reduction the parameter $n$
	in \textsc{Ind} is within a constant factor of that for the
	reduced \textsc{Gap-Ham} instance.
\end{theorem}

We are now ready to prove Theorem \ref{thm:l1estHardTwo}.

\begin{proof}
	Set $k = \lfloor 1/\eps^2 \rfloor$, and let $t = \lfloor \log(\alpha \eps^2)\rfloor$. The reduction is from \textsc{Ind}. Alice receives $x \in \{0,1\}^{kt}$, and Bob obtains $i^* \in [kt]$ and $x_j$ for $j > i^*$. Alice conceptually breaks her string $x$ up into $t$ contiguous blocks $b_i$ of size $k$. Bob's index $i^*$ lies in some block $b_{j^*}$ for $j^* = j^*(i^*)$, and Bob knows all the bits of the blocks $b_i$ for $i > j^*$. Alice then applies the reduction of Theorem \ref{thm:indextoHam} on each block $b_i$ separately to obtain new vectors $y_i$ of length $ck$ for $i \in [t]$, where $c \geq 1$ is some small constant. Let $\beta = c^2 \eps^{-2} \alpha$. Alice then creates a stream $f$ on $ckt $ items by inserting the update $((i,j),\beta 2^i + 1)$ for all $(i,j)$ such that $(y_i)_j=1$. Here we are using $(i,j) \in [t] \times [ck]$ to index into $[ckt]$. Alice then computes the value $v_i = \|y_i\|_1$ for $i \in [t]$, and sends $\{v_1,\dots,v_t\}$ to Bob, along with the state of the streaming algorithm run on $f$.
	
	Upon receiving this, since Bob knows the bits $y_z$ for $z \geq i^*$, Bob can run the same reductions on the blocks $b_i$ as Alice did to obtain $y_i$ for $i > j^*$. He then can make the deletions $((i,j),- \beta 2^i)$ for all $(i,j)$ such that $i > j^*$ and $(y_i)_{j} = 1$, leaving these coordinate to be $f_{(i,j)} = 1$. Bob then performs the reduction from  \textsc{Ind} to \textsc{Gap-Ham} specifically on the block $b_{j^*}$ to obtain a vector $y(B)$ of length $ck$, such that deciding whether $\|y(B) - y_{j^*}\|_1 > ck/2 + \sqrt{ck}$ or $\|y(B) - y_{j^*}\|_1 < ck/2 - \sqrt{ck}$ with probability $11/12$ will allow Bob to solve the instance of \textsc{Ind} on block $b_{j^*}$ with index $i^*$ in $ b_{j^*}$. Then for each $i$ such that $y(B)_i = 1$, Bob makes the update $((j^*,ji),- \beta 2^{j^*} )$ to the stream $f$. He then runs an $L_1$ approximation algorithm to obtain $\tilde{L} = (1 \pm \eps)\|f\|_1$ with probability $11/12$. Let $A$ be the number of indices such that $y(B)_i > (y_{j^*})_i$. Let $B$ be the number of indices such that $y(B)_i < (y_{j^*})_i$. Let $C$ be the number of indices such that $y(B)_i = 1 = (y_{j^*})_i$, and let $D$ be the number of indices $(i,j)$ with $i > j^*$ such that $(y_i)_j=1$. Then we have
	\[\|f\|_1 =  \beta 2^{j^*} A + ( \beta2^{j^*} + 1)B + C + D+ \sum_{i < j^*} v_i ( \beta2^i + 1)  \]
	Let $Z = C + D + B$, and note that $Z < ckt < \beta$. Let $\eta = \sum_{i < j^*} v_i (\beta 2^i + 1)$. Bob can compute $\eta$ exactly knowing the values $\{v_1,\dots,v_t\}$.  Rearranging terms, we have
	$\|y(B) - y_{j^*}\|_1 = (\|f\|_1 - Z - \eta)/(\beta 2^{j^*}) =  (\|f\|_1 - \eta)/(\beta 2^{j^*}) \pm 1$. Recall that Bob must decide whether $\|y(B) - y_{j^*}\|_1 > ck/2 + \sqrt{ck}$ or $\|y(B) - y_{j^*}\|_1 < ck/2 - \sqrt{ck}$. Thus, in order to solve this instance of \textsc{Gap-Ham} it suffices to obtain an additive $\sqrt{ck}/8$ approximation of $\|f\|_1 /(\beta2^{j^*})$. Now note that  
	\[\|f\|_1 /(\beta2^{j^*}) \leq ckt(\beta2^{j^*})^{-1} + (\beta2^{j^*})^{-1} \sum_{i=1}^{j^*}\beta 2^{j} \cdot ck\]
	\[	\leq1 +  4ck 	\]
	where the $ckt < \beta$ term comes from the fact that every coordinate that is ever inserted will have magnitude at least $1$ at the end of the stream. Taking $\eps' = \sqrt{ck}/(40 ck) =  1/(40\sqrt{ck}) = O(1/\eps)$, it follows that a $(1 \pm \eps')$ approximation of $\|f\|_1$ gives a  $\sqrt{ck}/8$ additive approximation of $\|f\|_1 /(\beta2^{j^*})$ as required. Thus suppose there is an algorithm $A$ that can obtain such a $(1 \pm O(1/\eps^2))$ approximation with probability $11/12$. By the reduction of Theorem \ref{thm:indextoHam} and the hardness of \textsc{Ind} in Lemma \ref{lem:augIndex}, it follows that this protocol just described requires $\Omega(kt)$ bits of space. Since the only information communicated between Alice and Bob other than the state of the streaming algorithm was the set $\{v_1,\dots,v_t\}$, which can be sent with $O(t\log(k)) = o(kt)$ bits, it follows that $A$ must have used $\Omega(kt) = \Omega(\eps^{-2}\log(\alpha \eps^))$ bits of space, which is the desired lower bound.
	
	Now note that every coordinate $i$ that was updated in the stream had final magnitude $|f_i| \geq 1$. Furthermore, no item was inserted more than $\beta 2^t +1 < \alpha^2c^2 + 1$ times, thus the stream has the strong $O(\alpha^2)$ property. We have proven that any algorithm that gives a $(1\pm \eps)$ approximation of $\|f\|_1$ where $f$ is a strong $\alpha$-property stream with probability $11/12$ requires $\Omega(\eps^{-2} \log(\eps^2 \sqrt{\alpha})) \ab= \Omega(\eps^{-2} \log(\eps^2 \alpha))$ bits, which completes the proof.
\end{proof}

We now give a matching lower bound for $L_1$ estimation of strict-turnstile strong $\alpha$-property streams. This exactly matches our upper bound of Theorem \ref{thm:l1est}, which is for the more general $\alpha$-property setting.

\begin{theorem}
	\label{thm:L1esthardstrict}
	For $\eps \in (0,1/2)$ and $\alpha < n$, any algorithm which gives an $\eps$-relative error approximation of the $L_1$ of a strong $L_1$ $\alpha$ property stream in the strict turnstile setting with probability at least $2/3$ must use $\Omega(\log(\alpha) + \log(1/\eps) + \log\log(n))$ bits.
\end{theorem}
	\begin{proof}
	The reduction is from \textsc{ind}. Alice, given $x \in \{0,1\}^{t}$ where $t = \log_{10}(\alpha/4)$, constructs a stream $u \in \R^{t}$ such that $u_i = \alpha 10^ix_i+1$. She then sends the state of the stream $u$ to Bob who, given $j \in [n]$ and $x_{j+1},\dots,x_{t}$, subtracts off $v$ where $v_i = \alpha 2^ix_i$ for $i \geq j+1$, and $0$ otherwise. He then runs the $L_1$ estimation algorithm on $u - v$, and obtains the value $L$ such that $L  = (1\pm \eps)L_1$ with probability $2/3$. We argue that if $L  = (1\pm \eps)L_1$  (for $\eps <1/2)$ then $L > (1-\eps) (\alpha10^j )> \alpha10^j/2$ iff $x_j = 1$. If $x_j=1$ then $(u_j - v_j) = \alpha 10^j +1$, if $L > (1-\eps)L_1$ the result follows. If $x_j = 0$, then the total $L_1$ is at most $\alpha/4 + \alpha\sum_{i=1}^{j-1}10^i < \alpha \frac{10^i}{9} + \alpha/4 < \alpha 10^i/3$, so $L < (1 + \eps)L_1 < \alpha 10^j / 2$ as needed to solve \textsc{ind}. Note that each coordinate has frequency at least $1$ at the end of the stream, and no coordinate received more than $\alpha^2$ updates. Thus the stream has the strong $\alpha^2$-property.  By Lemma \ref{lem:augIndex}, it follows that any one pass algorithm for constant factor $L_1$ estimation of a strict turnstile strong $\alpha$-property stream requires $\Omega(\log(\sqrt{\alpha})) = \Omega(\log(\alpha))$ bits.
	
	Finally, note that in the restricted insertion only case (i.e. $\alpha = 1$), estimating the $L_1$ norm means estimating the value $m$ given only the promise that $m \leq \mathbb{M} =  \poly(n)$. There are $\log(\mathbb{M})/\eps$ powers of $(1+\eps)$ that could potentially be a $(1 \pm \eps)$ approximation of $m$, so to represent the solution requires requires $\log(\log(\mathbb{M})/\eps) = O(\log\log(n) + \log(1/\eps))$ bits of space, which gives the rest of the stated lower bound. 
\end{proof}

We now prove a lower bound on $L_0$ estimation. Our lower bound matches our upper bound of Theorem \ref{thm:ours} up to $\log\log(n)$ and $\log(1/\eps)$ multiplicative factors, and a $\log(n)$ additive term. To do so, we use the following Theorem of \cite{kane2010exact}, which uses a one way two party communication complexity lower bound.

\begin{theorem}[Theorem A.2. \cite{kane2010exact}] \label{Thm:L0hardnessKane}
	Any one pass algorithm that gives a $(1 \pm \eps)$ multiplicative approximation of the $L_0$ of a strict turnstile stream with probability at least $11/12$ requires $\Omega(\eps^{-2}\log(\eps^2n))$ bits of space.
\end{theorem}
Setting $n=O(\alpha)$ will give us:
\begin{theorem}
	\label{thm:l0esthard}
	For $\eps \in (0,1/2)$ and $\alpha < n/\log(n)$, any one pass algorithm that gives a $(1 \pm \eps)$ multiplicative approximation of the $L_0$ of a $L_0$ $\alpha$-property stream in the strict turnstile setting with probability at least $11/12$ must use $\Omega(\eps^{-2}\log(\eps^{2}\alpha) + \log\log(n))$ bits of space.
\end{theorem}
\begin{proof}
	We can first construct the same stream used in the communication complexity lower bound of Theorem \ref{Thm:L0hardnessKane} on $n = \alpha - 1$ elements, and then allow Bob to insert a final dummy element at the end of the stream with frequency $1$. The stream now has $\alpha$ elements, and the $L_0$ of the resulting stream, call it $R$, is exactly $1$ larger than the initial $L_0$ (which we will now refer to as $L_0$). Moreover, this stream has the $L_0$ $\alpha$ property since the final frequency vector is non-zero and there $\alpha$ items in the stream. If we then obtained an estimate $\tilde{R} = (1 \pm \eps)R =(1 \pm \eps)(L_0 + 1) $, then the original $L_0 = 0$ if $\tilde{R}< (1 + \eps)$. Otherwise $\tilde{R} - 1 = (1 \pm O(\eps)L_0$, and a constant factor rescaling of $\eps$ gives the desired approximation of the initial $L_0$. By Theorem  \ref{Thm:L0hardnessKane}, such an approximation requires $\Omega(\eps^{-2}\log(\eps^2 \alpha))$ bits of space, as needed. The $\Omega(\log\log(n))$ lower bound follows from the proof of Lemma \ref{thm:L1esthardstrict}, replacing the upper bound $\mathbb{M} \geq m$ with $n$.
\end{proof}

Next, we give lower bounds for $L_1$ and support sampling. Our lower bound for $L_1$ samplers holds in the more restricted strong $\alpha$-property setting, and for such streams we show that even those which return an index from a distribution with variation distance at most $1/6$ from the $L_1$ distribution $|f_i|/\|f\|_1$ requires $\Omega(\log(n)\log(\alpha))$ bits. In this setting, taking $\eps=o(1)$, this bound matches our upper bound from Theorem \ref{thm:l1samp} up to $\log\log(n)$ terms. For $\alpha = o(n)$, our support sampling lower bound matches our upper bound in Theorem \ref{thm:suppsamp}. 

\begin{theorem} \label{thm:l1samphard}
	Any one pass $L_1$ sampler of a strong $L_1$ $\alpha$-property stream $f$ in the strict turnstile model with an output distribution that has variation distance at most $1/6$ from the $L_1$ distribution $|f_i|/\|f\|_1$ and succeeds with probability $2/3$ requires $\Omega(\log(n)\log(\alpha))$ bits of space.
\end{theorem}
\begin{proof}
	Consider the same strict turnstile strong $O(\alpha^2)$-property stream constructed by Alice and Bob in Theorem \ref{HH:Hardness}, with $\eps = 1/2$. Then $X = [n]$ is the set of all subsets of $[n]$ with $1$ item. If $i^* \in [n]$ is Bob's index, then let $j = j(i^*)$ be the block $y^j$ such that $y_{i^*} \in y^j$. The block $y^j$ has $\log(|X|) = \log(n)$ bits, and Bob's goal is to determine the set $x_j \subset [n]$ of exactly one item which is indexed by $y^j$.  Then for the one sole item $k \in x_{j}$ will be a $1/2$-heavy hitter, and no other item will be a $1/4$-heavy hitter, so Bob can run $O(1)$ parallel $L_1$ samplers and find the item $k'$ that is returned the most number times by his samplers. If his sampler functions as specified, having at most $1/6$ variational distance from the $L_1$ distribution, then $k' = k$ with large constant probability and Bob can recovery the bit representation $x_j$ of $k$, from which he recovers the bit $y_{i^*} \in y_j$ as needed. Since Alice's string had length $\Omega(\log(\alpha)\log(|X|)) = \Omega(\log(\alpha)\log(n))$, we obtain the stated lower bound.
\end{proof}

\begin{theorem}\label{thm:suppsamphard}
	Any one pass support sampler that outputs an arbitrary $i \in [n]$ such that $f_i \neq 0$, of an $L_0$ $\alpha$-property stream with failure probability at most $1/3$, requires $\Omega(\log(n/\alpha)\log(\alpha))$ bits of space.
\end{theorem}
\begin{proof}
	The reduction is again from \textsc{Ind}. Alice receives $y \in \{0,1\}^{d}$, for $d= \lfloor \log(n/\alpha) \log(\alpha/4)\rfloor$, and breaks it into blocks $y^1,\dots,y^{\log(\alpha/4)}$, each of size $\lfloor \log(n\ab/\alpha)\rfloor$. She then initializes a stream vector $f \in \R^n$, and breaks $f$ into $\lfloor 4n/\alpha\rfloor $ blocks of size $\alpha/4$, say $B_1,\dots,\ab B_{\lfloor 4n/\alpha\rfloor}$. She uses $y_i$ as an index into a block $B_{j}$ for $j = j(i)$, and then inserts $2^i$ distinct items into block $B_j$, each exactly once, and sends the state of her algorithm over to Bob. Bob wants to determine $y_{i^*}$ for his fixed $i^* \in [n]$. 
	Let $k$ be such that $y_{i^*} \in y^j$ and $B_{k}$ be the block indexed by $y^j$. He knows $y^{j+1}, \dots, y^{\log(\alpha/4)}$, and can delete the corresponding items from $f$ that were inserted by Alice. At the end, the block $B_{k}$ has $2^{j}$ items in it, and the total number of distinct items in the stream is less than $2^{j+1}$. Moreover no other block has more than $2^{j-1}$ items. 
	
	Now suppose Bob had access to an algorithm that would produce a uniformly random non-zero item at the end of the stream, and that would report FAIL with probability at most $1/3$. He could then run $O(1)$ such algorithms, and pick the block $B_{k'}$ such that more than $4/10$ of the returned indices are in $B_{k'}$. If his algorithms are correct, we then must have $B_{k'} = B_{k}$ with large constant probability, from which he can recover $y^j$ and his bit $y_{j^*}$, thus solving \textsc{Ind} and giving the $\Omega(d) = \Omega(\log(n/\alpha) \log(\alpha))$ lower bound by Lemma \ref{lem:augIndex}. 
	
	We now show how such a random index from the support of $f$ can be obtained using only a support sampler. Alice and Bob can used public randomness to agree on a uniformly random permutation $\pi:[n] \to [n]$, which gives a random relabeling of the items in the stream. Then, Alice creates the same stream as before, but using the relabeling and inserting the items in a randomized order into the stream, using separate randomness from that used by the streaming algorithm. In other words, instead of inserting $i \in [n]$ if $i$ was inserted before, Alice inserts $\pi(i)$ at a random position in the stream. Bob the receives the state of the streaming algorithm, and then similarly deletes the items he would have before, but under the relabeling $\pi$ and in a randomized order instead. 
	
	Let $i_1,\dots,i_r \in [n]$ be the items inserted by Alice that were not deleted by Bob, ordered by the order in which they were inserted into the stream. If Bob were then to run a support sampler on this stream, he would obtain an arbitrary $i = g(i_1,\dots,i_r) \in \{i_1,\dots,i_r\}$, where $g$ is a (possibly randomized) function of the ordering  of the sequence $i_1, \dots, i_r$. The randomness used by the streaming algorithm is separate from the randomness which generated the relabeling $\pi$ and the randomness which determined the ordering of the items inserted and deleted from the stream. Thus, even conditioned on the randomness of the streaming algorithm, any ordering and labeling of the surviving items $i_1,\dots,i_r$ is equally likely. In particular, $i$ is equally likely to be any of  $i_1,\dots,i_r$. It follows that $\pi^{-1}(i)$ is a uniformly random element of the support of $f$, which is precisely what Bob needed to solve \textsc{Ind}, completing the proof.
\end{proof}

Finally, we show that estimating inner products even for \textit{strong} $\alpha$-property streams requires $\Omega(\eps^{-1}\log(\alpha))$ bits of space. Setting $\alpha = n$, we obtain an $\Omega(\eps^{-1}\log(n))$ lower bound for unbounded deletion streams, which our upper bound beats for small $\alpha$.

 \begin{theorem}\label{thm:InnerProdHard}
Any one pass algorithm that runs on two strong $L_1$ $\alpha$-property streams $f,g$ in the strict turnstile setting and computes a value $\texttt{IP}(f,g)$ such that $\texttt{IP}(f,g) = \langle f,g \rangle + \eps \|f\|_1\|g\|_1$ with probability $2/3$ requires $\Omega(\eps^{-1}\log(\alpha))$ bits of space.
\end{theorem}
\begin{proof}
	The reduction is from \texttt{IND}. Alice has $y \in \{0,1\}^d$ where $y$ is broken up into $\log_{10}(\alpha)/4$ blocks of size $1/(8\eps)$, where the indices of the $i$-th block are called $B_i$. Bob wants to learn $y_{i^*}$ and is given $y_j$ for $j \geq i^*$, and let $j^*$ be such that $i^* \in B_{j^*}$. Alice creates a stream $f$ on $d$ items, and if $i \in B_j$, then Alice inserts the items to make $f_i = b_i 10^j + 1$, where $b_i = \alpha$ if $y_i=0$, and $b_i = 2\alpha$ otherwise. She creates a second stream $g = \vec{0} \in \R^d$, and sends the state of her streaming algorithm over to Bob. For every $y_i \in B_{j} = B_{j(i)}$ that Bob knows, he subtracts off $b_i 10^j$, leaving $f_i = 1$. He then sets $g_{i^*} = 1$, and obtains \texttt{IP}$(f,g)$ via his estimator. We argue that with probability $2/3$, \texttt{IP}$(f,g) \geq 3\alpha 10^{j^*}/2$ iff $y_{i^*} = 1$. Note that the error is always at most
	\[\eps \|f\|_1\|g\|_1 = \eps \|f\|_1 \]
	\[\leq \eps \big(d+ (8\eps)^{-1}( 2 \alpha 10^{j^*}) +   \sum_{j < j^*} \sum_{i \in B_j} (2 \alpha 10^j ) \big) \]
	\[	 \eps \big(1/(8\eps)\log(\alpha)/4+ (4\eps)^{-1}\alpha 10^{j^*} + \eps^{-1} \alpha 10^{j^*}/32 \big)	\]
	\[	< \alpha 10^{j*}/3 \]
	Now since $\langle f,g \rangle =(y_{i^*} + 1) \alpha 10^{j^*} + 1$, if $y_{i^*} = 0$ then if the inner product algorithm succeeds with probability $2/3$ we must have $\texttt{IP}(f,g) \leq \alpha 10^j + 1 + \alpha10^j/3 < 3\alpha 10^{j^*}/2$, and similarly if $y_{i^*} = 1$ we have $\texttt{IP}(f,g) > 2\alpha 10^j + 1 - \alpha10^j/3  > 3\alpha 10^j /2$ as needed. So Bob can recover $y_{i^*}$ with probability $2/3$ and solve \texttt{IND}, giving an $\Omega(d )$ lower bound via Lemma \ref{lem:augIndex}. Since each item in $f$ received at most $5(\alpha^2)$ updates and had final frequency $1$, this stream has the strong $5\alpha^2$-property, and $g$ was insertion only. Thus obtaining such an estimate of the inner product between strong $\alpha$-property streams requires $\Omega(\eps^{-1}\log(\sqrt{\alpha})) = \Omega(\eps^{-1}\log(\alpha)$ bits, as stated.
\end{proof}

%% file: Conclusion.tex
\section{Conclusion}
\label{sec:conclusion}
We have shown that for bounded deletion streams, many important $L_0$ and $L_1$ streaming problems can be solved more efficiently. For $L_1$, the fact the $f_i$'s are approximately preserved under sampling $\poly(\alpha)$ updates paves the way for our results, whereas for $L_0$ we utilizes the fact that $O(\log(\alpha))$-levels of sub-sampling are sufficient for several algorithms. 
Interestingly, it is unclear whether improved algorithms for $\alpha$-property streams exist for $L_2$ problems, such as $L_2$ heavy hitters or $L_2$ estimation. The difficulty stems from the fact that $\|f\|_2$ is not preserved in any form under sampling, and thus $L_2$ guarantees seem to require different techniques than those used in this paper. 

However, we note that by utilizing the heavy hitters algorithm of \cite{braverman2016bptree}, one can solve the general turnstile $L_2$ heavy hitters problem for $\alpha$-property streams in $O(\allowbreak \alpha^2\allowbreak \log(n)\log(\alpha))$ space. A proof is sketched in Appendix \ref{app:L2sketch}. Clearly a polynomial dependence on $\alpha$ is not desirable; however, for the applications in which $\alpha$ is a constant this still represents a significant improvement over the $\Omega(\log^2(n))$ lower bound for turnstile streams. We leave it as an open question whether the optimal dependence on $\alpha$ can be made logarithmic. 

Additionally, it is possible that problems in dynamic geometric data streams (see \cite{indyk2004algorithms}) would benefit by bounding the number of deletions on the streams in question. We note that several of these geometric problems directly reduce to problems in the data stream model studied here, for which our results directly apply.

Finally, it would be interesting to see if analogous models with lower bounds on the ``signal size'' of the input would result in improved algorithms for other classes of sketching algorithms. In particular, determining the appropriate analogue of the $\alpha$-property for linear algebra problems, such as row-rank approximation and regression, could be a fruitful direction for further research.

%% file: Appendix.tex
\appendix

\section{Sketch of $L_2$ heavy hitters algorithm}
\label{app:L2sketch}

We first note that the BPTree algorithm for \cite{braverman2016bptree} solves the $L_2$ heavy hitters problem in space $O(\eps^{-2}\log(n)\ab\log(1\ab/\eps))$ with probability $2/3$. Recall that the $L_2$ version of the problem asks to return a set $S \subset [n]$ with all $i$ such that $|f_i| \geq \eps \|f\|_2$ and no $j$ such that $|f_j| < (\eps/2)\|f\|_2$. Also recall that the $\alpha$ property states that $\|I+D\|_2 \leq \alpha \|f\|_2$, where $I$ is the vector of the stream restricted to the positive updates and $D$ is the entry-wise absolute value of the vector of the stream restricted to the negative updates. It follows that if $i \in [n]$ is an $\eps$ heavy hitter then $\|I+D\|_2 \leq \alpha \|f\|_2 \leq \alpha / \eps |f_i| \leq \alpha / \eps |I_i + D_i|$. So in the insertion-only stream $I+D$ where every update is positive, the item $i$ must be an $\eps/\alpha$ heavy hitter.

Using this observation, we can solve the problem as follows. First run BPTree with $\eps' = \eps/\alpha$ to obtain a set $S \subset n$ with all $i$ such that $|I_i + D_i| \geq (\eps /\alpha )\|I+D\|_2$  and no $j$ such that $|I_j + D_j| < (\eps /2\alpha) \|I+D\|_2$. It follows that $|S| = O(\alpha^2/\eps^2)$. So in parallel we run an instance of Countsketch on the original stream $f$ with $O(1/\eps^2)$ rows and $O(\log(\alpha/\eps))$ columns. For a fixed $i$, this gives an estimate of $|f_i|$ with additive error $(\eps/4)\|f\|_2$ with probability $1-O(\eps/\alpha)$. At the end, we then query the Countsketch for every $i \in S$, and return only those which have estimated weight $(3\eps/4)\|f\|_2$. By union bounding over all $|S|$ queries, with constant probability the Countsketch will correctly identify all $i \in |S|$ that are $\eps$ heavy hitters, and will throw out all items in $S$ with weight less than $(\eps/2) \|f\|_2$. Since all $\eps$ heavy hitters are in $S$, this solves the problem. The space to run BPTree is $O(\alpha^2/\eps^2 \log(n) \log(\alpha/\eps)$, which dominates the cost of the Countsketch.